\documentclass[12pt]{article}
\usepackage[english]{babel}
\usepackage{amsthm}
\usepackage{amssymb}
\usepackage{amsmath}
\usepackage{float}
\usepackage{caption}
\usepackage{natbib}
\usepackage{bm}
\usepackage{multirow}
\usepackage{xcolor}
\usepackage{pdfpages}
\usepackage{titling}

\theoremstyle{plain}
\newtheorem{proposition}{Proposition}
\newtheorem{lemma}{Lemma}
\newtheorem{theorem}{Theorem}
\newtheorem{corollary}{Corollary}
\newtheorem{remark}{Remark}

\usepackage[bf,SL,BF]{subfigure}
\usepackage{algorithm}
\usepackage{algpseudocode}
\usepackage{graphicx}
\usepackage{enumerate}
\usepackage{natbib}
\usepackage{url} % not crucial - just used below for the URL 
\usepackage[colorlinks,linkcolor=blue,citecolor=blue,urlcolor=blue,filecolor=blue,backref=page]{hyperref}

\providecommand{\keywords}[1]
{
  \small	
  \textbf{\textit{Keywords---}} #1
}

\usepackage{tikz}
\usetikzlibrary{shapes.geometric, arrows}

\tikzstyle{startstop} = [rectangle, rounded corners, 
minimum width=3cm, 
minimum height=1cm,
text centered, 
draw=black, 
fill=blue!30]

\tikzstyle{process} = [rectangle, 
minimum width=3cm, 
minimum height=1cm, 
text centered, 
text width=3cm, 
draw=black, 
fill=orange!30]

\tikzstyle{arrow} = [thick,->,>=stealth]

\tikzstyle{line} = [draw, thick, color=black!50, -latex']

\def\spacingset#1{\renewcommand{\baselinestretch}%
{#1}\small\normalsize} 

\def\spacingset#1{\renewcommand{\baselinestretch}%
{#1}\small\normalsize}
\spacingset{1.1}
\date{}

  \title{\bf \Large Non-segmental Bayesian Detection of Multiple
Change-points}
  \author{ Chong Zhong\thanks{
  \footnotesize
   Co-first author. 
   The author is a PhD student of Department of Applied Mathematics, The Hong Kong Polytechnic University. },\hspace{.1cm}
     Zhihua Ma\thanks{  
    \footnotesize
    Co-first author. The author is an Assistant professor of School of Economic, Shenzhen University, China.
   },\hspace{.1cm}
Xu Zhang\thanks{
    \footnotesize
    The author is a PostDoc fellow of Department of Applied Mathematics, The Hong Kong Polytechnic University, HKSAR.},\hspace{.1cm} and
    Catherine C. Liu \thanks{
      \footnotesize
    The author is an Associate Professor of Department of Applied Mathematics, The Hong Kong Polytechnic University, HKSAR (\url{macliu@polyu.edu.hk}). } 
    }
\begin{document}
\maketitle

\begin{abstract}
We propose an original and general NOn-SEgmental (NOSE) approach  for the detection of multiple change-points. 
NOSE identifies change-points by the non-negligibility of posterior estimates of the jump heights. 
Alternatively, under the Bayesian paradigm, 
NOSE treats the step-wise signal as a global infinite dimensional parameter drawn from a proposed process of atomic representation, 
where the random jump heights determine the locations and the number of change-points simultaneously. 
The random jump heights are further modeled by a Gamma-Indian buffet process shrinkage prior under the form of discrete spike-and-slab.
The induced  maximum a posteriori estimates of the jump heights are consistent and enjoy zerodiminishing false negative rate in discrimination under a 3-sigma rule. 
The success of NOSE is guaranteed by the posterior inferential results such as the minimaxity of posterior contraction rate, and posterior consistency of both locations and the number of abrupt changes.
NOSE is applicable and effective to detect scale shifts, mean shifts, and structural changes in regression coefficients under linear or autoregression models.
Comprehensive simulations and several real-world examples demonstrate the superiority of NOSE in detecting abrupt changes under various data settings.
\end{abstract}
\keywords{Change-point, Minimax optimal rate, Posterior consistency, Spike-and-slab, 3-sigma}

\newpage
\spacingset{1.7}
\section{Introduction}

Detection of multiple change-points has long been an active research topic with a broad range of applications in economics, health study, genetics, and finance, to name a few.
The change detection is needy in cases with mean shifts (\cite{frick2014multiscale}; \cite{fryzlewicz2014wild}; \cite{du2016stepwise}; \cite{romano2022detecting}; among others), scale shifts (\cite{killick2012optimal}; \cite{haynes2017computationally}; among others), and structural abrupt changes in regression models (\cite{bai2003computation}; \cite{korkas2017multiple}; \cite{baranowski2019narrowest}; among others). 
Since the abrupt change pattern used to be mathematically expressed as a stepwise function or sum of segment-wise functions, 
existing methods incline to study segmental parameters such as piecewise mean parameters and segment-wise log-likelihood ratios to unveil the changes such as the number, locations, and jump sizes.
In this article, we attempt to propose an original and general procedure of change-point detection under a novel NOn-SEgmental (NOSE) spirit which models the pure jump process of the change mechanism by a \textit{global infinite-dimensional parameter}.

Our approach is motivated by a suspected change-point under-discrimination case arising from asset pricing and portfolio management.
Specifically, we look into the US log daily returns of agriculture industry portfolios (DRAIP) from January 2007 to December 2019, available at \href{http://mba.tuck.dartmouth.edu/pages/faculty/ken.french/data\_library.html}{http://mba.tuck.dartmouth.edu}. 
Understanding the shifts on the scale of the recast daily return data can help evaluate the risk of investment on these portfolios since the variation of daily returns usually acts as a measure of the risk of a portfolio. 
The DRAIP dataset is displayed as a black line in Figure \ref{fig:scaleexample}.
One can observe noticeably that, 
i) the data have no shifts on the mean since all data are centered around zero stably; 
ii) the variations of daily returns have uneven shifts, most of which are modest except the apparent variation on time interval $(400, 500)$.
Existing methods such as
 NOT \citep{baranowski2019narrowest}, SMUCE \citep{frick2014multiscale}, and PELT \citep{killick2012optimal} can work on this dataset to detect scale changes, summarized in Figures \ref{Res_SMUCE}-\ref{Res_PELT}. 
The numbers of change-points detected are 4, 4, and 5, respectively. \textit{Nonetheless}, one may suspect the possibility of under-detection of change-points for areas highlighted in, a) the orange rectangle between (200, 400) that is bouncing-visible and b) the blue rectangle between (0, 200) that is bouncing-mild. 
Note that the aforementioned methods share the same spirit of modeling the \textit{local segment parameters} directly, and may lose the structural information. 
Instead, we are driven to formulate a \textit{global process} for the underneath abrupt change mechanism to discover the possible changes. 
Our approach is introduced in subsections \ref{subsec:global}-\ref{subsec:3sigma}.

\begin{figure}[H]
    \centering
    \subfigure[]{
    \begin{minipage}[t]{0.45\linewidth}
      \centering
\includegraphics[height = .5\textwidth]{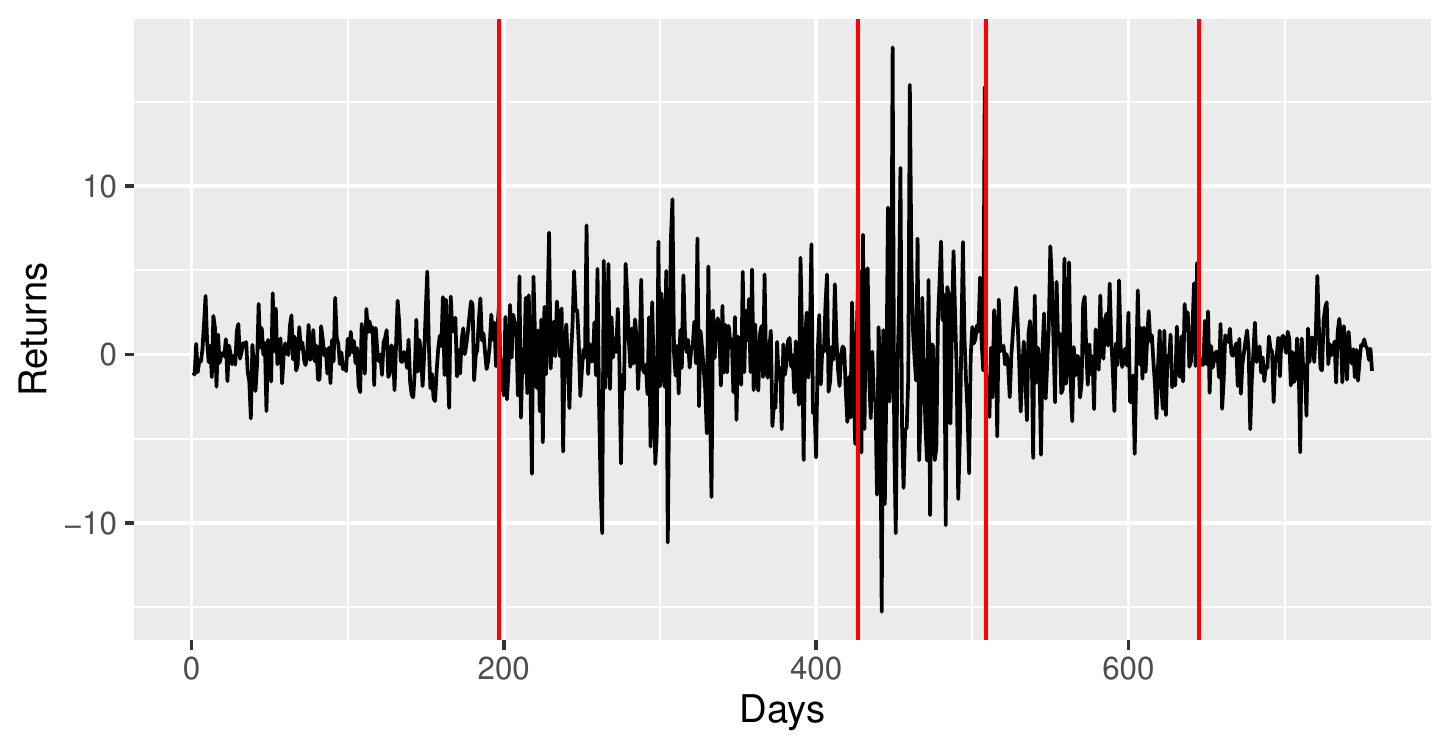}
\label{Res_SMUCE}
    \end{minipage}
    }
\subfigure[]{
    \begin{minipage}[t]{0.45\linewidth}
      \centering
\includegraphics[height = .5\textwidth]{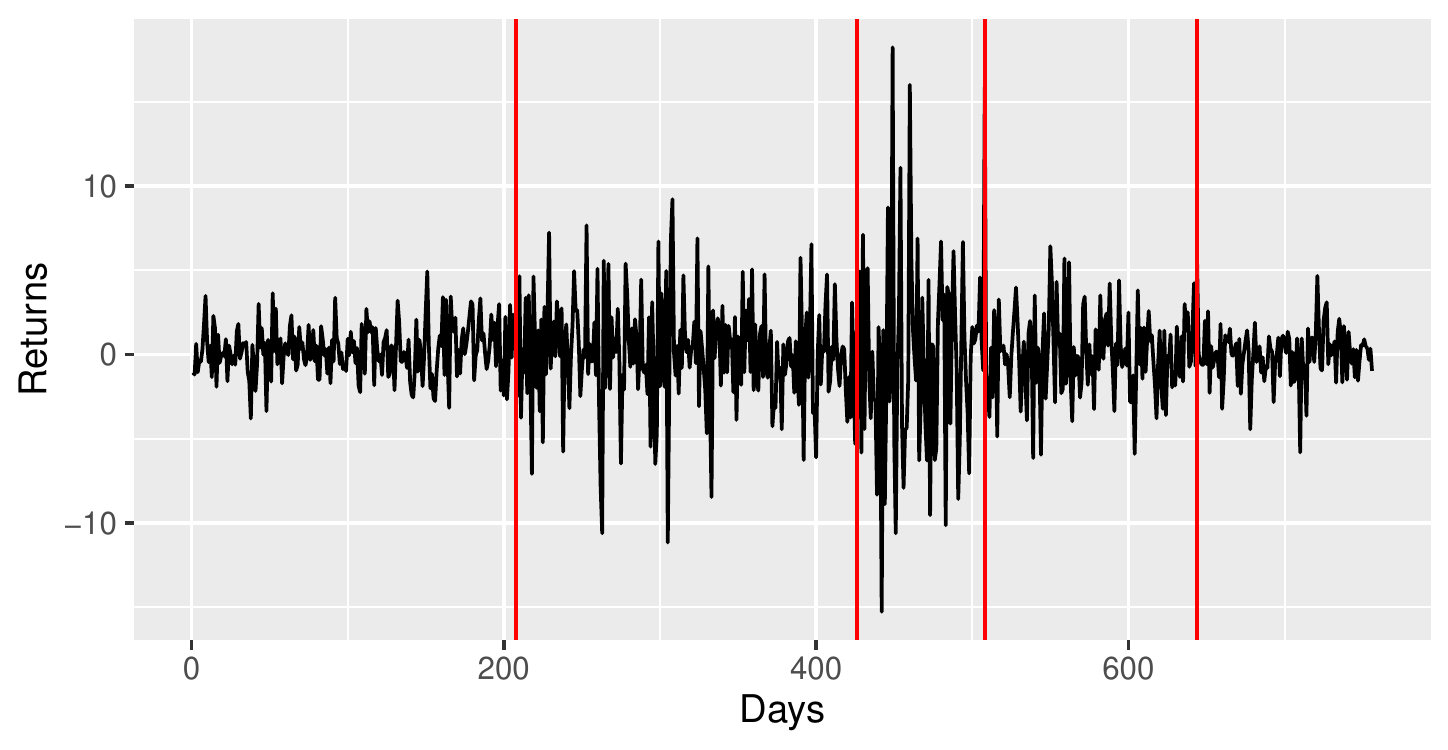}
\label{Res_NOT}
    \end{minipage}
    }
    \subfigure[]{
    \begin{minipage}[t]{0.45\linewidth}
      \centering
\includegraphics[height = .5\textwidth]{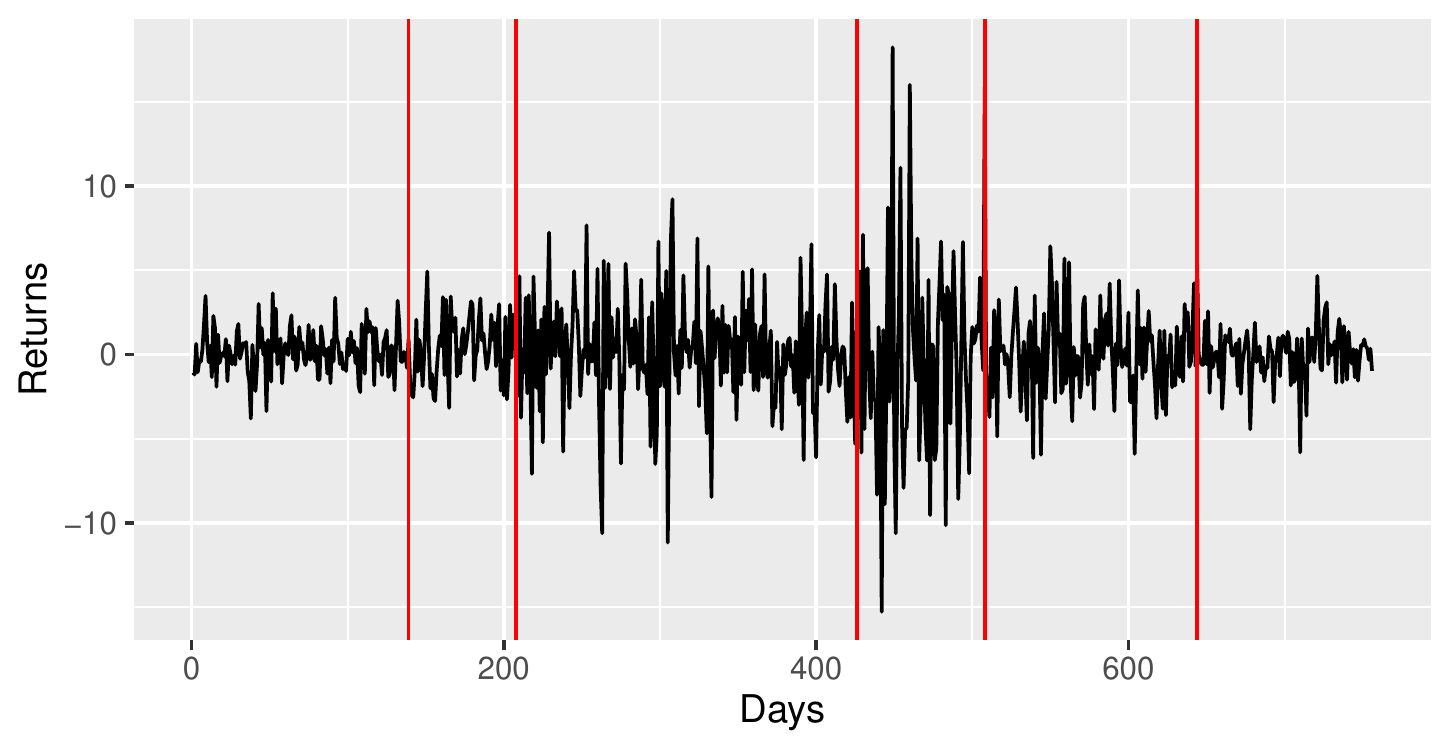}
\label{Res_PELT}
    \end{minipage}
    }
\subfigure[]{
    \begin{minipage}[t]{0.45\linewidth}
      \centering
\includegraphics[height = .5\textwidth]{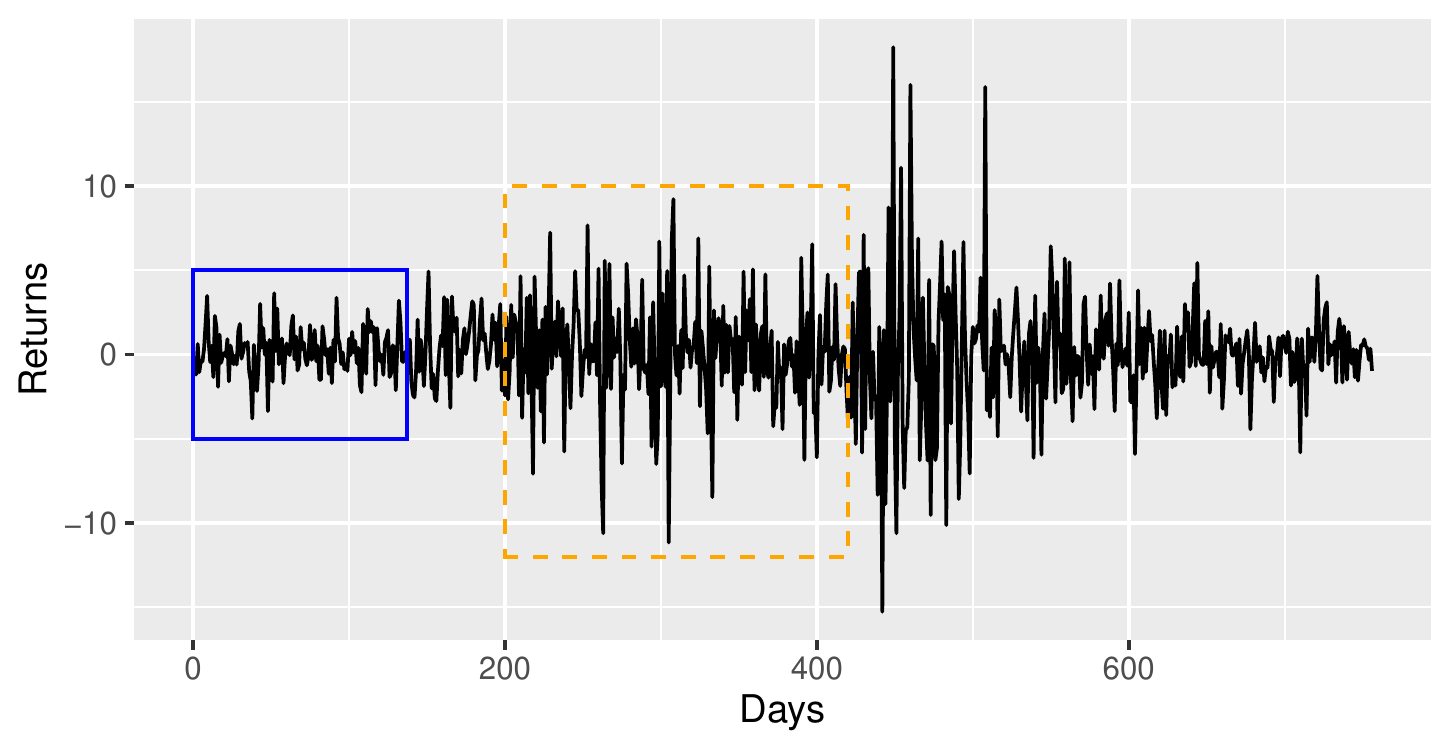}
\label{MissRegion}
    \end{minipage}
    }
    \caption{\footnotesize{Plots of  estimated locations of change-points (in red vertical lines) by different methods and DRAIP data (in black lines). (a), SMUCE; (b), NOT; (c), PELT; (d) original data.  }}
    \label{fig:scaleexample}
\end{figure}

\subsection{Global curve function parameter $\theta(t)$}\label{subsec:global}
The abrupt change, in almost all literature, is characterized as a \textit{pure jump process} $\sum_{k=1}^{K+1} \theta_k \\ I(\tau_{k-1} \le t < \tau_{k})$, and have been dealt with by focusing on segment parameters $\theta_k$ directly. 
Here $K$ denotes the unknown total number of change-points, $\tau_k$ denotes the $k$-th change-point,
and  the argument $t$ is defined on a state space $\mathcal{T}$ that is not limited to a temporal or spatial state.
Let $\bm{\tau}_{1:K} = \{\tau_1, \ldots, \tau_K\}$, where $\tau$ can be a placeholder.
We assume that the adjacent $\theta_k$'s are distinguishable in the sense that $\theta_k \not = \theta_{k+1}$ for all $1\le k \le K$. 
Rather than looking into local segmental parameters $\theta_k$, we globally denote the pure jump process or the stepwise function as $\theta(t)$. 
Consequently, our approach starts from an atomic representation of the curve function $\theta(t)$ from the perspective of jump sizes and locations of change-points.

Let $(h_1, \xi_1),  (h_2, \xi_2), \ldots$ be a countably infinite collection of atoms and heights at locations. 
A draw of an atomic random measure is written as
\begin{align}\label{model:atomicprocess}
    q(\cdot) \equiv \sum_{\ell=1}^\infty h_\ell\delta_{\xi_\ell}(\cdot), 
\end{align}
where $\delta_{\xi_\ell}$ is  an atom at $\xi_\ell$ with $h_\ell$ being its height of the jump in $q$.   
Then, we propose a prior process $\bm{Q}$ for $\theta(t)$ in the form of the cumulative integral of $q$
\begin{eqnarray}
\label{prior:theta}
\begin{aligned}
   \theta(t) \sim \bm{Q} \equiv \int_{-\infty}^t q(u)du = \sum_{\ell=1}^\infty h_\ell I(\xi_\ell \le t).
\end{aligned}
\end{eqnarray}
As the jumps may be downward or upward,
the jump sizes $h_\ell \in \mathbb{R}$ are allowed to be \textit{sign-varying} and may be \textit{dependent}
rather than being \textit{non-negative} and \textit{independent} in the atomic representation in a completely random measure \citep{kingman1967completely}. 

Since those jumps with negligible heights are not considered to be abrupt changes, 
one may approximate the prior process $\bm{Q}$ in a truncation form $\bm{Q}^L$, 
\begin{align}
    \label{prior:truncation}
     \bm{Q}^L = \int_{-\infty}^tq^{L}(u)du = \sum_{\ell=1}^L h_\ell I(\xi_\ell \le t) ~\text{with}~  q^{L} = \sum_{\ell=1}^L h_\ell \delta_{\xi_\ell}.  
\end{align}
In practice, one may assume the number of change-points $K$ is bounded by some sufficiently large number $ L$, say, $L=[n/D]$, the integer part of the ratio between the number of observations $n$ and $D$.
Here $D$ reflects one's prior belief on the minimum distance between any two adjacent change-points.
For example, the PELT method sets the default minimum segment length as $D=2$ in the \texttt{R} package \texttt{changepoint} \citep{killick2014changepoint}. 
In Theorem \ref{theorem:converge} of Section \ref{sec:implement}, we will state the asymptotic equivalence of the truncation form \eqref{prior:truncation} to the atomic expression \eqref{prior:theta} under the Gamma-IBP prior model proposed in \eqref{prior:h}.

\subsection{Shrinkage prior for $\theta(t)$}\label{subsec:shrinkage}
Let $\theta(t) \equiv \theta$. The underlying distribution for drawing a sample sequence $\bm{y}=(y_1,\ldots,y_n) $ is denoted by $f(\cdot|\theta, \bm{\gamma})$, where $\theta$ is the abrupt change parameter that determines the abrupt changes and $\bm{\gamma}$ is  the nuisance parameters that does not contribute to the abrupt change mechanism. 
Suppose that the $n$ samples $\bm{y}$ are observed at $\bm{t}_{1:n}$. 
Then the likelihood is 
$$
\bm{l}(\bm{y} |\theta, \bm{\gamma}) = \prod_{i=1}^n f(y_i|\theta(t_i), \bm{\gamma}). 
$$

This brings us to the posterior estimate of $\theta(t)$ under prior \eqref{prior:truncation}. 
Once we obtain a posterior estimate based on the observed data $\bm{y}$, 
we immediately have the increments of $\theta(t)$ between $t_i$ and $t_{i+1}$, denoted as $d_i = \theta(t_{i+1}) - \theta(t_i)$. 
The increment sequence $d_i$ acts as a KEY signal of change-points in our methodology:
clearly, the jump height vector $\bm{d} = (d_1, \ldots, d_{n-1})$ represents the jump heights/sizes at all states. 
Thus, those locations with non-negligible jump sizes are naturally segregated from those ignorable and thus, identified as change-points. 
Consequently, we tend to employ posterior estimates of $d_i$ sequence as the features to discriminate change-points based on some criterion rule that will be presented in subsection \ref{subsec:3sigma}.

Note that drawing a random trajectory of $\theta(t)$ is equivalent to randomly drawing vectors $\bm{\xi} = (\xi_1, \ldots, \xi_L)$ and $\bm{h} = (h_1, \ldots, h_L)$. 
Since $\bm{h}$ are heights of atoms at $\bm{\xi}$, we sample $\bm{\xi}$ first and then sample $\bm{h}$, and randomly assign $\bm{h}$ to the atoms. 
Since one can only observe $\bm{y}$ at discrete states $\bm{t}_{1:n}$, it is meaningless to assume that the change-points take place between two adjacent data points. 
Hence,  we assume that all jumps of $\theta(t)$ only take place on $t_i, i=1, \ldots, (n-1)$ without loss of generality (the last data point is omitted as a change-point). 
Then the prior for atoms $\xi_\ell$ is naturally defined as 
\begin{eqnarray}
\label{prior:xi}
\begin{aligned}
     &\xi_1 \sim U(\bm{t}_{1:(n-1)}), 
     &\xi_\ell|\xi_1, \ldots, \xi_{\ell-1} \sim U(\bm{t}_{1:(n-1)}\setminus \bm{\xi}_{1:(\ell-1)}), ~\ell\ge 2,
\end{aligned}
\end{eqnarray}
where $Z \setminus A$ denotes the complement of set $A$ given the universe $Z$. 
In other words, $\xi_\ell$ are sampled from $\bm{t}_{1:(n-1)}$ uniformly without replacement. 
As a result, $\bm{\xi}$ is just a subset of $\bm{t}_{1:(n-1)}$ for any $L < (n-1)$.   

Note that under prior \eqref{prior:xi}, $\bm{h}$ is a subset of $\bm{d}$ containing all non-zero entries of $\bm{d}$. 
Hence we will discuss the sparseness of the jump height vector $\bm{d}$ before the prior elicitation of $\bm{h}$. 

\subsubsection*{Nearly black vector: $K_n$-sparsity}
In general, we allow  the number of change-points $K$ to be arbitrarily large but require $K << n$ as $n \to \infty$. 
One may select a sufficiently large truncation number $L$ so that $K<< L$ too. 
Then the jump height vector $\bm{d}$ belongs to $l_0[K_n]$, a class of \textit{nearly black vectors} (\cite{donoho1992maximum}; \cite{castillo2012needles}), explicitly expressed as 
$$
l_0[K_n] = \{\bm{v} \in \mathbb{R}^{p}: \sum_{i=1}^p I(|v_i|>0) \le K_n\},
$$
where $v_i$ is the $i$th entry of $\bm{v}$ and $K_n (\ge K)$ is a given integer so that $K_n =o(L)$, as $n , L\to \infty$.
We call that $\bm{d}$ possesses $K_n$-sparsity. 
Note that $\bm{h}$ is also $K_n$-sparse since $\bm{d}$ and $\bm{h}$ share the same cardinality. 

Under the above $K_n$ sparsity, we transfer change-point detection to searching for a sparse posterior solution to the jump height vector $\bm{d}$ and $\bm{h}$. 
Therefore, we will introduce next a shrinkage prior for the random vector $\bm{h}$ in model \eqref{prior:h}. 
Our $K_n$-sparsity is inspired by the ``horizontal" sparsity of the vector of jump locations in \citet[subsection 6.3]{frick2014multiscale} under Gaussian linear models, though we take a ``vertical" view on the jump heights instead. 
By penalizing the number of change points, the SMUCE method by Frick, Munk, and Sieling attains a minimax optimal rate up to a logarithm term on the distance between locations of true and estimated change-points;
by a constructed shrinkage prior, our proposed NOSE achieves the minimax optimal posterior contraction rate over the $l_0[K_n]$ class within the Bayesian context.
Nonetheless, these two different kinds of views on sparsity lead to different estimation procedures and consistency. 
SMUCE has to estimate the number and locations of change-points sequentially and obtains the consistency of the number of change-points only. 
In contrast, NOSE estimates the number and the locations of change-points \textit{simultaneously} because, under the jump-size-weighted atomic representation \eqref{prior:truncation}, a non-negligible jump size certainly indicates a change-point. 
As a result, NOSE achieves consistency of both the number and  locations of change-points. 

\subsubsection*{Prior for $\bm{h}$: Gamma-IBP model}
The prior for $\bm{h}$ is expressible as follows. 
\begin{eqnarray}\label{prior:h}
\begin{aligned}
h_\ell|Z_\ell &\sim (1-Z_\ell) \delta_0 + Z_\ell F_0, ~F_0 = \text{Laplace}(0, \lambda),\\
Z_\ell|\eta_\ell &\sim \text{Bernoulli}(\eta_\ell),
 ~\eta_\ell = \prod_{j=1}^\ell p_j,
~p_j |\alpha  \sim \text{Beta}(\alpha, 1), ~ \alpha \sim \text{Gamma}(a, b),
\end{aligned}
\end{eqnarray}
where $Z_\ell$ are latent binary variables determined by the sparsity parameters $\eta_\ell$, $\delta_0$ denotes the mass function at $0$, $\text{Laplace}(0, \lambda)$ represents a zero-centered Laplace distribution with precision parameter $\lambda$, and $\text{Gamma}(a, b)$ represents the Gamma distribution with density $\{\Gamma(a) b^a\}^{-1}x^{a-1}\exp(-x/b)$. 
Prior \eqref{prior:h} is a special class of \textit{discrete spike-and-slab prior} with a surely-zero spike $\delta_0$ and a Laplace slab $F_0$. 
Specifically, the sparsity parameters $\eta_\ell$ are exponentially decreasing products of a series of Beta variables with a mass parameter $\alpha$, which is modeled by a Gamma hyperprior for the purpose of dominating the whole sparsity of prior \eqref{prior:h}. 
Consequently, $\bm{Z} = (Z_1, \ldots, Z_L)$ can be viewed as a stick-breaking representation of an $L$-truncated single row in the Indian buffet process (IBP) \citep{teh2007stick}. 
Therefore, prior $\bm{h}$ is named as the Gamma-IBP model hereafter. 

The nest of the IBP construction and the Gamma hyperprior results in a \textit{strict exponential decrease} on the dimensionality $|\bm{Z}|$, and maintains sufficient weight on the true sparsity level $K_n$. 
Therefore, it suffices to reach the minimax optimal posterior contraction rate \citep{castillo2012needles}. 
On the other hand, the IBP construction further controls the tail probability $Pr\{|\bm{Z}| > k\}$ for any $k>0$, and hence, obtains consistent posterior model selection with a smaller cut-off compared to \cite{castillo2015bayesian}.
The detailed justifications and results are summarized in Section \ref{sec:theory}.

\subsection{Discrimination of change-points} \label{subsec:3sigma}
After the prior elicitation in subsection \ref{subsec:shrinkage}, we propose a change-point discrimination procedure based on the induced posterior. 
We first obtain posterior estimates of the increments $\bm{d}$ and then simply compare the value of the estimates with some data-driven threshold. 
Under the priors \eqref{prior:xi} and \eqref{prior:h}, the posterior of $\bm{\xi}$ and $\bm{h}$ are sampled through Markov Chain Monte Carlo (MCMC). 
Suppose one has drawn $N$ posterior samples of $\bm{h}$ and $\bm{\xi}$, denoted as $^{j}h_\ell$ and $^{j}\xi_\ell, j=1, \ldots, N$. 
Then for any $t_i$, the marginal posterior samples of $\theta(t_i)$ are determined as $^{j}\theta(t_i) = \sum_{\ell=1}^{L}  ~ ^{j}h_\ell I(^{j}\xi_\ell \le t_i)$. 

With $N$ marginal posterior samples of $\theta(t_i)$, one can approximate the maximum of posteriori (MAP) estimate of $\theta(t_i)$ as the mode of sample density of $\{^{j}\theta(t_i)\}_{j=1}^N$, denoted as $\hat{\theta}(t_i)^{\text{MAP}}$. 
Let $\{\zeta_i\}_{i=2}^n$ be 
$$ 
\zeta_i = \hat{\theta}(t_{i+1})^{\text{MAP}} - \hat{\theta}(t_i)^{\text{MAP}}, i=1, \ldots, (n-1), 
$$
the diffed series of $\hat{\theta}(t_i)^{\text{MAP}}$. 
Note that $\zeta_i$ is a posterior estimate of $d_i$ i.e. a posterior estimate of the jump size at $t_i$. 
Nevertheless, $\zeta_i$ is not the MAP estimate $\hat{d}_i^{\text{MAP}} = \{\widehat{\theta(t_{i+1}) - \theta(t_i)}\}^{\text{MAP}}$ but an approximation to $\hat{d}_i^{\text{MAP}}$ in practice. 
The reason why we do not employ $\hat{d}_i^{\text{MAP}}$ directly is that the marginal posterior of $d_i$ is poorly approximated by MCMC samples due to high auto-correlation between samples of $^{j}d_i = \{^{j}\theta(t_{i+1}) - ^{j}\theta(t_{i})\}, j=1, \ldots, N$. 
Therefore, the density of $d_i$ estimated from MCMC samples of $\theta(t_i)$ is useless and so is the mode. 
Let $\hat{\sigma} \equiv (\text{Var}\{\zeta_i\}_{i=1}^{n-1})^{1/2}$ be the sample standard deviation of $\{\zeta_i\}_{i=2}^n$. 
Then we determine change-point locations $\tau_k, k \in 1, \ldots, K$ based on the following discrimination rule. 

\subsubsection*{Discrimination rule} 
 \textbf{3-sigma} \textit{If at $t_i$, the absolute posterior estimate of jump size $|\zeta_{i}|> 3\hat{\sigma}$, then $t_i$ is discriminated as a change-point; otherwise, not a change-point.}

It is intuitive to employ the above  3-sigma rule for change-point discrimination due to the  nearly black nature of $\bm{d}$. 
The 3-sigma rule  has been widely used in outlier detection \citep{pukelsheim1994three}, where the outliers are considered to be far away from the center of the population. 
In our case, the nearly black $\bm{d}$ indicates that the population of $\zeta_i$ concentrates at zero except for some outliers. 
Hence, those points that are sufficiently far away from zero are naturally discriminated as outliers, i.e. change-points. 

The threshold for negligibility takes the value $3\hat{\sigma}$. 
It is a kind of ``global" threshold based all entries of the posterior estimates of vector $\bm{d}$. 
In existing approaches, most thresholds for spike-and-slab priors are ``local". 
Some local thresholds shrink those coordinates whose posterior estimates are under some prespecified values to zero (\cite{pati2014posterior}; \cite{rovckova2016fast}; \cite{rovckova2018bayesian}; among others), and the others shrink those coordinates whose posterior non-zero probability is smaller than 0.5 (\cite{barbieri2004optimal}; \cite{scheipl2012spike}; \cite{cappello2023bayesian}; among others). 
However, a local threshold may be sensitive to the ratio between jump sizes and within-segment variations in our numerical experience.
The 3-sigma global criterion grants us a strong ability to recognize those even small jump sizes since each jump size is compared with the vast majority of zeros on stationary points, regardless of the within segment variations.  
Under the 3-sigma rule, we show the near zero false negative rate of discrimination; see Corollary \ref{Falsenegative} in Section \ref{sec:theory}. 

\begin{figure}[!htb]
    \centering
    \small
\begin{tikzpicture}[node distance=5cm]

\node (prior) [startstop] {$\begin{array}{c}\textbf{Prior elicitation}:\\ 
\theta(t) \sim \bm{Q}^L =\sum_{\ell=1}^L h_\ell I(\xi_\ell \le t);\\
\bm{\xi} \sim \mathrm{Uniform};\\
\bm{h} \sim \text{Gamma-IBP model}.\\
\end{array}$};
\node (mcmc) [startstop, below of=prior] { $
\begin{array}{c}
\textbf{Posterior  estimates}:\\
   d_i = \theta(t_{i+1}) - \theta(t); \\
   \hat{\theta}(t_i)^{\text{MAP}}: \text{marginal posterior mode of} ~\theta(t_i);   \\
   \zeta_i=  \hat{\theta}(t_{i+1})^{\text{MAP}} - \hat{\theta}(t_i)^{\text{MAP}}, i=1, \ldots, (n-1).\\
\end{array}
$ };
\node (discrimination) [startstop, below of=mcmc] {
$
\begin{array}{c}
    \textbf{Change-point discrimination (3-sigma)}:   \\
    \hat{\sigma}: \text{sample SD of}~ \zeta_i; \\
    \text{Change-points set}~ \mathcal{S}_{\mathcal{C}} =  \{t_i: I(|\zeta_i| > 3\hat{\sigma}), i<n\}.\\
\end{array}
$
};

\draw [arrow] (prior) -- (mcmc);
\draw [arrow] (mcmc) -- (discrimination);

\end{tikzpicture}
\caption{Flowchart of the proposed methodology.}
    \label{fig:flowchart}
\end{figure}
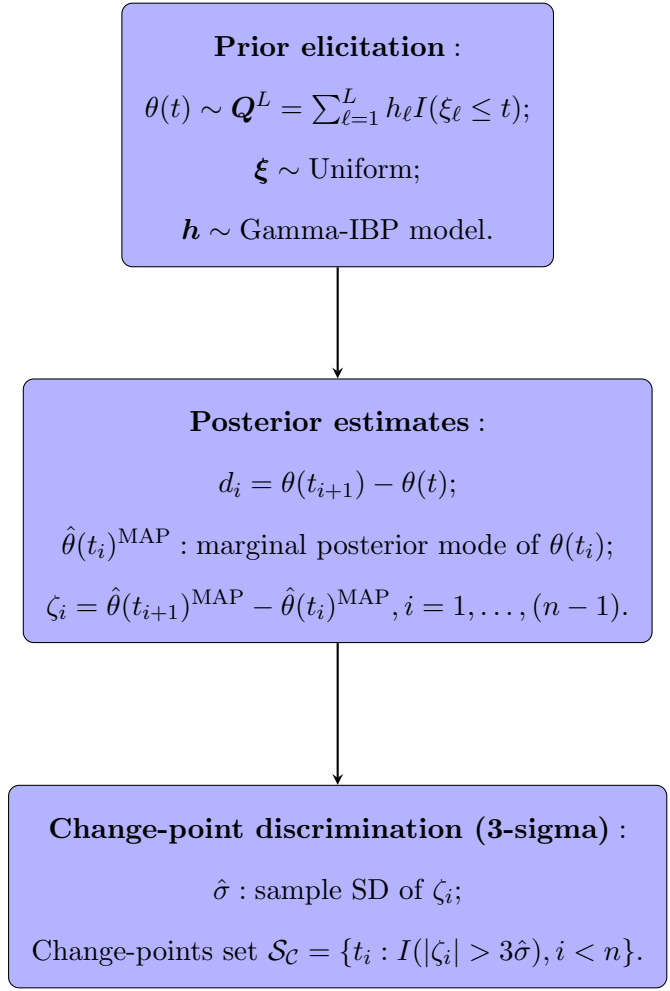

We provide an overview of the workflow of the proposed change-point detection method in Figure \ref{fig:flowchart} and summarize it as follows. 
\begin{enumerate}
    \item[] Step 1: construct a truncated prior for $\theta(t)$ in the form of \eqref{prior:truncation}. 
    Assign priors \eqref{prior:xi} and \eqref{prior:h} to $\bm{\xi}$ and $\bm{h}$, respectively. 
    \item[] Step 2: draw $N$ posterior samples of $\bm{\xi}$ and $\bm{h}$. 
    Obtain the marginal MAP estimate of $\theta(t)$ as $\hat{\theta}(t_i)^{\text{MAP}} = \arg \max_x f_i(x)$, where $f_i$ is the empirical density of $^{j}\theta(t_i) = \sum_{\ell=1}^L ~ ^{j} h_\ell I(^{j}\xi_\ell \le t_i), ~j=1, \ldots, N, ~i=1, \ldots, n$. 
    \item[] Step 3: obtain $\zeta_i = \hat{\theta}(t_{i+1})^{\text{MAP}} - \hat{\theta}(t_i)^{\text{MAP}}$ as an estimate of $d_i$. 
    The set of discriminated change-points is $\mathcal{S}_{\mathcal{C}} = \{t_i: I(|\zeta_i| > 3\hat{\sigma}), i<n\}$. 
\end{enumerate}

\subsection{Application scenarios}\label{subsec:app}
We illustrate some application scenarios of the proposed method here. 
NOSE works in the detection of mean shifts and scale shifts such as,
\begin{enumerate}
    \item[] Scenario 1: shifts in means of Gaussian variables (Gaussian mean-shifted model). We have a series of real observations $y_i \sim N\{\theta(t_i), \sigma^2\}$, for $i=1, \ldots, n$. 
    The global parameter $\theta(t)$ represents the location parameter. 
    
    \item[] Scenario 2: shifts in the parameter of Poisson variables. 
    We have a series of integer observations $y_i \sim \text{Poisson}\{\theta(t_i)\}$, for $i=1, \ldots, n$. 
    The global parameter $\theta(t)$ characterizes the changes in mean and variance simultaneously.

    \item[] Scenario 3: shifts in the scale parameters  of Gaussian variables (Gaussian scale-shifted model).
     We have a series of real observations $y_i \sim N\{\mu, \exp[\theta(t_i)]\}$, for $i=1, \ldots, n$. 
     The global parameter $\theta(t)$ represents the scale parameter through an exponential transformation to guarantee the non-negativity.
  \end{enumerate}

 Meanwhile, NOSE is also applicable to detect structural changes in regression/autoregression models. 
  \begin{enumerate}
     \item[] Scenario 4: structural changes of an AR(1) model. 
     Data are generated from the model 
     $$
     y_t = \phi_0 + \theta(t)y_{t-1} + \epsilon_t, 
     $$ 
     where $\phi_0$ is the fixed intercept,  $E(\epsilon_t) = 0$ and $E(\epsilon_t\epsilon_s) = \sigma^2 I(t=s)$.
     The global parameter $\theta(t)$ represents the autocorrelation coefficient. 
    
    \item[]  Scenario 5: structural changes of a linear regression model. 
    Data are recorded as independent pairs of $(y_{tj}, X_{tj})$, for $j=1, \ldots, n_t, t = 1, \ldots, T$. 
    The association between $y$ and $X$ is characterized by 
    $$
    y_{tj} = \beta_0 + \theta(t)X_{tj} + \epsilon_{tj}, 
    $$
    where $\beta_0$ is a fixed intercept, $E(\epsilon_{tj}) = 0$ and $E(\epsilon_{tj}\epsilon_{sj'}) = \sigma^2 I(t=s)$.
     The global parameter $\theta(t)$ represents the regression coefficient at time $t$. 
     Note that by taking $n_t = 1$ for all $t$ and $X_t = y_{t-1}$, this scenario reduces to Scenario 4. 
\end{enumerate}

\subsection{Related work}
\subsubsection*{Review on segmental approaches}
As we state at the very beginning, most existing methods of change-point detection are segmental approaches in the sense that they \textit{estimate multiple segment parameters or conduct a series of tests based on segment parameters}. One may summarize them into two main streams. 

i) Penalized methods. \textit{Penalized methods optimize an objective function in the sum of segment-specific costs and a penalty}. 
The cost is versatile and chosen based on types of changes (mean, scale, or autocorrelation for instance) while the penalty term is deterministic to the methodology. 
For the penalty term, linear $l_0$ penalization to the \textit{vector of segment parameters/features} to control
\textit{the number of change-points} might be the most popular choice (\cite{yao1984estimation}; \cite{killick2012optimal};  \cite{frick2014multiscale}; \cite{romano2022detecting}; \cite{jula2021multiscale}; among others). 
Alternatively, $l_1$ penalization to the \textit{vector of segment parameters/features and their jump sizes} is also considered (\cite{tibshirani2005sparsity}; \cite{chernozhukov2017lava}; among others). 
We note that Bayesian approaches can be attributed to penalized methods in the sense that one employs priors to automatically penalize the number of change-points (\cite{fearnhead2006exact}; \cite{wyse2011approximate}; \cite{ko2015dirichlet}; among others), or even cover ratios between observations in segments and total sample size \citep{du2016stepwise}. 

ii) Binary-segmentation (BS) variants. 
The BS procedure involves the sequential partitioning of a given data stream into two distinct subsegments \citep{vostrikova1981detecting}. This partitioning is carried out based on the identification of a change-point, which is determined by applying specific testing criteria to the previously split subsegments.
Under this spirit, \cite{fryzlewicz2014wild} developed the so-called ``bottom-up" strategy in the sense that one \textit{determines a change-point from subsets of the data (local ground) and then aggregates local features} as the overall model. 
\cite{baranowski2019narrowest} further enhanced the ``bottom-up" strategy by a narrowest over threshold (NOT) so that they draw the subsample set from the narrowest interval.
There are some other BS variants works such as \cite{cho2015multiple}, \cite{fryzlewicz2018tail}, \cite{fang2020segmentation}; among others. 

\subsubsection*{Spike-and-slab prior revisit}
The spike-and-slab priors are usually categorized as continuous and discrete priors. 
The continuous spike-and-slab employs two continuous densities for both spike and slab terms, with one highly concentrated and the other dispersed (\cite{carlstein1988nonparametric}; \cite{narisetty2014bayesian}; \cite{hahn2015decoupling}; among others). 
It is convenient in MCMC sampling, while the posterior solution may not provide sparse estimates automatically. 
The discrete  spike-and-slab priors (\cite{yen2011majorization};  
\cite{yang2016computational}; 
\cite{shin2021neuronized}; \cite{ray2022variational}; among others) have great progress in recent years from the computational aspect. 
Under a special Gaussian sequence model, \cite{castillo2012needles} establishes the conditions for the minimax optimal contraction with discrete spike-and-slab priors while remaining consistent model selection unsolved. 
Conditions for consistent posterior model selection with discrete spike-and-slab priors are given by \cite{castillo2015bayesian}, while the posterior contraction is not optimal.
With a data-dependent slab term, \cite{martin2017empirical} obtains both minimax optimality and model selection consistency under an empirical Bayes approach.

Most of the existing work for discrete spike-and-slab priors considers i.i.d. sparsity parameters. 
In this article, our discrete spike-and-slab prior is coupled with  dynamic IBP stick-breaking weights. 
Such kind of dynamic spike-and-slab prior was first employed by \citep{williamson2010ibp} for topic modeling. 
 It has been extended to factor models with possibly infinite many factors (\cite{knowles2011nonparametric}; \cite{rovckova2016fast}; \cite{james2017bayesian}; \cite{ma2022posterior}; \cite{ohn2022posterior}; among others). 
We are the first to employ the IBP discrete spike-and-slab to change-point detection, unlike existing work that employs continuous spike-and-slab prior with invariant sparsity parameter \citep{cappello2023bayesian}. 

The rest of this article is organized as follows. 
Section \ref{sec:theory} studies the asymptotic behavior of the posterior and detection performance. 
Section \ref{sec:implement} provides technical details of the Bayesian implementation of our method. 
Sections \ref{sec:simu} and \ref{sec:real_data} present comprehensive simulations and applications to extensive real-world data examples, followed by a brief discussion in Section \ref{sec:discussion}. 
Mathematical proofs and results of additional simulations are included in Appendices. 
The companion \texttt{R} package NOSE is available online.

\section{Asymptotic behavior of posterior}\label{sec:theory}
In this section, we present the theoretical results of the proposed change-point detection method in the asymptotic regime $n, L \to \infty$. 
Particularly, we will analyze the aforementioned Gaussian mean-shifted model with invariant variance, the Scenario $1$ in subsection \ref{subsec:app}. 
Since detecting shifts of means might be the most common and important change-point detection problem, studying the asymptotic behavior of the proposed method in this scenario is meaningful. 
As we mentioned before, the jump height vector $\bm{d}$ contains all information about the jump sizes, which are deterministic in our approach. 
Therefore, we will focus on the posterior of $\bm{d}$. 
We study THREE aspects of asymptotic behaviors, $\bm{1})$ minimax optimal posterior contraction rate and recovery with under detection, $\bm{2})$ posterior consistency of  model selection, and $\bm{3})$ asymptotic zero false negative rate of change-point discrimination under the 3-sigma rule. 

From our insight, given the scale parameter $\sigma$ in Scenario 1, the Gaussian mean-shifted model can be rewritten as a Gaussian sequence model \citep{castillo2012needles}. 
Without loss of generality, we assume $\sigma = (\sqrt{2})^{-1}$. 
If not, one can simply transform the data and will not change the results. 
Let $\bm{y}^*$ be the differed series of $\bm{y}$, so that $y_i^* = y_{i+1} - y_i$  for $i=1, \ldots, n-1$. 
Then we obtain the following Gaussian sequence model 
\begin{align}\label{GaussianSeq}
     y_i^* \sim N\left(d_i, 1\right), i=1, \ldots, (n-1). 
\end{align}
Our theoretical results are given under model \eqref{GaussianSeq}. 

\subsubsection*{Notation}
Let $p = n-1$ and $\bm{d}_0 = (d_{01}, \ldots, d_{0p})^T$ be the ``true" jump height vector.  
We shall assume that the $\bm{d}_0 \in l_0[K_n]$ for some given number $K_n$ such that the number of change-points $K \le K_n$. 
Since the specification of $L$ depends on $n$ or $p$, we use $L_n$ in this section. 
Hereafter, let $\Pi_{n, L_n}{(\mathcal{B}|\bm{y}^*)}$ denotes the posterior probability on a Borel set $\mathcal{B}$ under priors \eqref{prior:xi} and \eqref{prior:h} given data $\bm{y}^*$. 
Let $P_{\bm{d}_0}$ and $E_{\bm{d}_0}$ denote the probability measure and the expectation operator under the law $N(\bm{d}_0, I_p)$, respectively.

\subsection{Posterior contraction}
\label{subsec:posterior contract}
We first give asymptotic results on the posterior contraction of the jump height vector $\bm{d}$. 
This contraction rate evaluates the capability that the posterior recover the true jump height vector $\bm{d}$.   
We have the following assumption about $n=p+1, L_n$, and $K_n$. 
\begin{enumerate}
    \item[] (\textbf{A1}) $L_n < p$; $K_n/L_n \to 0$, as $L_n \to  \infty$. 
\end{enumerate}
By selecting $L_n = [n/D]$, where $D > 1$ is some fixed constant, Assumption (A1) is satisfied as $K_n/n \to 0$, which is a common setting in  both high-dimensional regression and change-point literature. 

The posterior contraction rate is the rate that the most mass of the posterior concentrates around a ball of the true vector $\bm{d}_0$. 
In this article, we define the radius of the ball by the following $l^q$ losses \citep{castillo2012needles}
$$
d_q(\bm{d}, \bm{d}_0) = \sum_{i=1}^p |d_i - d_{0i}|^q. 
$$
For $q \in (0, 2]$, \cite{donoho1992maximum} shows that the minimax optimal rate over $l_0[K_n]$ is 
$$
r_n^* = K_n\log^{q/2}(p/K_n).
$$ 
The following theorem gives the posterior contraction rate of $\bm{d}$, which reaches the minimax optimal rate under $l^q$ metrics.

%%%% Minimax Theorem 
\begin{theorem}[Minimax optimal posterior contraction rate]\label{minimax}
Let $a=c_1L_n^{-c_3}, b=c_2L_n^{c_4}$ for some constants $c_1, c_2 >0$ and $c_3 > c_4+1 \ge 2$ in prior \eqref{prior:h}. 
Under Assumption (A1), as $n, L_n, K_n \to \infty$, for a sufficiently large constant $M$, we have 
$$
 \sup_{\bm{d}_0 \in l_0[K_n]}{{E}_{\bm{d}_0}\Pi_{n, L_n}}\{\bm{d}: d_q(\bm{d}, \bm{d}_0) > Mr_n^q K_n^{1-/q/2}|\bm{y}^*\} \to 0, 
$$
where  $r_n\ge \sqrt{K_n \log(L_n/K_n)}$. 
\end{theorem}

It clearly finds that for $q \in (0, 2]$, the posterior contraction rate given by Theorem \ref{minimax} is at the same order of the minimax optimal rate $r_n^*$. 
This result is similar to \citet[Theorem 2.2]{castillo2012needles}, though the Gamma-IBP model in \eqref{prior:h} does not belong to any examples studied by them. 
Actually, the nest form of the IBP prior and the Gamma hyperprior plays a key role in the establishment of Theorem \ref{minimax}. 
As shown by \citet[subsection 3.1]{teh2007stick}, with a fixed $\alpha$, as the truncation number $L_n \to \infty$, $\eta_\ell$ become the order statistics of $\text{Beta}(\alpha/L_n, 1)$ and hence, 
the distribution of the cardinality of the latent indicator $\bm{Z}$ converges to $\text{Poisson}(\alpha)$. 
With the Gamma hyperprior for $\alpha$, the whole prior for $\bm{d}$ can be approximated by a Poisson-Gamma model and hence has strict exponential decrease \citep[Example 2.3]{castillo2012needles}. 
The choices of hyperparameter $(a,b)$ are also essential but not too strict. 
On one hand, the relatively large choice of $b$ in the Gamma hyperprior further grants sufficient weight on the true sparsity level $K_n$ so that the posterior can contract in an optimal rate. 
On the other hand, the very small choice of $a$ makes the Gamma-IBP model sufficiently close to the approximated Poisson-Gamma model. 
We defer the detailed proof to Appendix \ref{subsec:proofTheominmax}. 
Note that we only require the first moment of the Gamma hyperprior $ab = o(L_n^{-1})$ here. 
In practice, one may allow $ab^2 \to \infty$ as $n, L_n \to \infty$ and hence obtain a very flat Gamma prior which is nearly ``noninformative" or ``objective". 

Theorem \ref{minimax} requires that $K_n \to \infty$, which is not a common pattern in change-point problems. 
In most existing literature, 
the number of change-points is assumed to be arbitrarily large but finite (\cite{frick2014multiscale}; \cite{du2016stepwise}; \cite{baranowski2019narrowest}; \cite{romano2022detecting}; among others). 
To this end, in the following, we study the posterior behavior with a finite $K_n$ and set the true number of change-points $K=K_n$.
That is, equivalently, the cardinality of the true jump height vector is $|\bm{d}_0| = K_n$. 

The following theorem tells the posterior contraction rate with under detection of change-points for any $K_n< L_n/2$. 
\begin{theorem}[Recovery with under selection]
\label{theo:RecDim}
 Under conditions in Theorem \ref{minimax}, for $M \ge 10$ and any fixed $K_n < L_n/2$, as $n, L_n \to \infty$, we have 
\begin{align*}
     \sup_{\bm{d}_0 \in l_0[K_n]}{{E}_{ \bm{d}_0}\Pi_{n, L_n}} \{d_1(\bm{d}, \bm{d}_0) > Mr_n, |\bm{d}| \le K_n|\bm{y}^*\} \to 0. 
\end{align*}
\end{theorem}
Theorem \ref{theo:RecDim} is a direct result of Proposition 5.1 in \cite{castillo2012needles} by taking $A=1$. By fact that 
$
\binom{L_n}{K_n} \le (eL_n/K_n)^{K_n} \le \exp(c r_n^2)  
$
for some sufficiently large constant $c$, the right hand side of Proposition 5.1 in \cite{castillo2012needles} tends to zero and hence, Theorem \ref{theo:RecDim} holds. 
The detailed proof is deferred to \citet[Section 5]{castillo2012needles}.

\subsection{Posterior consistency of model selection}
\label{subsec:modelselection}
From the perspective of change-points detection, the model selection corresponds to the capability of correctly detecting the number of change-points, the foremost concern in change-point detection. 
As mentioned before, our approach distinguishes non-negligible jumps  from those zero or near zero. 
Actually, those too close to zero jumps cannot be detected by any method. 
Hence, it is necessary to determine a ``sufficiently small " cut-off of non-negligible jump sizes i.e. the non-negligible entries of the true jump height vector $\bm{d}_0$. 
Let $S_0 = \{i: |d_{0i}|\not = 0\}$ be the support of non-zero coordinates of $\bm{d}_0$ and $S_0^c$ be the support of other zero coordinates. 
In our change-point context, $S_0  = \bm{\tau}_{1:K_n}$. 
Let $S = \{i: |d_{i}|\not = 0\}$ be the support of non-zero coordinates of  $\bm{d}$.
Hence, we will study the model selection result on the following class of jump sizes vectors 
\begin{align*}
    \Tilde{l}_0[K_n] = \{\bm{v} \in l_0[K_n]: \min\limits_{i \in S_0} |d_{0i}|\ge M\sqrt{ K_n\log(L_n/K_n)} \}, 
\end{align*}
where $M$ is given by Theorem \ref{theo:RecDim}. 
The class $\Tilde{l}_0[K_n]$ is similar to those classes with cut-offs for model selection consistency in sparse regression literature. 
In change-point setting, it indicates that all the jump sizes on change-points are bounded away from zero. 
We will show that when $K_n$ is bounded, this cut-off still suffices for model selection consistency. 
In this sense, our cut-off of order ${ K_n\log(L_n/K_n)}$ is slightly better than those cut-offs of order $O(\sqrt{ \log p})$, which are commonly presented in existing literature (\cite{castillo2015bayesian}; \cite{martin2017empirical}; \cite{jeong2021unified}; among others).

Theorem \ref{theo:RecDim} guarantees that if $\bm{d}_0 \in \Tilde{l}_0[K_n]$, the posterior dimensionality of $\bm{d}$ can cover all change-points. 
Meanwhile, we would expect the risk of over-detection to be as small as possible. 
The Gamma-IBP model \eqref{prior:h} provides an exponentially decreasing tail probability for the dimension of $\bm{d}$, controlling the risk of over-detection of change-points. 
Besides, we have to carefully select the precision parameter $\lambda$ of the Laplace slab in prior \eqref{prior:h}. 
Roughly speaking, we require $\lambda$ to be sufficiently small so that the slab is dispersed enough to provide sufficient mass to recover the non-zero entries of $\bm{d}_0$. 
Strictly, we require a precision $\lambda$, so that $\lambda||\bm{d}_0||_1 < \delta$ for some positive but finite constant $\delta$. 
However, $||\bm{d}_0||_1$ is unknown in practice. 
Therefore, we provide the following adaptive $\lambda_n(\delta)$ as the choice of the precision parameter of the Laplace slab under the Gaussian sequence model \eqref{GaussianSeq}. 

Let $\Bar{|\bm{y}|} = p^{-1} \sum_{i=1}^p |y_i^*|$. 
The adaptive $\lambda_n(\delta)$ is given by 
\begin{align}\label{lambdan}
  \lambda_n(\delta) = \frac{\delta}{p\Bar{|\bm{y}|}}. 
\end{align}

With the adaptive $\lambda_n(\delta)$, we obtain the following result of no supersets in model selection.

\begin{theorem}[No supersets]
    \label{theo:NoSupset}
Let $a=c_1L_n^{-c_3}, b=c_2L_n^{c_4}$ for some constants $c_1, c_2 >0$ and $c_3 > c_4+2 \ge 3$ in prior \eqref{prior:h}. 
Under Assumption (A1), for any fixed $K_n < L_n$ and $\delta$, with $\lambda_n(\delta)$  defined in \eqref{lambdan}, as $n, L_n \to \infty$, we have
    $$
       \sup_{\bm{d}_0 \in \Tilde{l}_0[K_n]} E_{\bm{d}_0} \Pi_{n, L_n} \{\bm{d}: |\bm{d}| > K_n|\bm{y}\} \to 0. 
    $$
\end{theorem}

In Theorem \ref{theo:NoSupset}, we take a technical route that is  different from the fashions of either \cite{castillo2015bayesian} or \cite{martin2017empirical}, which depends on an extremely fast decreasing speed on the prior for dimensionality and the conjugacy of data-dependent normal slab respectively.
If one adopts the conditions by \cite{castillo2015bayesian}, the posterior contraction rate may be suboptimal. 
Although \cite{martin2017empirical} can reach both minimax optimality and no supersets simultaneously, their empirical Bayes approach may be difficult to be extended to other change-point scenarios. 
Actually, here we borrow the strength from the bound of the tail probability of IBP weights given by factor model literature \cite{ohn2022posterior}. 
However, the prior by Ohn and Kim is non-adaptive in the sense that it requires information about the true sparsity level $K_n$. 
In contrast, our choice of hyperparameters here only depends on the data sizes $n$ and the truncation number $L$, and hence is adaptive. 
We defer the detailed proof to Appendix \ref{subsec:proofnosuperset}. 

The above theorems indicate the following corollary of the posterior consistency of model selection. 
%%%%% Selection consistency
\begin{corollary}[Consistent model selection]\label{Coro:consisMS}
Under the conditions of Theorem \ref{theo:NoSupset}, as $n, L_n \to \infty$, we have
$$
\inf_{\bm{d}_0 \in \Tilde{l}_0[K_n]} {E_{\bm{d}_0}\Pi_{n, L_n}}\{\bm{d}: S = S_0|\bm{y}\} \to 1. 
$$
\end{corollary}
\begin{proof}
    According to \cite{castillo2015bayesian}, to prove Corollary \ref{Coro:consisMS}, it suffices to proving the following two assertions
    \begin{align*}
        \inf_{\bm{d}_0 \in \Tilde{l}_0[K_n]} &E_{\bm{d}_0} \Pi_{n, L_n} \{\bm{d}: S \supset S_0|\bm{y}\} \to 1, \\
        \sup_{\bm{d}_0 \in \Tilde{l}_0[K_n]} &E_{\bm{d}_0} \Pi_{n, L_n} \{\bm{d}: S \supset S_0,  S \not = S_0|\bm{y}\} \to 0. \\
    \end{align*}
    The first assertion is a direct result of Theorem \ref{theo:RecDim}, and 
    the second assertion  is a direct result of Theorem \ref{theo:NoSupset} since $K=K_n$. 
\end{proof}

Note that Corollary \ref{Coro:consisMS} is about the non-zero coordinates of $\bm{d}$.
In other words, Corollary \ref{Coro:consisMS}  indicates that we obtain posterior consistency of both the number and locations of change-points.

\subsection{False negative rate of discrimination}\label{subsec:falsenegative}
As mentioned in subsection \ref{subsec:3sigma}, we regard the posterior estimator of $\bm{d}$ as the feature to discriminate change-points $\bm{\tau}_{1:K_n}$ from $\bm{t}_{1:n}$ under the 3-sigma rule. 
To study the asymptotic performance of the 3-sigma discrimination, we use the marginal MAP estimator $\hat{d}_i^{\text{MAP}}$ as the signal at $t_i$ for the theoretical concern. 
Note that the 3-sigma criterion in subsection \ref{subsec:3sigma} can be viewed as a data-driven threshold based on series $\{\hat{d}_i^{\text{MAP}}\}_{i=1}^{n-1}$. 

The result of consistent model selection enables us to study the asymptotic performance of $\hat{d}_i^{\text{MAP}}$ for $i \in S_0$.  
Let $\hat{\bm{d}}_{S_0}$ be the least square estimator of non-zero coordinates of $\bm{d}_0$ given the correct model selection $S_0$, that is, 
$$
\hat{\bm{d}}_{S_0} =  \arg \min \limits_{\bm{d}_{S_0}} ||\bm{y}^* - X_{S_0}\bm{d}_{S_0}||_2^2, 
$$
where $X_{S} \in \mathbb{R}^{p \times |S|}$ is the submatrix of $I_p$ with colums on the non-zero coordinates. 
Clearly $X_{S_0}^T X_{S_0} = I_{|S_0|}$. 
Let $\hat{\bm{d}}_{S_0}^{\text{MAP}}$ be the marginal MAP estimators of $\bm{d}$ on the ture non-zero support $S_0$.  
Let $\bm{d}_{0S_0}$ be the true non-zero entries in $\bm{d}_0$. 
The follow corollary states the consistency and asymptotic normality of $\hat{\bm{d}}_{S_0}^{\text{MAP}}$.

\begin{corollary}[Consistency of MAP under strong model selection]
\label{Coro:MAPasym}
    Under conditions in Corollary \ref{Coro:consisMS}, for $\bm{d}_0 \in \Tilde{l}_0[K_n]$ as $n, L_n \to \infty$, we have 
    $$
    \hat{\bm{d}}_{S_0}^{\text{MAP}} \xrightarrow{p} \hat{\bm{d}}_{S_0}, ~ \sqrt{p}(\hat{\bm{d}}_{S_0}^{\text{MAP}} - \bm{d}_{0S_0}) \xrightarrow{d} N(0, I_{|S_0|}). 
    $$
\end{corollary}

The proof of Corollary \ref{Coro:MAPasym} is trivial. 
Under the correct model selection, the prior for $\bm{d}_{S_0}$ is reduced to the continuous Laplace slab and hence, the MAP estimator $\hat{\bm{d}}_{S_0}^{\text{MAP}}$ converges to the maximum likelihood estimator $\hat{\bm{d}}_{S_0}$ almost surely \citep[Theorem 4.16]{pronzato2013design}. 
Since the model selection converges to be correct in probability, it suffices showing the weak convergence of the MAP estimator $\hat{\bm{d}}_{S_0}^{\text{MAP}}$ to $\hat{\bm{d}}_{S_0}$. 
Then the second assertion is established by the central limit theorem.

The above distribution approximation about $\hat{\bm{d}}_{S_0}^{\text{MAP}}$ controls the false negative rate under the 3-sigma rule. 
Let $\Bar{d}_0 = p^{-1}\sum_{i=1}^{p}d_{0i}$, 
$\Bar{d} = p^{-1}\sum_{i=1}^{p}\hat{\bm{d}}_{i}^{\text{MAP}}$, 
$\psi_0 = \sqrt{p^{-1}\sum_{i=1}^{p}(d_{0i} - \Bar{d}_0)^2}$, and 
$\psi = \sqrt{p^{-1}\sum_{i=1}^{p}(\hat{\bm{d}}_{i}^{\text{MAP}} - \Bar{d})^2}$. 
The 3-sigma rule acts as a special hard threshold that shrinks all $|\hat{\bm{d}}_{i}^{\text{MAP}}|<3\psi$ to zero. 
We require an upper bound assumption on the norm of $\bm{d}_0 \in \Tilde{l}_0[K_n]$. 

\begin{enumerate}
    \item[] (\textbf{A2}) There exists a universal constant $M_0$, so that $p^{-1/2}||\bm{d}_0||_2 < M_0[\sqrt{K_n \log(L_n/K_n)}]$. 
\end{enumerate}

Assumption (A2) implies that $3\psi_0$ will not exceed any non-zero entries in $\bm{d}_0$ and hence the 3-sigma rule is suitable for the true jump sizes vector $\bm{d}_0$ is 
The following corollary states that under the 3-sigma rule, the probability that a change-point is wrongly discriminated as a stationary point is asymptotically zero. 
We defer the proof to Appendix \ref{Falsenegative}. 

\begin{corollary}\label{Falsenegative}
    Under the conditions in Corollary \ref{Coro:consisMS} and Assumption (A2), 
    as $n, L_n \to \infty$, we have 
    $$
    \sup_{\bm{d}_0 \in \Tilde{l}_0[K_n]} E_{\bm{d}_0}\Pi_{n, L_n}
    \{|\hat{\bm{d}}_{i}^{\text{MAP}}|<3\psi, i \in S_0 | \bm{y}^*\} \to 0. 
    $$
\end{corollary}
Corollary \ref{Falsenegative} theoretically justifies the 3-sigma criterion for change-point discrimination. 
In general, the 3-sigma rule is employed for outlier detection, especially for the Gaussian population. 
In general, the performance of discriminating the outliers depends on two properties, the variation of the population and the distance between the outliers and the center. 
The cut-off of the $\Tilde{l}_0[K_n]$ class guarantees that those outliers (change-points) differ significantly from the zero-center population (stationary points), while the additional Assumption (A2) avoids those outliers from affecting the variation of all the samples too much. 
Corollary \ref{Falsenegative} implies that even under a very high precision level (3-sigma criterion usually yields a high precision), the recall of the discrimination is sufficiently large and asymptotically converges to one. 
This is supported by our finite sample simulations under the Gaussian mean-shifted model of Scenario $(i)$, where NOSE enjoys higher recall than other competing approaches.

\section{Bayesian implementation}\label{sec:implement}
In this section, we introduce technical details for the Bayesian implementation of the proposed method. 

\subsubsection*{Uniform convergence of $\theta (t)$}
Recall that our methodology stands on $\theta(t)$, the truncated form of $\theta(t)$. 
Hence it is necessary to check the convergence of the truncated form as $L \to \infty$. 
We present the uniform convergence of $\theta (t)$ by the following theorem. 
We defer the proof to Appendix \ref{subsec:proofTheo4}. 

\begin{theorem}[Uniform convergence]\label{theorem:converge}
For any continuous density $F_0$ with support $\mathbb{R}$ in \eqref{prior:h}, given $\bm{\xi}$ and fixed $a, b$ in the Gamma prior for $\alpha$, 
the truncated $\bm{Q}^L$ in \eqref{prior:truncation}  converges to $\bm{Q}$ in \eqref{prior:theta} uniformly for all $t\in \mathcal{T}$ in probability. 
\end{theorem}
In practice, the choice of the truncation number $L$ depends on one's prior belief on the minimum distance between change-points. 
In the case where the number of change-points $K$ is not large, a relatively small $L$ is suggested to simplify MCMC sampling. 
In our experience, when the truncation number exceeds a sufficiently large $L$, the detection result is stable with $L$ increasing, numerically demonstrating Theorem \ref{theorem:converge}. 

\subsubsection*{Cauchy slab}
Note that Theorem \ref{theorem:converge} holds for any continuous density for the slab term. 
This implies that the choice of slab density for $h_\ell$ is not limited to Laplace, but also includes some polynomial-tailed densities such as Student-t or Cauchy which prevent over-shrinkage of the non-negligible entries \citep{bai2020spike}. 
In practice, we recommend a standard Cauchy slab in finite sample cases since we find it improves the accuracy of the estimated number of change-points compared with the Laplace slab. 
Therefore, we use the Cauchy slab throughout all numerical studies in this article. 
An intuitive reason for the use of Cauchy slab is that the  adaptive precision parameter for Laplace slab in subsection \ref{subsec:modelselection} is only suitable for the Gaussian mean-shifted model of Scenario $(i)$, and hence, is not a unified choice. 
In contrast, the Cauchy distribution has infinite first and second moments, acting as a very special precision parameter $\lambda=0$. 
Therefore, the Cauchy slab is unified for all application scenarios and free of parameters to be prespecified. 

Although we have no theoretical evidence for the superiority of the Cauchy slab, it might be explained from the perspective of optimizing the minus log posterior. 
In a discrete spike-and-slab model, the Laplace slab can be viewed as a mixture of $l_0$ and $l_1$ penalties, while the Cauchy slab an be viewed as a mixture of $l_0$ norm and a penalty term increasing in a $\log(1+x^2)$ rate. 
By the fact that $\log(1+x^2) < |x|$ for all $x \not = 0$, the Cauchy slab seems to be a better approximation of $l_0$ penalization, compared with the Laplace slab. 
As discussed by \cite{frick2014multiscale}, $l_0$ penalization might be more suitable for change-point problems than the $l_1$ penalization when the number of change-point may be much smaller than the data size. 

Another numerical evidence for the superiority of the Cauchy slab may be given by \cite{shin2021neuronized}. 
For discrete spike-and-slab priors with i.i.d. sparsity parameters, the Cauchy slab appears to enjoy a lower false positive rate and higher cosine similarity to the true parameter compared with the Laplace slab under linear regression model settings.

\subsubsection*{MCMC sampling}
We approximate the posterior distribution through MCMC sampling. 
Our computation is facilitated by the \texttt{nimble} \citep{de2017programming} package in \texttt{R}, which uses \texttt{BUGS} type syntax \citep{lunn2000winbugs} and compiles the code into \textbf{C++} to facilitate automatic posterior sampling. 
Samplers for different parameters are automatically assigned by \texttt{nimble}. 
For conjugate parameters, say, $p_\ell$, \texttt{nimble} assigns Gibbs samplers; 
for parameters $\xi_\ell$ and $\alpha$, \texttt{nimble} assigns the default Metropolis-Hasting sampler; 
for $h_\ell$ and the corresponding binary indicator $Z_\ell$, we configure a reversible jump MCMC sampler to speed up the sampling. 
The \texttt{R} package NOSE based on \texttt{nimble} includes several \texttt{R} functions applied to 
application scenarios mentioned in subsection \ref{subsec:app}. 

\subsubsection*{Continuous $\xi_\ell$}
To determine a discrete draw from states $\bm{t}_{1:n}$ without replacement is difficult in \texttt{nimble}. 
Hence, we have to make a continuous adjustment to adopt the programming framework of \texttt{nimble}. 
Note that for any $t_i$ and $t_{i+1}$ with an increment $d_i = \theta(t_{i+1}) - \theta(t_i) > 0$, it is equivalent to either draw an atom $\xi_\ell$ at $t_{i+1}$ exactly, or to draw an atom $\xi_\ell \in (t_i, t_{i+1})$. 
This motivates us to consider a continuous prior for $\xi_\ell$ as an approximation. 
Without loss of generality, we assume $t_i = i$ for $i=1, \ldots, n$. 
Then we sample $\xi_\ell$ from a continuous uniform distribution $U(0, n)$ in \texttt{nimble} as the continuous prior for $\xi_\ell$.

A risk of the continuous prior $\xi_\ell$ is that more than one atoms fall into the same interval $(t_i, t_{i+1})$, which may lead to an ill posterior of increment $d_i$. 
Note that the probability that the minimum distance between $L$ uniform $U(0, n)$ variables exceeds 1 is $(1-n^{-1})^L$. 
As $n$ increases to $L/n \to 0$, the probability converges to 1, that is, the probability that an interval $(t_i, t_{i+1})$ contains more than one atom converges to zero. 
Therefore, the continuous scheme of $\xi_\ell$ suffices to approximate prior \eqref{prior:xi} when $n >> L$. 

In the finite sample case, too closely located atoms may cause over-detection of change-points by wrongly putting increments to data points that are close to the true change-points. 
To avoid over-detection, we conduct post-processing of change-point. 
We use the prior belief in the minimum distance $D$ between change-points as the lower bound of the distance between change-points. 
For each two consecutive estimated change-points $\hat{\tau}_k, \hat{\tau}_{k+1}$, if $|\hat{\tau}_k-\hat{\tau}_{k+1}| <D$, we only retain the left end-point $\hat{\tau}_k$ as a change-point but remove the rest.
Such a kind of post-processing based on the prior belief in the minimum distance between change-points is common in most literature (\cite{matteson2014nonparametric}; 
\cite{baranowski2019narrowest}; \cite{cappello2023bayesian}; among others). 
This post-processing is applied throughout all numerical studies in this article. 

\subsubsection*{Adjustment of $\phi$}
In a finite sample experiment, Assumption (A3) may no longer hold, especially if $L$ is chosen as a relatively small number. 
For a sequence $\{\zeta_i\}_{i=1}^{n-1}$, those $\zeta_i$ whose absolute values exceed three times the sample standard deviation may cause a much larger variation than the variation of the zero-center population. 
To avoid a too large sample deviation, we adopt an empirically adjusted value of $\Tilde{\phi}$ rather than using the  sample standard deviation. 
Note that in a standard normal case, the 3-sigma rule indicates a tail probability of $0.001$. 
Therefore, we first obtain a trimmed sample of $\zeta_i$ by cutting off the two tails of $0.0005$ probability. 
Then we use the trimmed sample standard deviation as an empirical adjustment of $\Tilde{\phi}$. 
The adjustment of $\phi$ is used throughout the numerical studies in this article. 

\section{Simulations}
\label{sec:simu}	
Comprehensive simulations are conducted to evaluate the performance of NOSE by comparing it with other state-of-the-art methods available in R Archive Network. 
We consider examples in Scenarios 1-5 introduced in subsection \ref{subsec:app}.  
For Scenario 5, since most existing approaches are not available for this scenario when there are multiple responses observed at the same time, we report the results given by NOSE only. 
Results of additional simulations under model misspecification settings of changes in means with autocorrelated noises, changes in means with heavy-tailed noises, and changes in autocorrelation coefficient with model misspecification are deferred to Appendix \ref{subsec:modelmis}. 

\subsubsection*{Settings}
We consider the following settings. 
Under each simulation setting, $300$ Monte Carlo replicate datasets are generated. 

\begin{enumerate}
    \item[](\textbf{S.1}) 
Changes in normal means on equal segments (in Scenario 1). 
We have  $n=400$ independent Gaussian observations with $K=7$ change-points at $(50, 100, 150, 200, 250, \\ 300, 350)$, leading to $8$ segments with segment mean $\mu = (0, 1.5, 3, 1.5, 3, 0.5, 2, 0)$. 
The common scale parameter is set to be $\sigma = \sqrt{2}$. 

\item[](\textbf{S.2}) Changes of normal mean on unequal-length segments with large variations (in Scenario 1). 
We have  $n=916$ independent Gaussian observations with $K=11$ change-points at $(81, 134, 178, 267, 346, 413, 528, 577, 636, 741, 822)$, leading to $12$ segments with segment mean $\mu = (0, 1.23, -0.248, 
0.861, -0.534, 1.057, 0.369, 1.331, 0.483,\\ 1.105, \\ -1.101, 0)$. The common scale parameter is set to be $\sigma = 1$. 
Some jump sizes are smaller than the within-segment variation, leading to many difficulties in correctly identifying change-points. 

\item[](\textbf{S.3}) Changes of Poisson parameter (in Scenario 2). 
We have $n=400$ independent Poisson variables with $K=7$ change-points at $(50, 100, 150, 200, 250, 300, 350)$, leading to $8$ segments with segment parameter $\lambda = (1, 0.25, 2, 1, 3, 1.5, 2.5, 1)$.

\item[](\textbf{S.4}) Changes of normal scale with small variations on the mean (in Scenario 3). 
The data are generated to simulate the DRAIP data. 
We have $n = 756$ independent Gaussian observations with $K=7$ change-points at $(150, 250, 300, 450, 550, 650, 700)$, leading to $8$ segments with segment scales $\sigma =(1, 1.68, 0.57, 0.20, 2.18, 3.09, 1.83, 1)$. 
Meanwhile, we allow small variations on the mean such that the segment mean is $\mu = (0.056, 0.047, -0.034, -0.017, 0.032,\\ 0.068,
-0.042, 0.017)$. 

\item[](\textbf{S.5}) Changes of autocorrelation coefficient in an AR(1) model (in Scenario 4). 
The data generating process is $Y_t = \phi Y_{t-1} + \phi_0 + \epsilon_t$. 
We have $N=450$ observations with $5$ change-points at $t=(50, 100, 200, 300, 400)$, leading to $6$ segments with segment autocorrelation coefficient $\phi = (0.5, -0.5, 0.65, -0.25, -0.85, 0.45)$. 
The model error $\epsilon_t \sim N(0, 1)$. 

\item[](\textbf{S.6}) Changes of regression coefficient in a linear regression model (in Scenario 5). 
Data are generated by $y_{tj} = \beta_0 + \theta(t)X_{tj} + \epsilon_{tj}, j=1, 2, t=1, \ldots, 240$, where $\beta_0 = 0.5, X_{tj} \sim U(-2, 2)$, and $\epsilon_{tj} \sim N(0, 1)$. 
We set $K=5$ change-points at $t=(40, 80, 120, 160, 200)$, with the segment-wise values $\theta(t) = (1, -1, 0.5, -0.5, 1, -1)$. 

\end{enumerate}

Examples of simulated data are presented in Figure \ref{fig:simexample}. 
Figures \ref{GaussMean7CP} to \ref{Poison7CP} find that some jump sizes are relatively small and the corresponding change-points are imperceptible in the data stream. 
Figure \ref{AR5CP} finds that the data with identical signs are clustered in those segments with positive auto-correlation, and opposite signs of data appear alternately in those segments with negative auto-correlation. 
Figure \ref{Scale7CP} presents the centered absolute data $|Y-EY|$ and the true $\theta(t)$ together, where the heights of the centered absolute data reflect the changes in the scale parameters. 
Figure \ref{Linear5CP} presents the covariates and the responses grouped by the state $t$ and labels the curves by the segments at which they are located. 
\begin{figure}[!htb]
    \centering
    \subfigure[]{
    \begin{minipage}[t]{0.45\linewidth}
      \centering
\includegraphics[height = .65\textwidth]{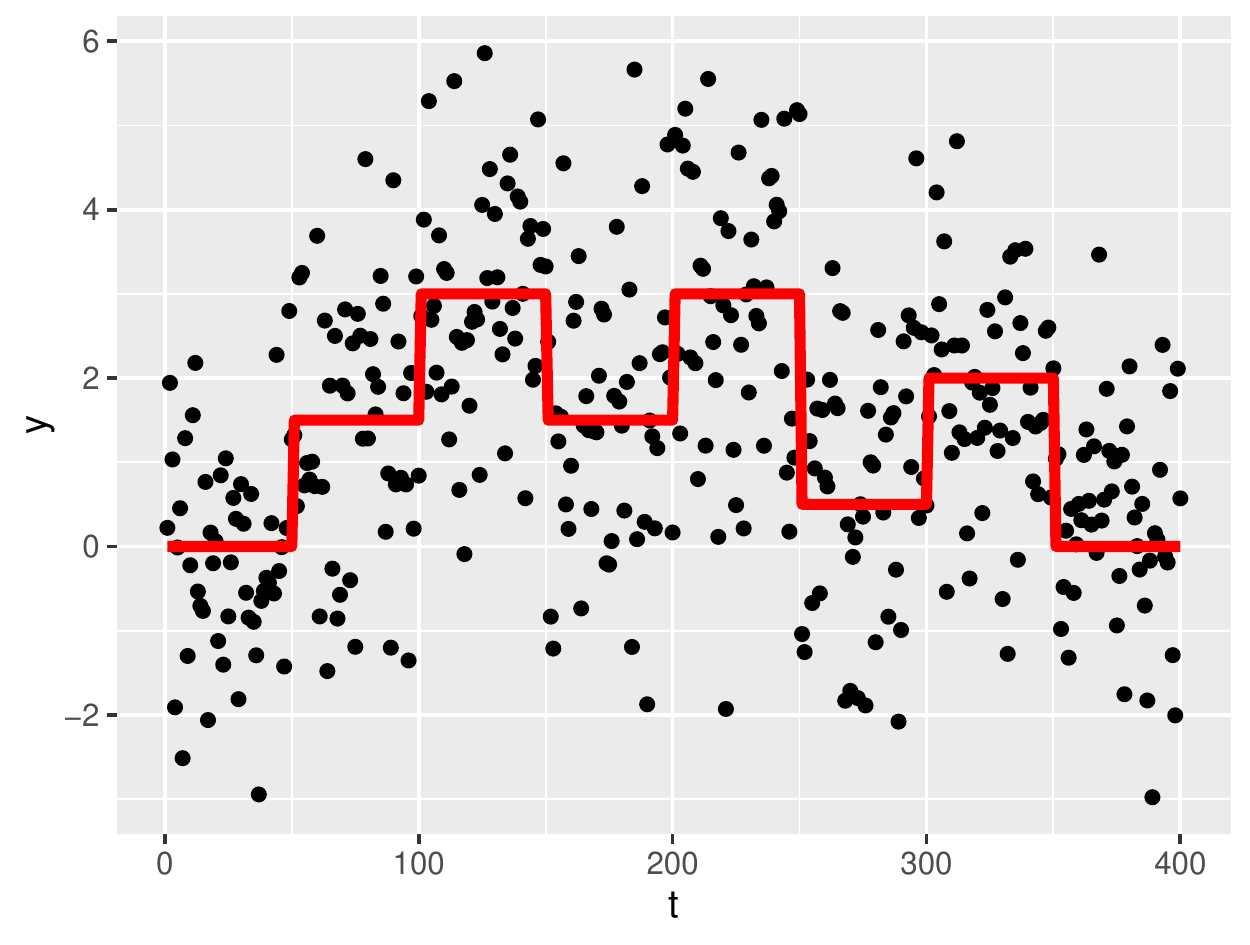}
\label{GaussMean7CP}
    \end{minipage}
    }
\subfigure[]{
    \begin{minipage}[t]{0.45\linewidth}
      \centering
\includegraphics[height = .65\textwidth]{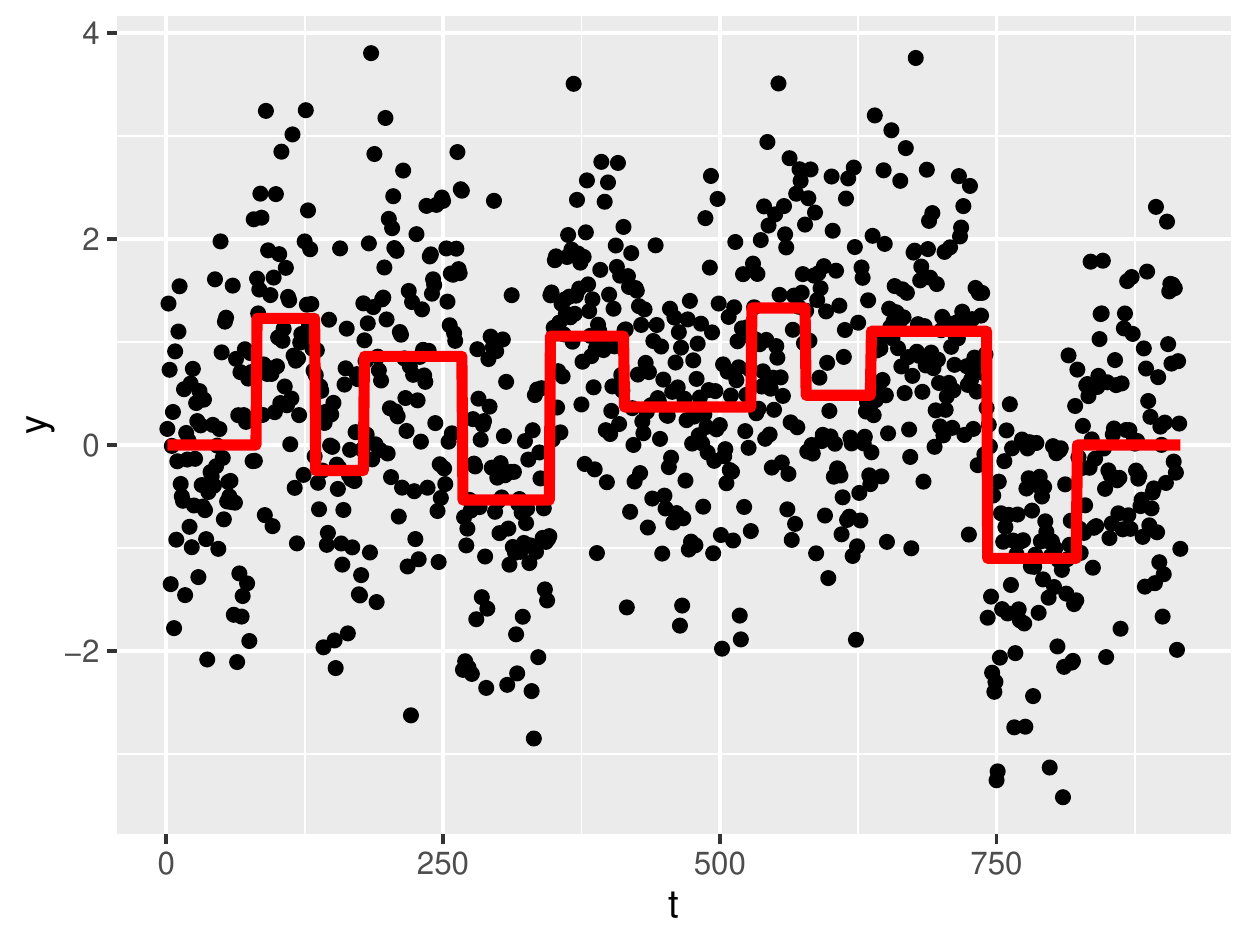}
\label{GaussMean11CP}
    \end{minipage}
    }
    \subfigure[]{
    \begin{minipage}[t]{0.45\linewidth}
      \centering
\includegraphics[height = .65\textwidth]{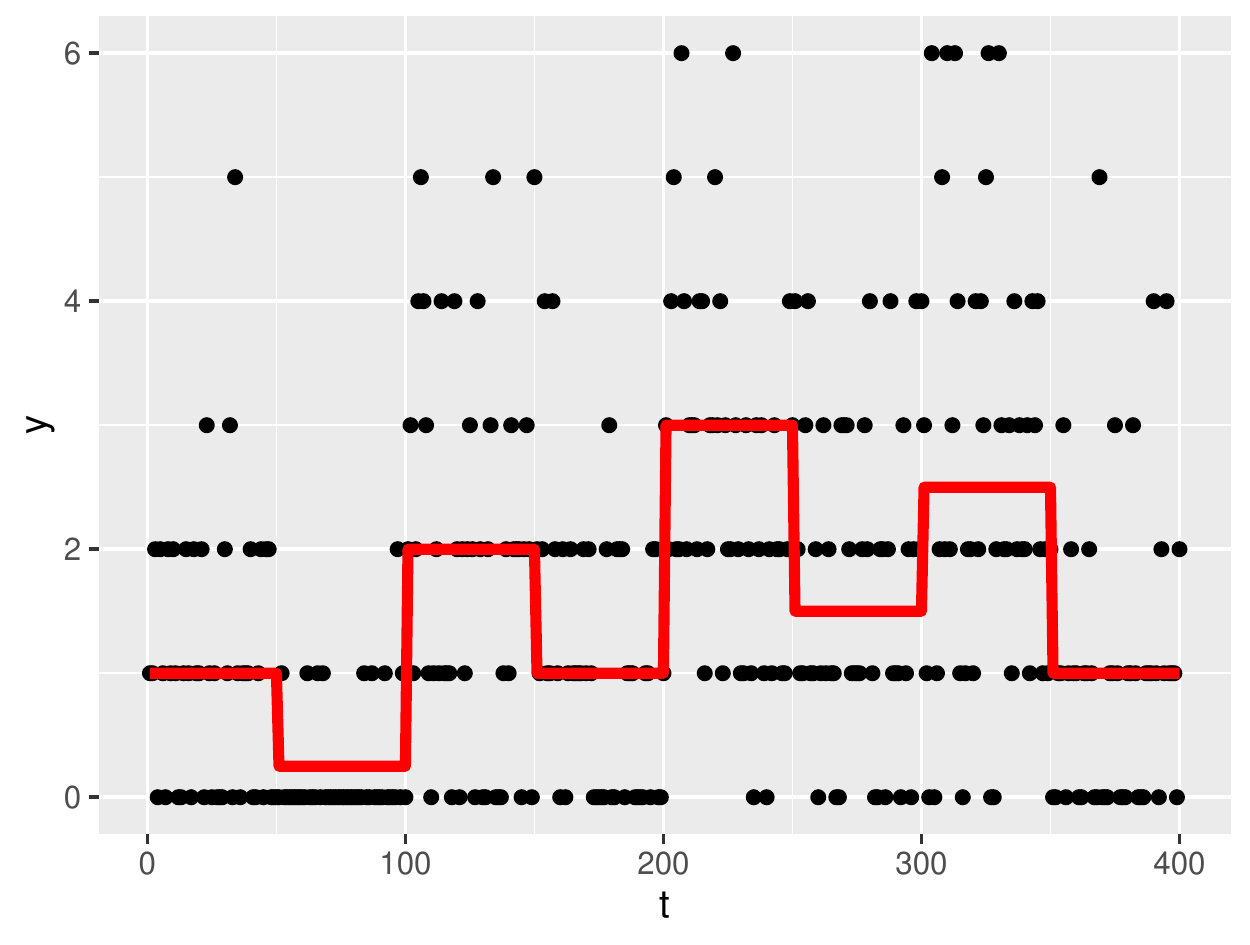}
\label{Poison7CP}
    \end{minipage}
    }
    \subfigure[]{
    \begin{minipage}[t]{0.45\linewidth}
      \centering
\includegraphics[height = .65\textwidth]{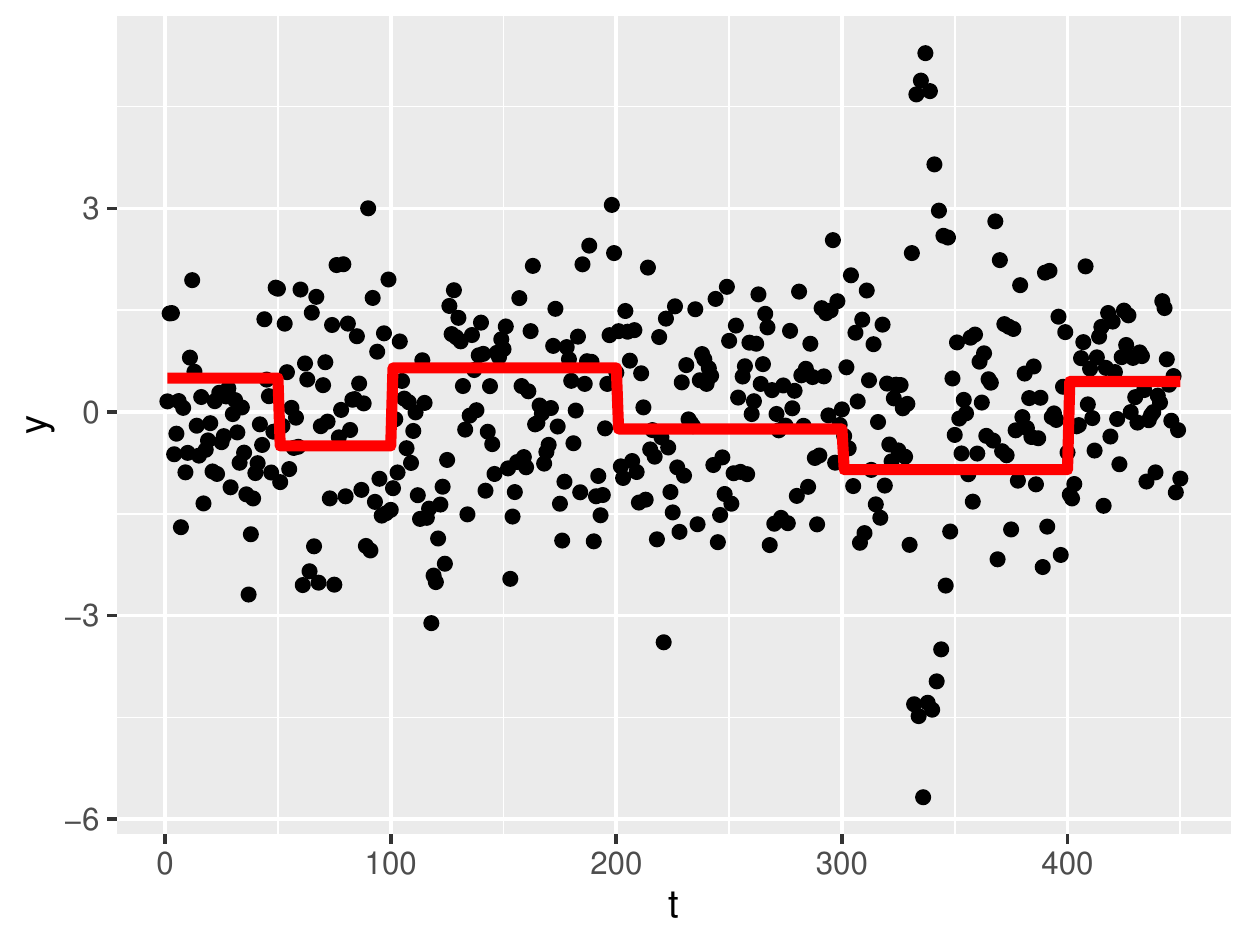}
\label{AR5CP}
    \end{minipage}
    }
\subfigure[]{
    \begin{minipage}[t]{0.45\linewidth}
      \centering
\includegraphics[height = .65\textwidth]{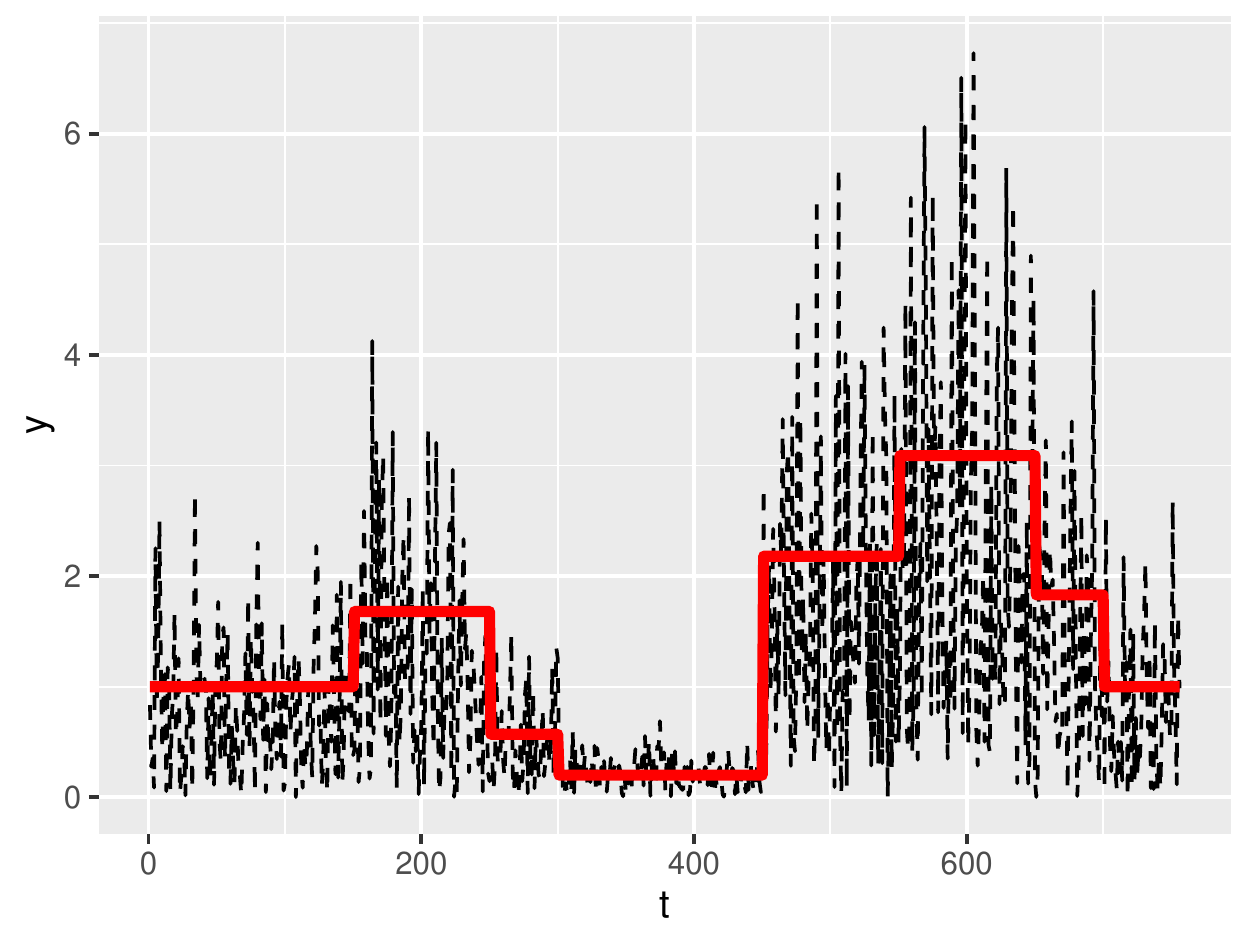}
\label{Scale7CP}
    \end{minipage}
    }
    \subfigure[]{
    \begin{minipage}[t]{0.45\linewidth}
      \centering
\includegraphics[height = .65\textwidth]{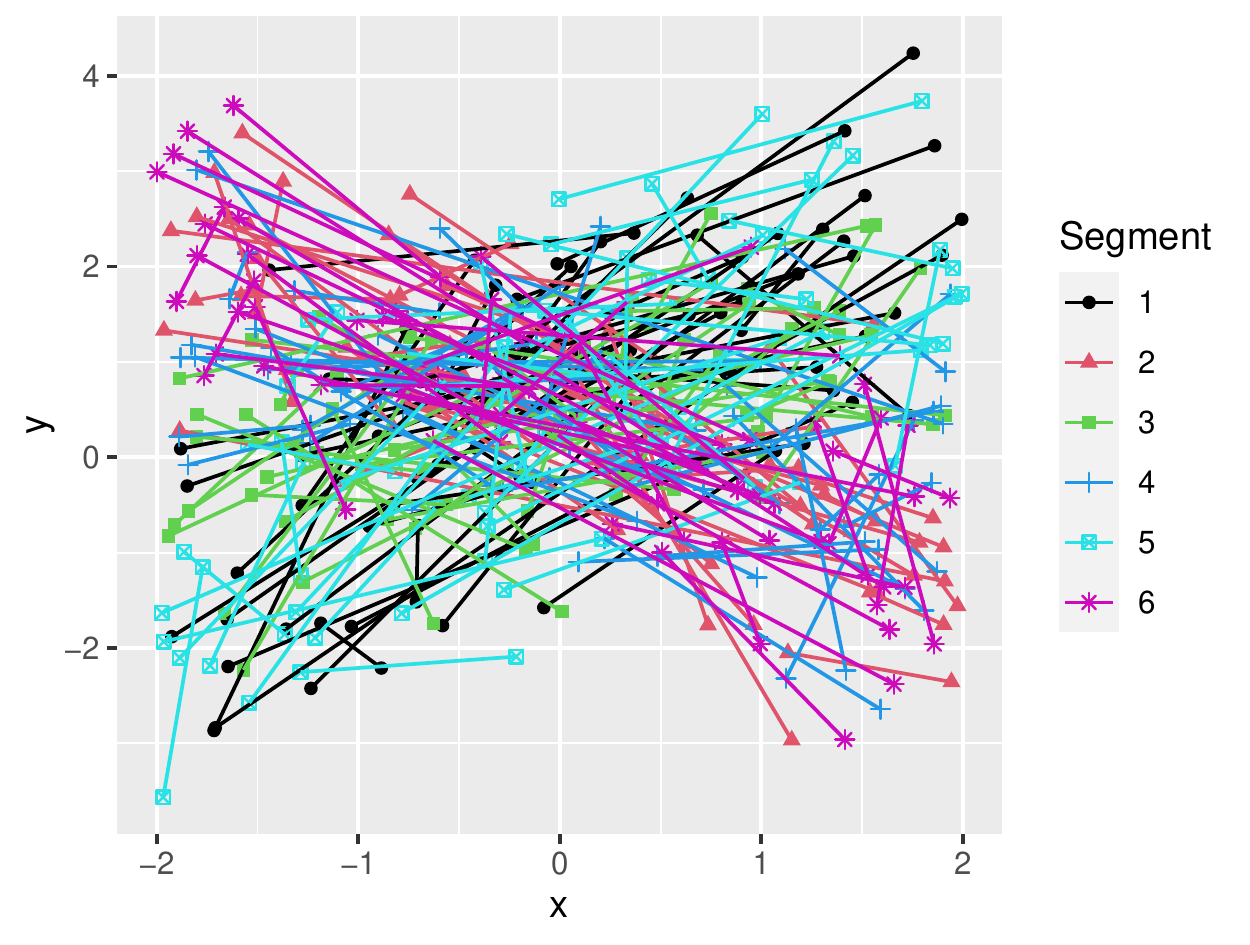}
\label{Linear5CP}
    \end{minipage}
    }
    \caption{\footnotesize {Examples of generated data in simulations. (a) to (d), data stream (in points) and $\theta(t)$ (in red lines). (e), centered absolute data stream $|Y_i - E(Y_i)|$ (in dashed line) and $\exp\{\theta(t)\}$ (in red line). (f), data grouped by $t$ (in polylines labeled by segments). (a), \textbf{S.1} (Scenario 1); (b), \textbf{S.2} (Scenario 1); (c), \textbf{S.3} (Scenario 2); (d), \textbf{S.5} (Scenario 4); (e), \textbf{S.4} (Scenario 3); (f), \textbf{S.7} (Scenario 5). }}
    \label{fig:simexample}
\end{figure}

\subsubsection*{Estimators}
In all simulations, we adopt a unified setting of truncation number $L = 25$ and the prior belief on the minimum distance between change-points $D = 15$ for NOSE. 
We run $4$ independent MCMC chains and obtain $1000$ scans in each chain thinned from a total $28000$ after a burn-in period of $8000$ iterations. 
Finally, we get $4000$ posterior samples for change-point discrimination.

Competitors vary among different settings since none of them can be applied to all the above simulation settings. 
For settings \textbf{S.1}, \textbf{S.2} and \textbf{S.3}, where the mean parameter changes, we compare with
the NOT method by \cite{baranowski2019narrowest} in package \texttt{not}, 
the TUGH method by \cite{fryzlewicz2018tail} in package \texttt{breakfast} \citep{breakfast2022}, 
the MOSUM method by \cite{eichinger2018mosum} in package \texttt{mosum} \citep{meier2021mosum}, 
the FDRSeg method by \cite{li2016fdr} in package \texttt{FDRSeg}, 
the SMUCE method by \cite{frick2014multiscale} in package \texttt{StepR}, 
the WBS method by \cite{fryzlewicz2014wild} in package \texttt{wbs},  
and the PELT method by \cite{killick2012optimal} in 
package \texttt{changepoint} \citep{killick2014changepoint}, ; 
for setting \textbf{S.4}, where the scale parameter changes, we compare with NOT, SMUCE, and PELT methods;
for setting \textbf{S.5}, where data are autocorrelated, we compare with
the WBSTS method by \cite{korkas2017multiple} in pacakge \texttt{wbsts} and 
the B-P method by \cite{bai2003computation} in package \texttt{struchchange} \citep{zeileis2002strucchange}. 
The tuning parameters for the competing methods are set as the default values in the corresponding \texttt{R} packages. 
We do not present results by Bayesian approaches such as \texttt{StepSignalMargiLike} \citep{du2016stepwise} and \texttt{solo.cp} \citep{cappello2023bayesian} here. 
We find the results of \texttt{StepSignalMargiLike} are sensitive to the choices of a maximum number of segments and cannot find a stable estimation of the number; 
\texttt{solo.cp} cannot detect most of change-points in the mean under our simulation settings. 
We conjecture the reason is that \texttt{solo.cp} identifies change-points based on the jump probability, which may fall around $1/2$ when the jump sizes are relatively small, say, our simulation settings.

\subsubsection*{Assessments and results}
Several assessments are employed to measure the accuracy of the detected number of  change-points and the accuracy of locations of estimated change-points. 
We report the frequency table for $\hat{K} - K$, the difference between the number of detected change-points and the true number of change-points to evaluate the accuracy of the detected number of change-points. 
To measure the accuracy in locations, three assessments are considered, precision, recall, and the scaled Hausdorff distance (Hausdorff). 
For all true change-points, we count one true positive (TP) if there is at least one change-point identified within a window of $10$ data points and compute the number of false positive (FP) as the number of predicted changes minus TP. 
Let $K$ be the true number of change-points. 
Then precision is computed as $\text{TP}/(\text{TP}+\text{FP})$, and recall is computed as $\text{TP}/K$. 
The scaled Hausdorff distance is computed as 
\begin{align*}
    d_H = n^{-1} E[&\max \{ \max\limits_{j=0,\cdots,K+1} \min\limits_{k=0,\cdots,\hat{K}+1} |\tau_j -\hat{\tau}_k|, \\
   &\min\limits_{k=0,\cdots,\hat{K}+1} \min\limits_{j=0,\cdots, K+1} |\hat{\tau}_k - \tau_j|  \}], 
\end{align*}
where $t_0 = \tau_0<\cdots<\tau_K<\tau_{K+1}=t_N$ and $t_0 = \hat{\tau}_0 < \hat{\tau}_1<\ldots<\hat{\tau}_{\hat{K}} < \hat{\tau}_{\hat{K}+1} = t_N$ denotes true and estimated change-points, respectively. 
The scaled Hausdorff distance takes values in $[0, 1]$ and is the smaller the better. 

From Table \ref{tab:simres} we find that NOSE outperforms in the frequency of correctly specifying the number of change-points in all settings. 
In contrast, other competitors tend to under detect the number of change-points except for the setting \textbf{S.3}, where changes take place on both the mean and variance of data. 
Although the jump sizes under these simulation settings (especially setting \textbf{S.2}) are not significant enough to make the changes be identified by eyes, 
NOSE still enjoys the highest recall in all settings, demonstrating its capability to correctly identify change-points.
These results may be evidence that the performances of segmental approaches seem to be less sensitive to small jump sizes than our non-segmental approach, particularly when the nuisance parameter (say, the scale parameter $\sigma$ in the mean-shifted model) has substantial impacts on the variation of the whole data stream. 
The precision and Hausdorff distance given by NOSE outperforms  under setting \textbf{S.3}, and are competitive under other settings. 
Note that other winners on precision and scaled Hausdorff distance actually underestimate the number of change-points, while a most parsimonious estimator usually brings higher precision and lower Hausdorff distance. 
Under setting \textbf{S.6}, NOSE correctly specifies all change-points in almost all replications, with pretty high precision and recall. 
In summary, NOSE performs to be the most competitive and robust to correctly specify the number of change-points and estimate their locations accurately.

\begin{table*}[!htb]
	\footnotesize
	\scriptsize
	\caption{\footnotesize{Results of change-points detection under settings \textbf{S.1} to \textbf{S.5} among 300 Monte Carlo replicates. The best results are bold. }}\label{tab:simres}
	\begin{tabular*}
 {\textwidth}{@{\extracolsep{\fill}}cccccccccccc @{\extracolsep{\fill}}}
		\hline
	Setting &	Method & \multicolumn{7}{c}{Frequency of $\hat{K}-K$} & Precision & Recall & $d_H \times 10^2$\\
		\hline
  & & $\le-3$ &-2 & -1 & 0 & +1 & +2 & $\ge +3$  &  &  & \\
  \hline
	\textbf{S.1}&	NOSE & 1 & 1 & 33 & $\bm{252}$ & 13 & 0 & 0 & $0.95$ & $\bm{0.94}$ & $\bm{2.1}$\\
        &    NOT & 9 & 12 & 31 & 227 & 19 & 2 & 0 & 0.93 & 0.91 & 2.4 \\
         &   SMUCE & 47 & 68 & 130 & 55 & 0 & 0 & 0 & 0.85 & 0.7 & 3.1\\
          &  WBS & 16 & 35 & 95 & 138 & 14 & 0 & 2 & 0.93 & 0.84 & 2.5\\
          &  FDRSeg & 6 & 16 &63 & 171 &  29 & 10 & 5 & 0.90 & 0.88 & 3.0\\
        &    PELT & 1 &6 & 12 & 210 & 52 & 16 & 3 & 0.91 & 0.93 & 2.8\\
        & TUGH & 0 & 0 & 1 & 217 & 51 & 14 & 5 & 0.96 & 0.93 & 2.9 \\
        & MOSUM & 3 & 3 & 72 & 181 & 41 & 0 & 0 & \textbf{0.98} & 0.93 & 2.6\\
        \hline
\textbf{S.2} & NOSE & 15 & 48 & 77 & $\bm{144}$ & 15 & 1 & 0 & 0.93 & $\bm{0.87}$ & 1.5\\
& NOT & 52 & 91 & 49 & 101 & 7 & 0 & 0 & 0.94 & 0.82 & 1.4\\
 & SMUCE & 136 & 113 & 50 & 1 & 0 & 0 & 0 & 0.86 & 0.67 & 2.1 \\
 & WBS & 68 & 120 & 74 & 38 & 0 & 0 & 0 & $0.95$ & 0.79 & $\bm{1.2}$\\
 & FDRSeg & 28 & 71 & 74 & 100 & 23 & 2 & 2 & 0.88 & 0.81 & 2.2 \\
 & PELT & 38 & 101 & 42 & 107 & 12 & 0 & 0 & 0.83 & 0.83 & 1.4 \\
 & TUGH & 12 & 37 & 53 & 129 & 48 & 17 & 4 & 0.97 & 0.84 & 2.4 \\
& MOSUM & 71 & 97 & 98 & 30 & 4 & 0 & 0 & \textbf{1} & 0.80 & \textbf{1.2}\\
 \hline
 \textbf{S.3} & NOSE & 4 & 28 & 113 & \textbf{148} & 6 & 1 & 0 & \textbf{0.90} & \textbf{0.82} & \textbf{2.9}\\
& NOT & 37 & 71 & 77 & 90 & 23 & 1 & 1 & 0.87 &0.74 & 3.2\\
&  SMUCE & 10 & 68 & 151 & 69 & 2 & 0 & 0 &  0.89 & 0.76 & 3.0 \\
& WBS & 1 & 5 & 34 & 41 & 65 & 63 & 85 & 0.64 & 0.76 & 4.8 \\
& FDRSeg & 0 & 3 & 6 & 8 & 20 & 22 & 241 & 0.47 & 0.83 & 5.7 \\
& PELT & 25 & 50 & 102 & 61 & 38 & 15 & 9 & 0.77 & 0.69 & 3.5 \\
\hline
\textbf{S.4} & NOSE & 0 & 75 & 71 & $\bm{150}$ & 4 & 0 & 0 & 0.84 & $\bm{0.75}$ & 2.3 \\
& NOT & 25 & 221 & 39 & 14 & 0 & 0 & 1 & $\bm{0.91}$ & 0.67 & 1.5 \\
& SMUCE &  40 & 211 & 49 & 0 & 0 & 0 & 0 & 0.64 & 0.64 & $\bm{1.2}$\\
& PELT & 1 & 153 & 58 & 83 & 5 & 0 & 0 & 0.88 & 0.72 & 2.0\\
\hline
\textbf{S.5} & NOSE & 0 & 0 & 98 & $\bm{154}$ & 46 & 2 & 0 & 0.85 & $\bm{0.82}$ & $2.6$\\
& WBSTS & 4 & 36 & 74 & 122 & 48 & 14 & 2 & 0.61 & 0.47 & 2.8\\ 
& B-P & 102 & 68 & 128 & 2 & 0 & 0 & 0 & $\bm{0.89}$ & 0.38 & $\bm{1.8}$\\
\hline
\textbf{S.6} & NOSE & 0 & 0 & 1 & 293 & 6 & 0 & 0 & 0.99 & 1 & 0.75\\	
\hline
	\end{tabular*}
\end{table*}

\section{Applications}\label{sec:real_data}

\subsection{DRAIP data: shifts in scale}
\label{subsec:AppDRAIP}
We report detection results on DRAIP data given by NOSE here. 
\begin{figure} [!htb]
	\centering
	\includegraphics[width = .45\textwidth]{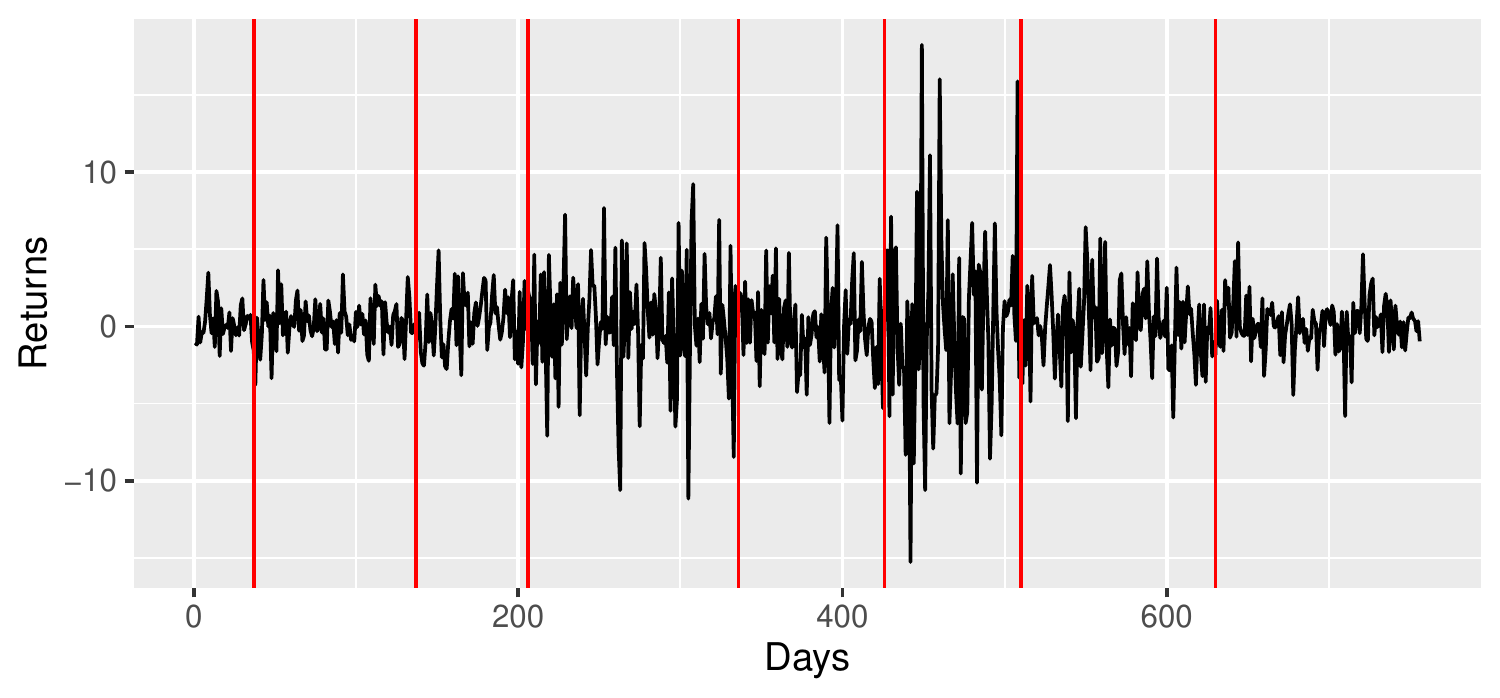}
 	\includegraphics[width = .45\textwidth]{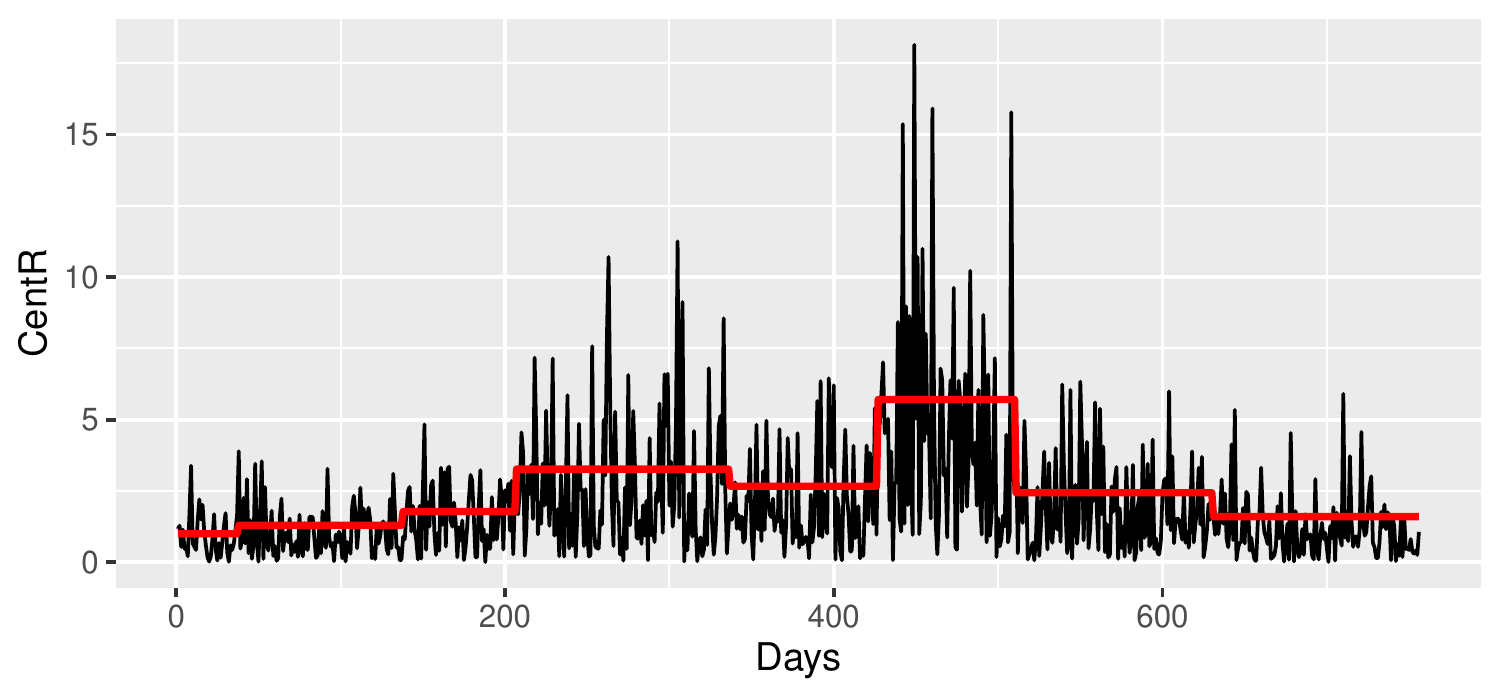}
	\caption{\footnotesize{DRAIP data and change-point detection results by NOSE. Top, original data and locations of estimated change-points (in vertical lines); bottom, centered absolute data and estimated segment-wise scale parameters (in the horizontal polyline). } }\label{fig:return-data-detection}
\end{figure}
We set $L=25$ and $D = 15$ in this case. 
As shown by Figure \ref{fig:return-data-detection}, NOSE detects $7$ change-points. 
\begin{table}[!htb]
    \centering
    \begin{tabular}{l|ccc}
    \hline
     Intervals & Estimated SD & Sample SD & Scale jump sizes\\
     \hline
   $[1, 37]$  & 1.000  & 1.173 & - \\
   $[38, 137]$ & 1.296 & 1.369 & \textbf{0.196} \\
   $[138, 206]$ & 1.778 & 1.873 & 0.504 \\
   $[207, 336]$ & 3.266 & 3.500 & 1.627\\
   $[337, 426]$ & 2.666 & 2.570 & \textbf{-0.930} \\
   $[427, 510]$ & 5.708 & 5.863 & 3.293 \\
   $[511, 630]$ & 2.437 & 2.426 & -3.437 \\
   $[631, 756]$ & 1.599 & 1.599 & -0.827 \\
   \hline
    \end{tabular}
    \caption{\footnotesize{Intervals, intervals partitioned by estimated change-points; Estimated: standard deviation estimated by NOSE; Sample SD: sample SDs on partitioned intervals; Jump sizes, jump sizes calculated from true SDs. }}
    \label{tab:ChangeScale}
\end{table}
We summarize the piecewise standard deviations and estimated standard deviations given by NOSE on the intervals partitioned by the estimated change-points as well as all jump sizes in Table \ref{tab:ChangeScale}.
The estimated scale parameters and sample standard deviations are quite close, and both suggest a shift in the estimated change-points, supporting the detection result by NOSE. 
According to Table \ref{tab:ChangeScale}, the first jump size is pretty small, and no wonder why other segmental approaches miss the point. 
Although the 4th jump size on $t=336$ is absolute enough to be observed by eyes, it is also missed by other segmental approaches. 
We conjecture the reason is that the dispersion of the data on the interval $[207, 427]$ is relatively large. 
As evidence, Figure \ref{fig:qqdensity} shows the Q-Q plot and the density curve of the data on the interval, where we find the samples on the interval are too dispersed to be Gaussian. 
It indicates that may hinder the traditional segmental approaches detecting the change-point on the interval. 
The results of simulations based on the DRAIP data are displayed in Appendix \ref{subsec:SimDRAIP}
The simulation results demonstrate the difficulty of correctly specifying all the change-points in DRAIP data. 
Even so, NOSE still outperforms other approaches.

\begin{figure}[!htb]
    \centering
    \includegraphics[width = 0.35\textwidth]{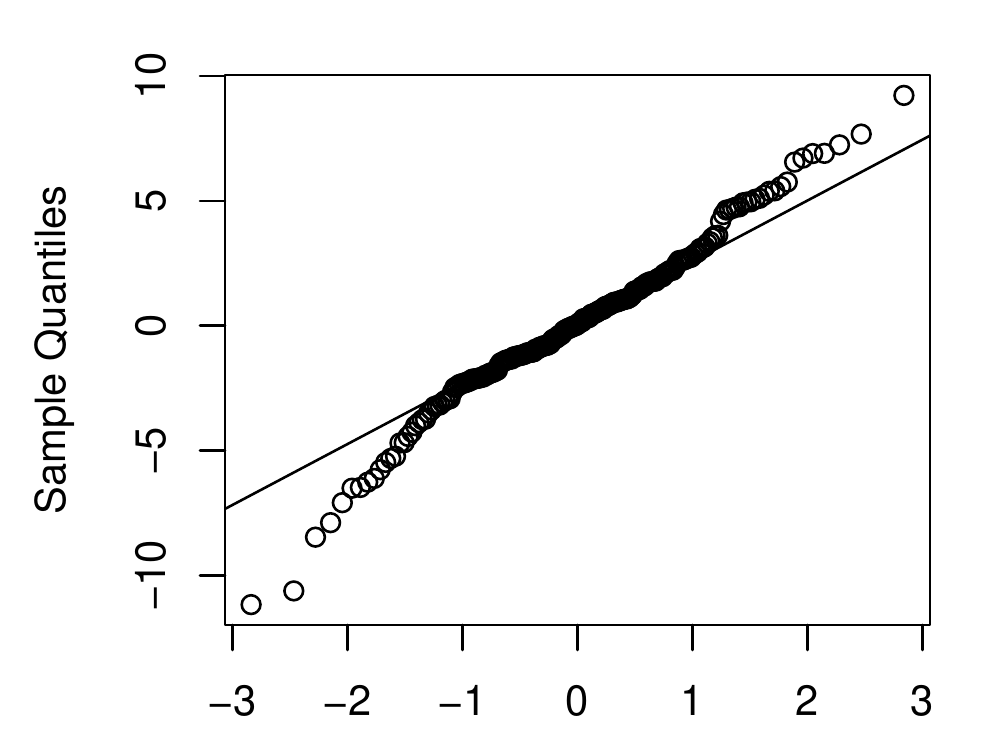}
        \includegraphics[width = 0.35\textwidth]{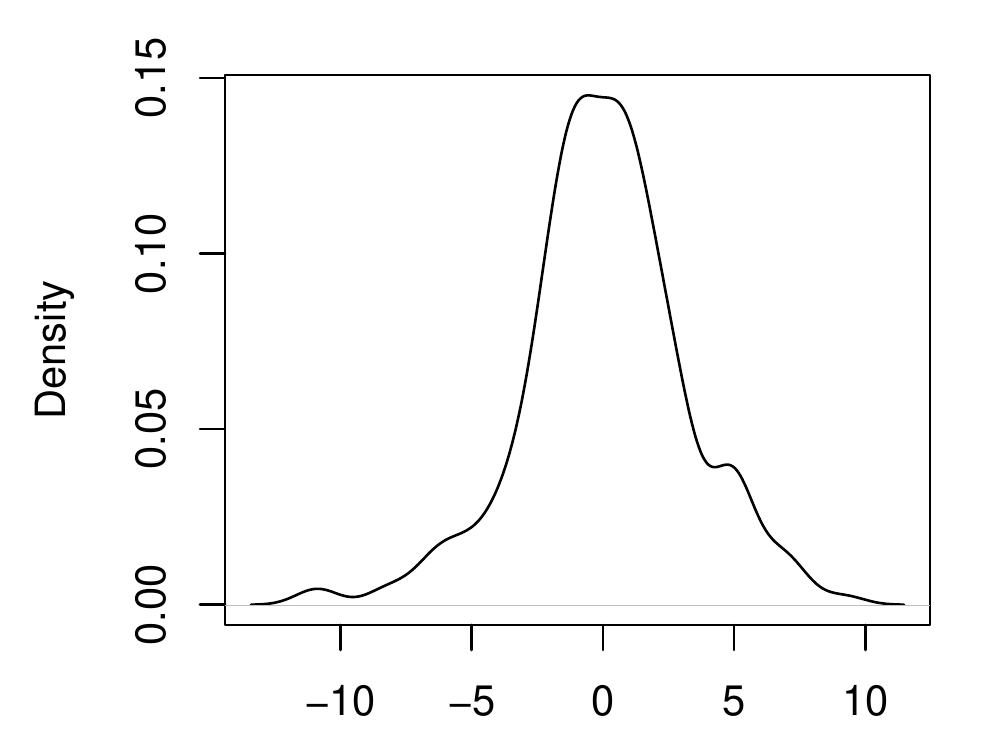}
    \caption{\footnotesize{Q-Q plot and density plot of DRAIP data on interval $[207, 427]$. Left, Q-Q plot; right, density plot. }}
    \label{fig:qqdensity}
\end{figure}

\subsection{ACGH data: shifts in mean}
\label{subsec:APPACGH}
In the second example, we analyze the public dataset of DNA copy numbers using ACGH
for 43 different individuals with a bladder tumor \citep{stransky2006regional}, 
which is available in R package \texttt{ecp} \citep{james2015ecp}. 
For each individual, the copy number is recorded on  2215 locations. 
We aim to detect the changes in the mean of the copy number. 
Hence we employ NOSE for Gaussian mean changes under scenario $(i)$. 
As the number of change-points is usually considered to be quite large, we set $L = 55$ to incorporate sufficiently many change-points. 
The prior belief on the minimum distance between change-points is set as $D = 15$. 
We display the analysis result of the 37th individual in this article. 

\begin{figure*}[!htb]
    \centering
    \subfigure[]{
    \begin{minipage}[t]{0.4\linewidth}
      \centering
\includegraphics[height = .65\textwidth]{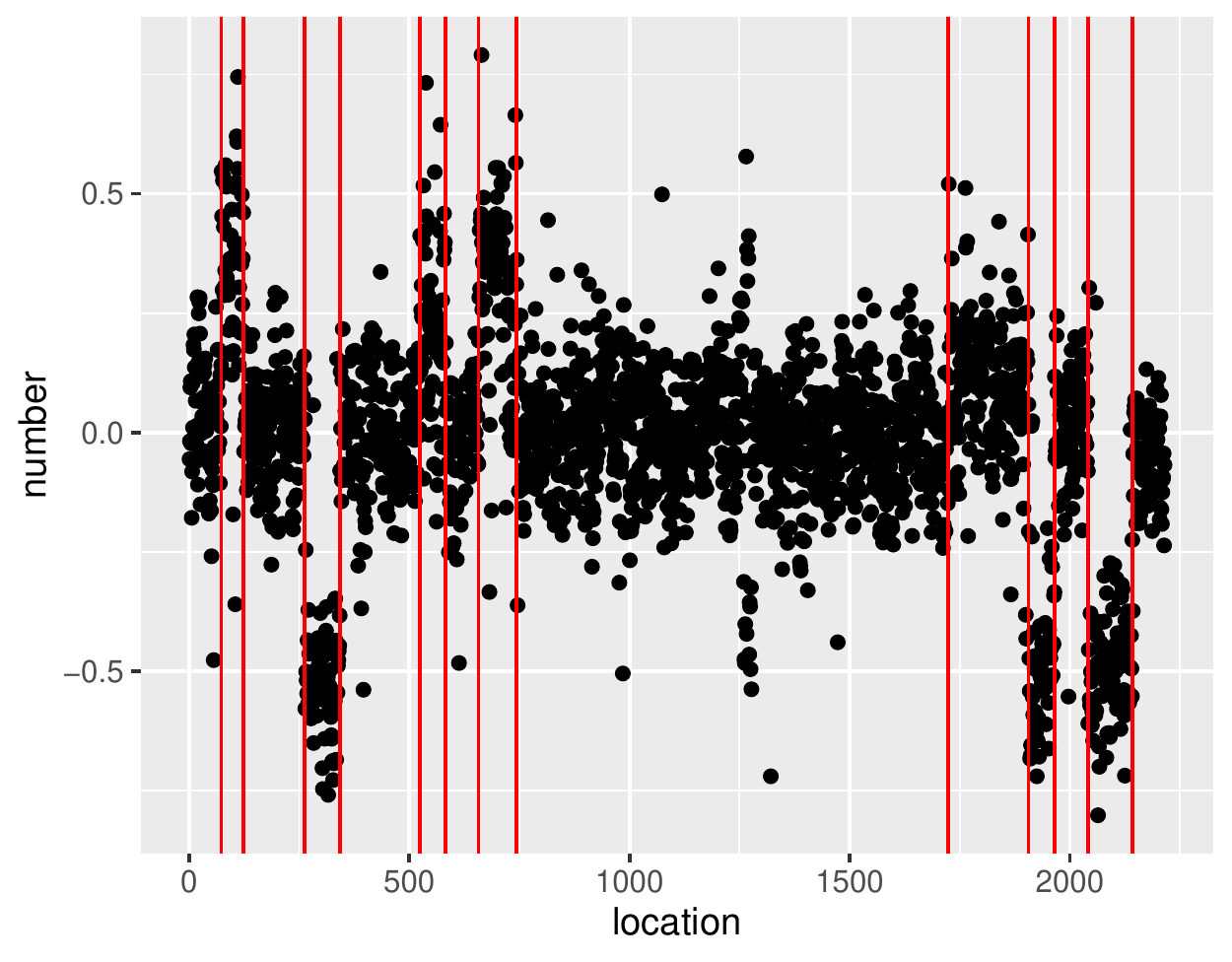}
\label{ACGH_NOSE}
    \end{minipage}
    }
\subfigure[]{
    \begin{minipage}[t]{0.4\linewidth}
      \centering
\includegraphics[height = .65\textwidth]{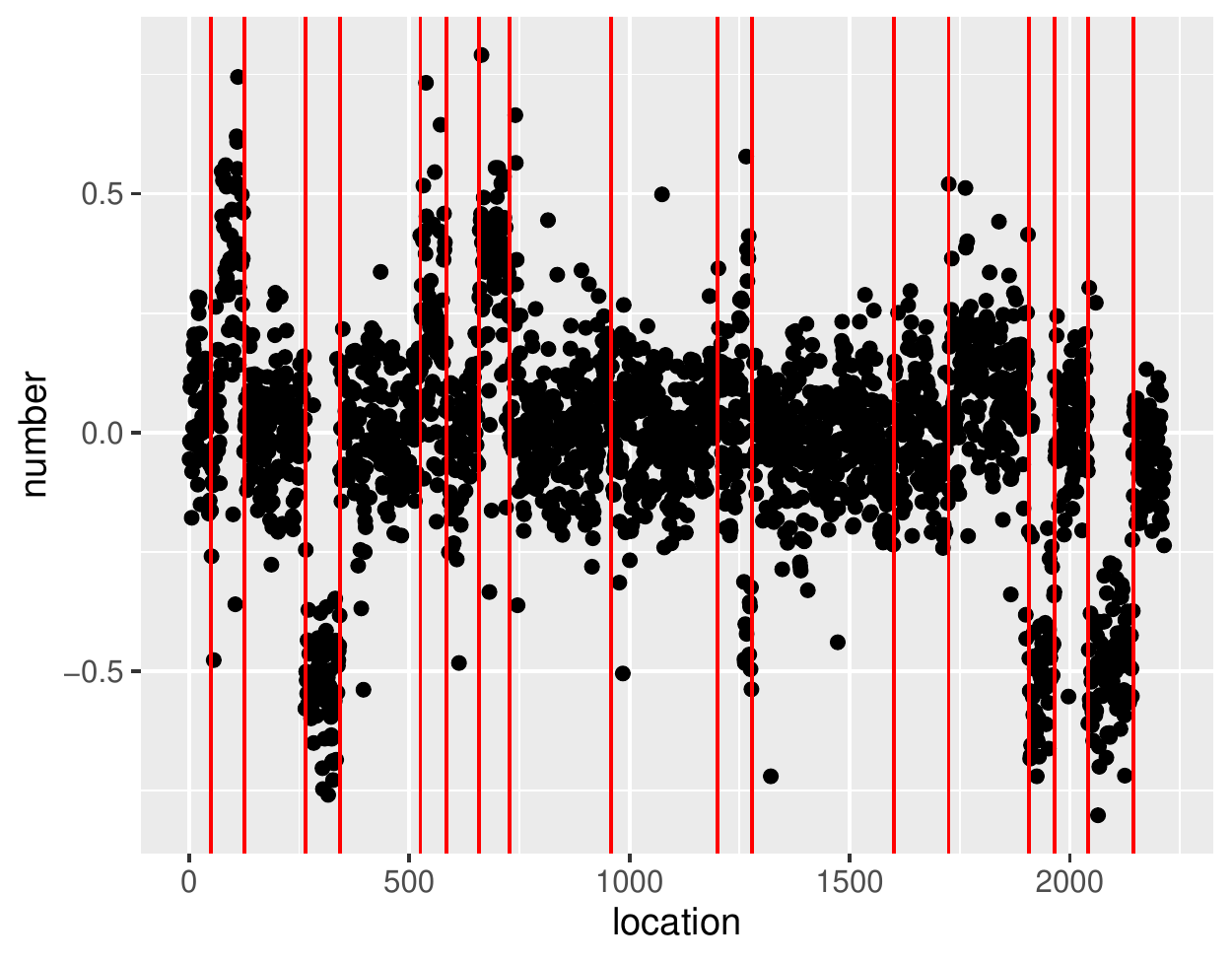}
\label{ACGH_HSMUCE}
    \end{minipage}
    }
    \subfigure[]{
    \begin{minipage}[t]{0.4\linewidth}
      \centering
\includegraphics[height = .65\textwidth]{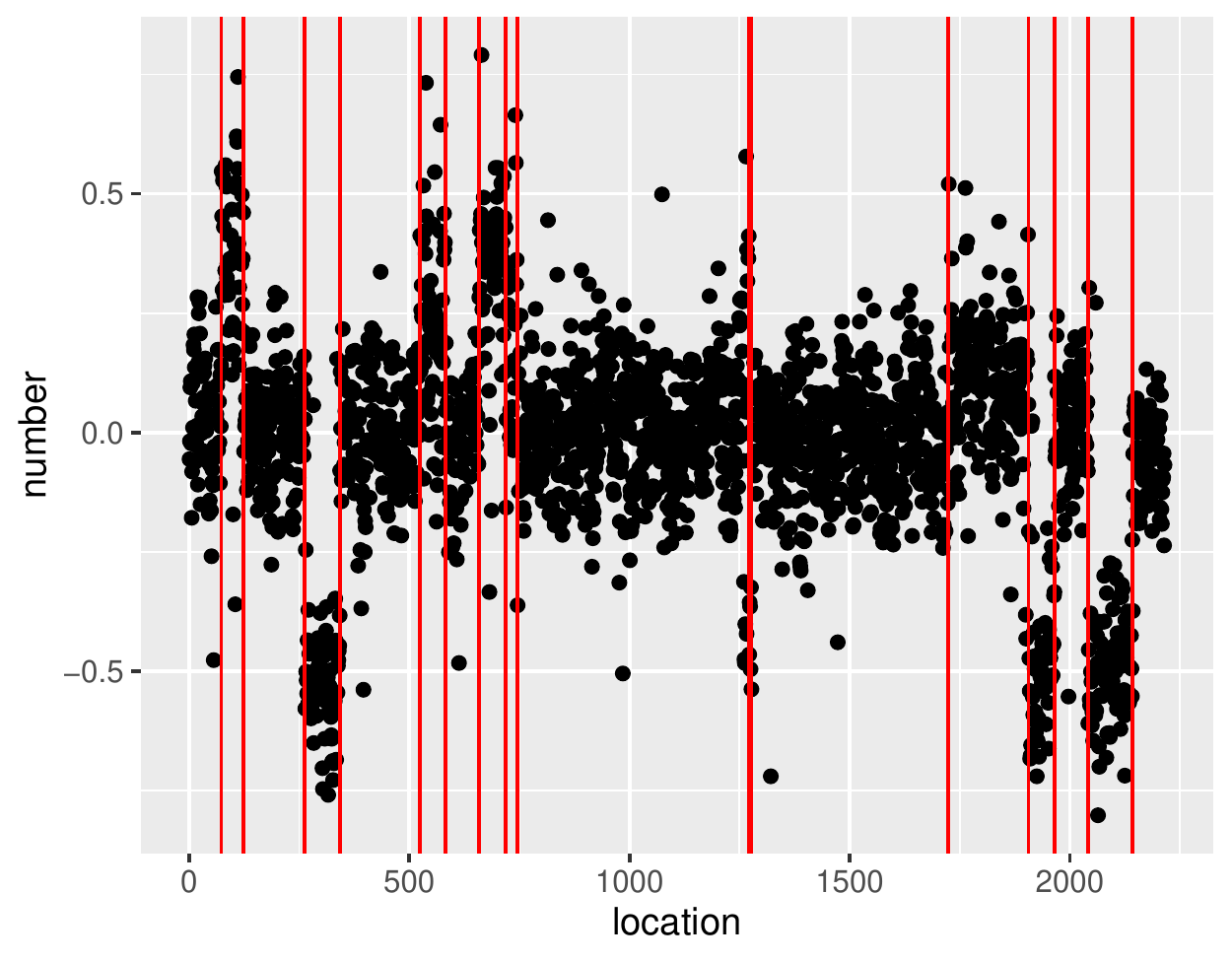}
\label{ACGH_NOT}
    \end{minipage}
    }
\subfigure[]{
    \begin{minipage}[t]{0.4\linewidth}
      \centering
\includegraphics[height = .65\textwidth]{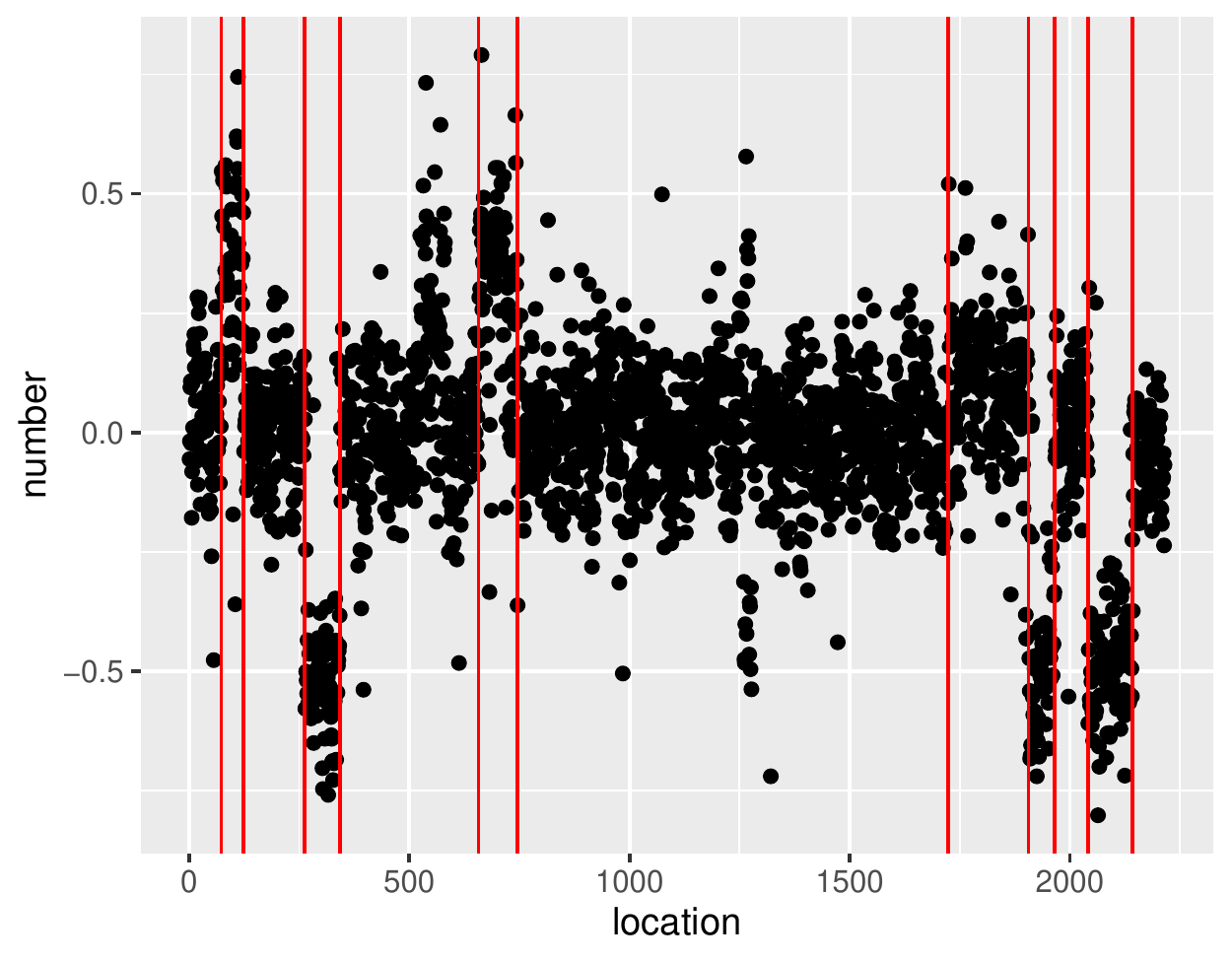}
\label{ACGH_FPOP}
    \end{minipage}
    }
    \caption{\footnotesize{Plot of ACGH data (in black points) and estimated locations of change-points (in red vertical lines). (a), NOSE; (b), HSMUCE; (c), NOT; (d), R-FPOP.  }}
    \label{fig:ACGH}
\end{figure*}

We display detection results of NOSE, HSMUCE \citep{pein2017heterogeneous} and \texttt{NOT} in Figures \ref{ACGH_NOSE}, \ref{ACGH_HSMUCE} and \ref{ACGH_NOT}, where they detect 13, 16, and 15 change-points, respectively. 
Despite some similarities among them, HSMUCE and \texttt{NOT} are more likely to create short 
segments gathering several data points that are far away from the means of adjacent segments. 
We conjecture the points in these short segments are outliers. 
To eliminate the influence of outliers,  we employ the outlier-robust R-FPOP method \citep{fearnhead2019changepoint} equipped with the Huber loss and penalized value $1.345$ as default; see Figure \ref{ACGH_FPOP}. 
We find the data points in those short segments divided by  HSMUCE and \texttt{NOT} are treated as outliers by R-FPOP. 
By comparison, NOSE and R-FPOP produce almost the same segmentation, with the only difference being the segment $(524, 583)$, where NOSE creates a new segment while R-FPOP does not. 
Since this segment contains $60$ data points, we feel that it is more appropriate to partition these points into a new segment rather than identifying them as outliers. 

We generate simulated data from the estimation results by NOSE in Figure \ref{ACGH_NOSE}. 
Since the simulated data are exactly Gaussian without outliers, the results of NOSE, HSMUCE, and R-FPOP are stable and similar to each other, while \texttt{NOT} slightly over-detects the change-points. 
Details are deferred to Appendix \ref{subsec:SimACGH}. 

\subsection{US age-specific fertility rate (ASFR) data: structural changes in linear models}
\label{subsec:APPASFR}
The declining birth rates in many developed countries arouses much interest to the analysis of the annual Age-Specific Fertility Rate (ASFR).
Given the year $t$, let $B_{tj}$ be the number of births during the year to females of a specified age $j$, and $N_{tj}$ be the number of females of the age $j$ in that reference year. 
In year $t$, the ASFR $y_{tj}$ is defined as the ratio between $B_{tj}$ and $N_{tj}$. 
We collect ASFR data in the US from 1940 to 2021 at ages 22 to 35, the age period which covers the age with the highest ASFR. 
Then totally we obtain 1134 responses $y_{tj}$. 

The relationship between the ASFR and specific ages from 22 to 35 seems to be linear. 
Hence, we consider a linear model with changes in the regression coefficient to characterize their association. 
We consider following linear models 
$$
y_{tj} = \beta_0 + \theta(t)X_{tj} + \epsilon_{tj},~ t=1, \ldots, 81, ~j=1, \ldots, 14,  
$$
where the regressor $X_{\cdot j} = 21+j$, the regression coefficient $\theta(t)$ may change along with time $t$, $\beta_0$ is a fixed intercept and $\epsilon_{ts} \sim N(0, \sigma^2)$ are i.i.d. model errors. 
We apply NOSE to detect  changes of $\theta(t)$, where the state of data is set to be the year $t$. 
We set $L=25$ and the minimum distance threshold $D = 15$.

\begin{figure} [!htb]
	\centering
	\subfigure[]{
	\begin{minipage}[t]{\linewidth}
	\centering
	\includegraphics[width = 0.7\textwidth]{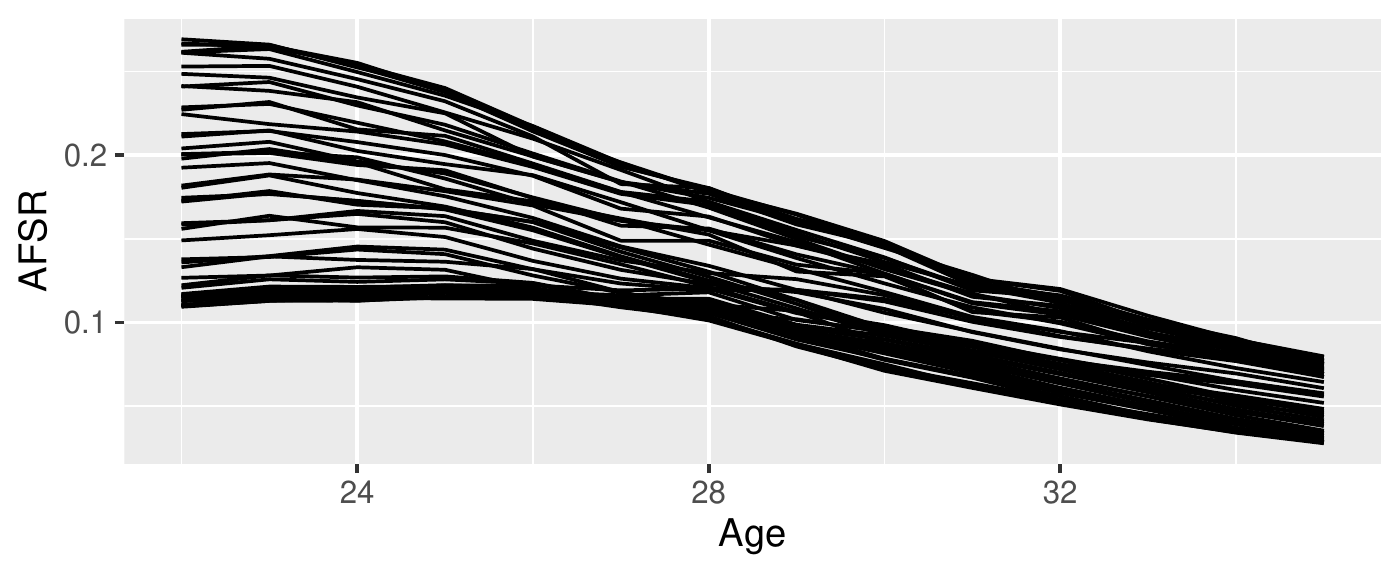}
	\end{minipage}
	}	
	\subfigure[]{
	\begin{minipage}[t]{\linewidth}
	\centering
	\includegraphics[width = 0.7\textwidth]{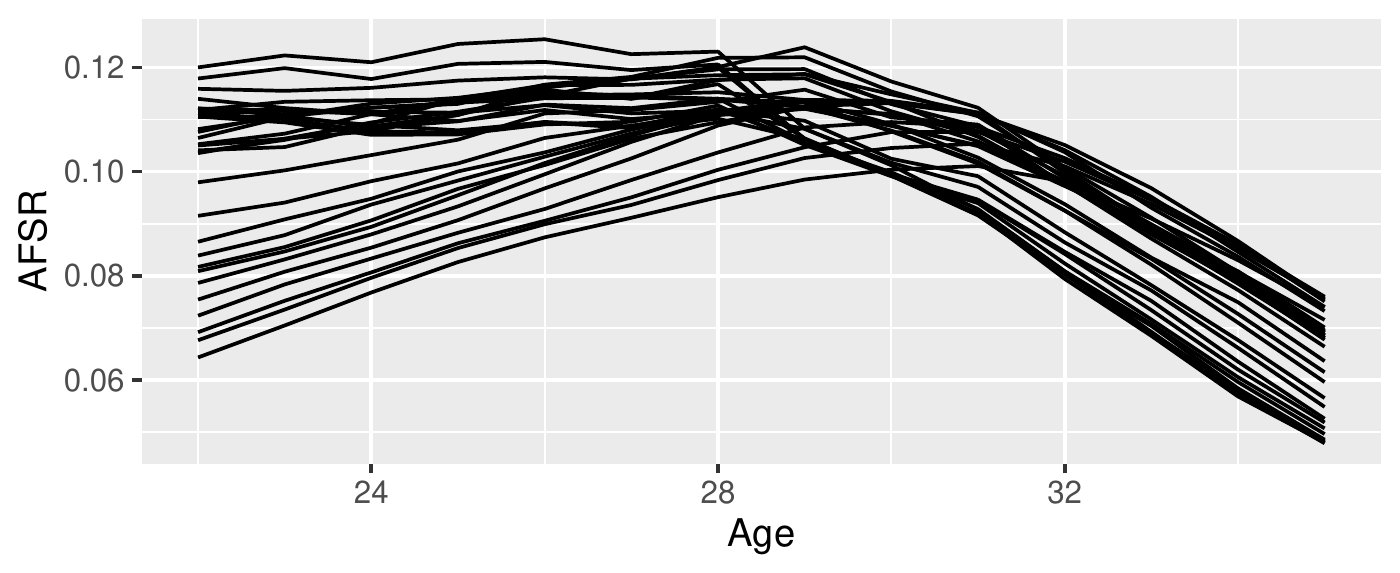}

	\end{minipage}
	}
	 \caption{\footnotesize{Visualization of the pre- and post-change-points ASFR data in US. (a), relationship between age and ASFR before  year 1992; (b), the relationship between age and ASFR after year 1992.}}
	 \label{fig:change_ASFR}
\end{figure}

Only one change-point is detected by NOSE at $t=1992$. 
To understand the effect of the change-point, we plot the curves of ASFR versus age before and after 1992 in Figure \ref{fig:change_ASFR}. 
From the figure, we can clearly see that before the change point, the ASFR decreases almost linearly with age, so that the ASFR is highest at age 22.
However, after the change point, the association between ASFR and age is non-linear and even non-monotonic, with ASFR first increasing and peaking at age 29 and then decreasing. 

\subsection{House prices in London Borough of Newham: structural changes in AR(1) models}
\label{subsec:APPHouseInd}
We further explore a real dataset, the average monthly property price $P_t$ in the London Borough of Newham. 
We take the average of all properties and select the data recorded from January 2010 to November 2020 and we totally have 131 observations. 
This dataset was once analyzed by \cite{fryzlewicz2021narrowest} to identify the shortest interval of change-points under an AR(1) model. 
We adopt the AR(1) model $P_t = \theta(t) P_{t-1} + \theta_0 + \epsilon_t$, where the autocorrelation coefficient $\theta(t)$ is treated as the global parameter that may change, the intercept $\theta_0$ is fixed, and $\epsilon_t \sim N(0, \sigma^2)$ are independent model errors. 
We set $L=25$ and $D = 15$. 

As shown in Figure \ref{fig:logpriceCP}, NOSE detects 1 change-point locating in Oct 2016 (location 82). 
The date of change-point is close to the beginning of the vote of Britain’s EU membership referendum, indicating that the structural change may be caused by the event. 
The WBSTS method cannot detect change-point after processing; 
the B-P method provides a similar result of change-point detection, where the estimated location is 79. 
Meanwhile, the estimated confidence interval given by R package \texttt{nsp} \citep{fryzlewicz2021narrowest} is (24, 97), which covers the change-point estimated by NOSE.

\begin{figure} [!ht]
	\centering
	\includegraphics[width = .6\textwidth]{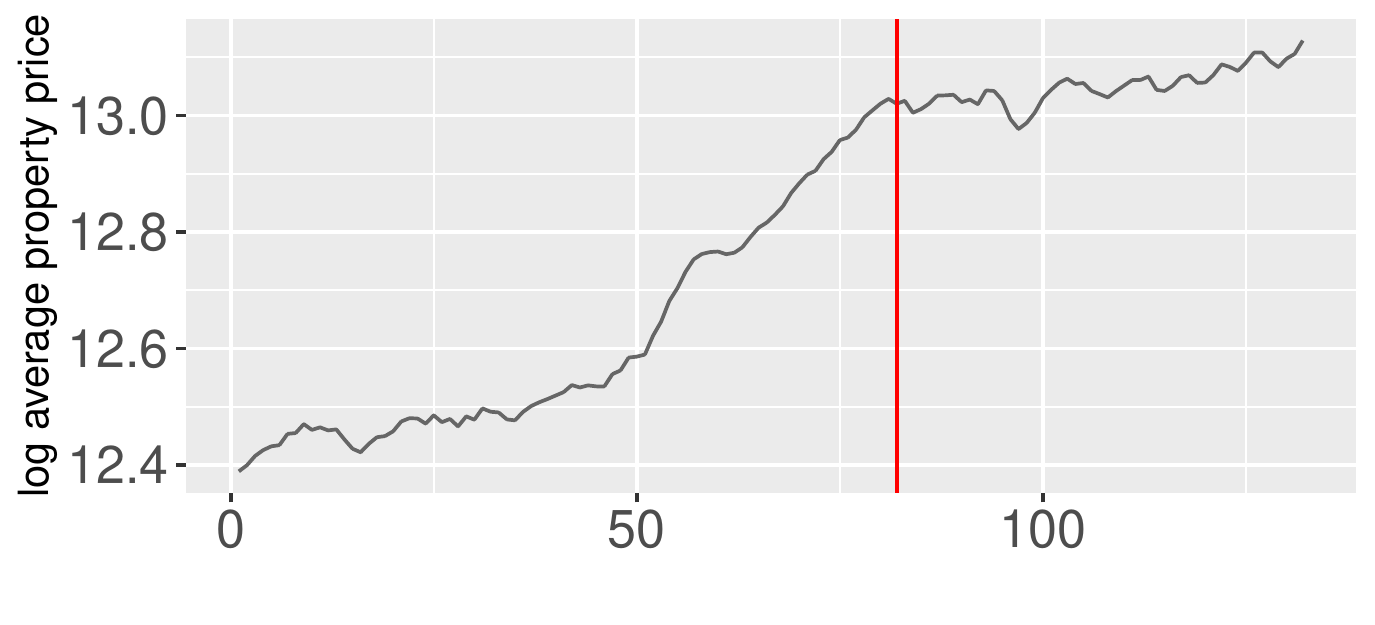}
	\caption{\footnotesize{House prices in London Borough of Newham and locations of estimated change-points given by NOSE (the red line). } }\label{fig:logpriceCP}
\end{figure}

\section{Discussion}\label{sec:discussion}
The proposed NOSE methodology in this article has two pieces of uniqueness. 

i.) NOSE models the entire abrupt change process directly through $\theta(t) $ ($\equiv \theta$) rather than the aggregating all sets of segment parameters in prevailing methods. 
In this sense, NOSE can be viewed as an infinite-dimensional extension of \texttt{StepSignalMargiLike} \citep{du2016stepwise}, which represents the abrupt change scheme through a finite-dimensional vector $\bm{\theta}_{1:m}$ with each entry being the latent feature of a segment. 
Their $m$ is the maximum number of segments and needs to be prespecified. 
Thus, any misspecification of $m$ is risky to their results of change-point detection. 
In contrast, the atomic expression of $\theta(t)$ in NOSE looks as if a much ``denser" segmentation than \texttt{StepSignalMargiLike} so that $m$ can go to infinity. 
Hence, NOSE is exempted from the sensitivity of the upper bound of the number of segments. 

ii.) NOSE may be the first approach that deals with the sparsity of the vector of \textit{jump heights} (vertical), unlike existing penalized approaches that focus on the  sparsity of the vector of \textit{jump locations} (horizontal). 
In detail, NOSE identifies change-points by the posterior estimates ($\zeta_i$) of jump heights/sizes ($d_i$) on states ($i$), where any  non-negligible jump height/size indicates a change. 
In the broad sense, NOSE may be viewed as a vertical extension of SMUCE \citep{frick2014multiscale} in searching for sparse solutions under a high-dimensional regression setting. 
Different sparsity reviews lead to different theoretical properties:
SMUCE reaches minimaxity in estimation of change locations (up to a logarithm) and consistency of estimation of the number of change-points under the frequentist paradigm; 
NOSE obtains the posterior minimax optimality in recovering the jump height vector and posterior consistency of \textit{both} the number and the locations of change-points under the Bayesian paradigm. 

We may try to explain the success of NOSE from the perspective of cohesion and repulsion in clustering \citep{natarajan2023cohesion}. 
To some extent, \textit{change-point detection may be viewed as an ordered clustering task on sequential data. 
Those data points within the same segment can be viewed as a cluster}. 
Quoting \cite{natarajan2023cohesion}, ``clusters are composed of objects which have small dissimilarities among themselves (cohesion) and similar dissimilarities to observations in other clusters (repulsion)". 
Intuitively, jump size may be viewed as a metric of dissimilarity between data points. 
In our approach, the nearly black jump size vector indicates that there are no dissimilarities with-in a cluster but significant dissimilarities across different clusters, leading to an ideal clustering under the cohesion-repulsion principle.

\newpage

\begin{center}
    {\LARGE \bf Appendix }
\end{center}
\renewcommand*{\thesection}{A.\arabic{section}}
\renewcommand*{\thesubsection}{A.\arabic{section}.\arabic{subsection}}
\setcounter{section}{0}

\spacingset{1.7} % DON'T change the spacing

\section{Proofs}\label{sec:proof}
\subsection{Proof of Theorem 1}
\label{subsec:proofTheominmax}
Before proving Theorem \ref{minimax}, the necessary propositions and a lemma are given as follows.
\begin{proposition}[Gaussian sequence prior]
    Let $S \subset \{1, \ldots, p\}$ be the non-zero coordinates of the jump size vector $\bm{d}$ of  cardinality $|S|$. 
Let $\bm{d}_S$ be the set of non-zero values $\{d_i, i\in S\}$. 
Let $\pi_{L_n}$ be a prior selects a dimension $s$ from $\{0, 1, \ldots, L\}$. 
Under the priors for $\bm{\xi}$ and $\bm{h}$ in \eqref{prior:xi} and \eqref{prior:h}, for a fixed truncation number $L$, the prior for $\bm{d}$ with non-zero coordinates $S$ is in the form of 
\begin{align}\label{prior:d}
\pi(\bm{d}) \propto \frac{1}{\binom{L_n}{|S|}}\pi_{L_n}(|S|)g_S(\bm{d}_S)\delta_{0}(\bm{d}_{S^c}). 
\end{align}
\end{proposition}
\begin{proof}
Drawing a sample of $\bm{d}$, with non-zero coordinates set $S$ from priors \eqref{prior:xi} and \eqref{prior:h} can be divided into the following steps
    \begin{enumerate}
    \item Draw $\bm{\xi}$ so that $S \subset \bm{\xi}_{1:L_n}$. 
    \item Given $\xi_\ell$, draw indicators $Z_\ell$ so that $\sum_{\ell=1}^{L_n} Z_\ell = |S|$ and assign those non-zero indicators to locations $S$. 
    \item Given the non-zero indicators $Z_\ell$, draw $\bm{d}_{S}$ from the slab term of $h_\ell$ and assign zeros to other coordinates. 
\end{enumerate}
In terms of step 1, recall that a draw of $\bm{\xi}_{1:L}$ is a draw of $L$ elements of $\{1, \ldots, p\}$ without replacement. 
Hence we have
\begin{align*}
   Pr\{S \subset \bm{\xi}_{1:L}\} = \left\{\binom{p}{L_n} \binom{L_n}{|S|}\right\}^{-1}.
\end{align*}
In step 2, we immediately have 
$$
\pi_{L_n}(|S|) = Pr\left\{|\bm{Z}| = |S|\right\}.
$$
In step 3, we immediately have that 
$$
g_S(\bm{d}_S) = \prod_{\ell \in S} F_0
$$
becomes the product of Laplace density. 
Then the prior form in \eqref{prior:d} is obtained as the product of the above terms. 
\end{proof}

\begin{remark}
Note that in the limiting case $L_n=p$, the prior in form \eqref{prior:d} takes the same form as the prior (1.2) in \cite{castillo2015bayesian}. 
Similarly, the dimension prior $\pi_{L_n}$ in \eqref{prior:d} plays the same role of $\pi_p$ in their seminal work and replaces $\pi_p$.
Consequently, it suffices to study the properties of $\pi_{L_n}(s)$ with $L_n \to \infty$, and definitely, $p=(n-1) \to \infty$. 
\end{remark}

In terms of the properties of dimension prior $\pi_{L_n}$, we shall show that $\pi_{L_n}$ has an exponential decrease by appropriate selection of the hyperparameters $(a, b)$ in the Gamma prior for $\alpha$, given that $L_n$ is sufficiently large. 
We starts from the following lemma of Poisson approximation.

\begin{lemma}
        [Serfling's Poissson approximation] \label{Serfling}
Let $X_1, \ldots, X_n$ be (possibly dependent) Bernoulli random variables with $p_1 = Pr\{X_1 = 1\}$ and 
$$
p_i = Pr\{X_i=1|\mathcal{F}_{i-1}\}, 
$$
where $\mathcal{F}_i$ denotes the $\sigma$-field generated by $X_1, \ldots, X_i$. 
Let $W_n = \sum_{i=1}^n X_i$  and $Y$ be Poisson with mean $\lambda = \sum_{i=1}^n E(p_i)$. 
Then
$$
\frac{1}{2}\sum_{k=1}^n|Pr\{W_n = k\} - P\{Y= k\}| \le \sum_{i=1}^n E(p_i^2) + \sum_{i=1}^n E|p_i - E(p_i)|. 
$$
\end{lemma}
The result of Lemma \ref{Serfling} will be used to prove the following proposition. 
Our assertions are given under any fixed $L_n$. 

\begin{proposition}[Exponential decrease] \label{prop: selection}
Let $a=c_1L_n^{-c_3}, b=c_2L_n^{c_4}$ for some constants $c_1, c_2 >0$ and $c_3 > c_4+1 \ge 2$ in prior \eqref{prior:h}. 
The following assertion holds as $n, L_n\to \infty$. 

There exists a constant $C_0 \in (0, 1)$, 
\begin{align}\label{decreaseCond}
   \pi_{L_n}(s) \le C_0 \pi_{L_n}(s-1), for~ s=1, \ldots, L_n. 
\end{align}
\end{proposition}
\begin{proof}

We first determine the prior $\pi_{L_n}$ in Step 2. 
Obviously, we have
$$
\pi_{L_n}(s) = \int_{\mathbb{R}} Pr\{|\bm{d}| = s | \alpha\} \pi(\alpha) d\alpha.  
$$
Hence we study the conditional probability $Pr\{|d| = s | \alpha\}$ first, or equivalently, $Pr\left\{|\bm{Z}| = s | \alpha\right\}$.

Note that $\eta_\ell$ have a Markov structure and for $\ell>1$, 
$$
p_\ell^* = Pr\{Z_\ell=1|\mathcal{F}_{\ell-1}\} = Pr\{Z_\ell=1|\eta_{\ell-1}\} = \eta_\ell|\eta_{\ell-1}. 
$$ 
Following \citet[Eq. 14]{teh2007stick}, given fixed $\alpha$,  for $\ell>1$,
$$
f(\eta_\ell|\eta_{\ell-1}) = \alpha \eta_{\ell-1}^{-\alpha}\eta_\ell^{\alpha-1}I(0<\eta_\ell < \eta_{\ell-1}). 
$$
To avoid confusion, we denote $p_1^* = p_1$. 
Then, one drives 
\begin{align*}
    E(p_1^*) &= \int \alpha \eta_1^{\alpha-1}d\eta_1 = \frac{\alpha}{\alpha+1}, \\
    E(p_2^*) &= \int_0^1 \int_0^{\eta_1} \alpha \eta_1^{-\alpha}\eta_2^{\alpha}d\eta_1 d\eta_2 = \left(\frac{\alpha}{\alpha+1}\right)^2, \\
   &  \vdots\\
   E(p_{L_n}^*) &= \int_{0<\eta_L<\cdots<\eta_1<1} \alpha^{L_n} \eta_{L_n}^{\alpha}\prod_{\ell=1}^{L_n-1} \eta_\ell^{-1} d\eta_1\ldots d\eta_{L_n} \\
   &= \left(\frac{\alpha}{\alpha+1}\right)^{L_n}. 
\end{align*}
Similarly, we have 
\begin{align*}
 E(p_1^{*2}) = \frac{\alpha}{\alpha+2}; ~ E(p_\ell^{*2}) = \left(\frac{\alpha}{\alpha+2}\right)^\ell, ~\ell>1. 
\end{align*}

We hence obtain the Poisson approximation of the probability $Pr\{|\bm{d}|=s|\alpha\}$, denoted as $\pi_{\alpha, L_n}^0$. 
As $n, L_n\to \infty$, $\sum_{\ell\ge 1} E(p_\ell^*) = \alpha$. 
We have $\pi_{\alpha, \infty}^0 = \pi_{\alpha}^0 = \text{Pois}(\alpha)$. 

By integrating out $\alpha$ under the Gamma prior in \eqref{prior:h} we obtain the approximated form for $\pi_{L_n}$, denoted as $\pi^0$. 
With the hyperprior $\text{Gamma}(a, b)$, $\pi^0$ becomes a truncated negative binomial distribution 
$$
\pi^0(s) \propto \frac{\Gamma(s+a)}{s!\Gamma(a)}\left(\frac{b}{b+1}\right)^s\left(\frac{1}{b+1}\right)^a, s=0, 1, 2, \ldots, L_n. 
$$

For some $(a, b)$ fixed with given $L_n$, 
$$
\frac{\pi^0(s+1)}{\pi^0(s)} = \left\{1-\frac{1-a}{s+1}\right\}\left(\frac{b}{b+1}\right), s = 0, \ldots, L_n-1. 
$$
And hence it naturally satisfies assertion \eqref{decreaseCond} with $C_0 = b/(b+1)$. 

By the fact that $\prod_{m=2}^M(1-1/m) = M^{-1}$, with $b = c_2L_n^{c_4}$ with $c_4 \ge 1$
we have 
$$
\pi^0(s) \ge Q_{n, s}^{-1} s^{-1}, ~s\ge 1, 
$$
where $Q_{n, s}$ acting as the denominator related to $L_n$ to guarantee that $\sum_{s=1}^{L_n} \pi^0(s)= 1$. 
Since $\log n \le \sum_{i=1}^n i^{-1} \le 1 + \log n$, we have 
\begin{align}
    \label{pi0s}
    \pi^0(s) \ge \frac{Q_0}{s(1+\log L_n)}
\end{align}
for some finite constant $Q_0$ unrelated to $s$.

We then show that the approximated distribution $\pi^0$ is sufficiently close to the true $\pi_{L_n}$ and hence assertion \eqref{decreaseCond} holds for $\pi_{L_n}$. 
By Jensen's inequality, for $\ell\ge 1$, 
\begin{align*}
E|p_\ell^* - E(p_\ell^*)| &\le \sqrt{\text{Var}(p_\ell^*)} \\
& = \sqrt{\left(\frac{\alpha}{\alpha+2}\right)^\ell - \left(\frac{\alpha}{\alpha+1}\right)^{2\ell}}\\
& < \sqrt{\ell \left(\frac{\alpha}{(\alpha+1)^2 (\alpha+2)}\right) 
\left(\frac{\alpha}{\alpha+2}\right)^{\ell}}\\
&< \ell\sqrt{\left(\frac{\alpha}{(\alpha+1)^2 (\alpha+2)}\right) 
\left(\frac{\alpha}{\alpha+2}\right)^{\ell}}
\end{align*}

Hence we have
\begin{align*}
    \sum_{\ell=1}^{L_n} E|p_\ell^* - E(p_\ell^*)| &< \sum_{\ell=1}^\infty E|p_\ell^* - E(p_\ell^*)| \\
    &<\frac{\alpha}{(\alpha+1)(\sqrt{\alpha+2} -\sqrt{\alpha})^2}\\
    &< \frac{\alpha}{(\alpha+1)^2} 
\end{align*}

Consequently, by Lemma \ref{Serfling}, for any $s=0, 1, \ldots, L$, we have 
$$
|Pr\{|d| = s | \alpha\} - \pi_{\alpha, L}^0(s)| \le \left(1+\frac{1}{(\alpha+1)^2}\right)\alpha < 2\alpha
$$
The RHS of the above inequality is obtained by taking $L \to \infty$ on the RHS of Lemma \eqref{Serfling}. 

Finally, we have 
$$
|\pi_{L_n}(s) - \pi^0(s)| = \int_0^{+\infty} |Pr\{|d| = s | \alpha\} - \pi_{\alpha, L}^0(s)| \pi(\alpha)d\alpha. 
$$
Again by Jensen's inequality and \eqref{pi0s}, for $a = c_1 L_n^{-c_3}$,  $b=c_2L^{c_4}$, and $c_3 > c_4+1$, we obtain 
$$
|\pi_{L_n}(s) - \pi^0(s)| \le 2ab = o[\min\limits_{s\ge 0}\pi^0(s)]. 
$$
Consequently, for all $s$, 
$$
\lim\limits_{L_n \to \infty} \frac{\pi_{L_n}(s+1)}{\pi_{L_n}(s)} = \frac{\pi^0(s+1)}{\pi^0(s)}. 
$$
Since $b/(b+1)$ is bounded away from zero, for sufficiently large $L_n$, assertion \eqref{decreaseCond} always holds. 
\end{proof}

Since Theorem \ref{minimax} gives the same assertion as  \citet[Thereom 2, recovery]{castillo2012needles}, we only need to check their conditions. 
\begin{proof}

    For the support of non-zero coordinates of $\bm{d}$, the density $g_S = \prod_{s=1}^{|S|}F_0$, which is product of $|S|$ univariate densities.
    Meanwhile, the Laplace density naturally satisfies condition (2.3) in \cite{castillo2012needles} with a finite second moment. 
    The assertion \eqref{decreaseCond} implies that the prior $\pi_{L_n}$ on dimension has a strict exponential decrease. 
    Furthermore, assertion \eqref{pi0s} implies that 
    $$
    K_n\log(L_n/K_n) \ge M \log(\frac{1}{\pi_{L_n}(K_n)})
    $$
    for a universal constant $M$. 
    Then all conditions required by \citet[Thereom 2, recovery]{castillo2012needles} are satisfied. 
\end{proof}

\subsection{Proof of Theorem \ref{theo:NoSupset}}
\label{subsec:proofnosuperset}
We introduce some necessary notations and present some auxiliary lemmas before proving Theorem  \ref{theo:NoSupset}. 

Under \eqref{GaussianSeq}, for any given data $\bm{y}$, the difference $\bm{y}^* \sim N(\bm{d}_0, I_p)$. 
Let $f_{p, \bm{d}}$ be the density of $N(\bm{d}, I_p)$.
For a Borel measurable subset $\mathcal{B}$ of the parameter space, the posterior probability of $\mathcal{B}$ is written as
   \begin{align}\label{posterior}
    \Pi_{n, L_n}{(\mathcal{B}|\bm{y}^*)} = \frac{\int_{\mathcal{B}}\frac{f_{p, \bm{d}} (\bm{y}^*)}{f_{p, \bm{d}_0} (\bm{y}^*)} d\pi(\bm{d}) }{\int \frac{f_{p, \bm{d}} (\bm{y}^*)}{f_{p, \bm{d}_0} (\bm{y}^*)} d\pi(\bm{d}) } = \frac{N_n(\mathcal{B})}{R_n}, 
   \end{align}
 where $\pi(\bm{d})$ is the prior distribution of $\bm{d}$ given by \eqref{prior:d}. 

We have the following lemma about the lower bound of the denominator $R_n$. 

\begin{lemma}[Lemma 2 in \cite{castillo2015bayesian}]
    \label{Lemma:castilo2015}
    For sufficiently large $p$ and any $\bm{d}_0 \in \mathbb{R}^p$, with support $S_0$, $K_n = |S_0|$, and $g_S$ being the product of Laplace density with scale parameter $\lambda$, we have, almost surely, 
    $$
    R_n \ge \frac{\pi_{L_n}(K_n)}{L_n^{2K_n}} \exp(-\lambda ||\bm{d}_0||_1-1). 
    $$
\end{lemma}
Lemma \ref{Lemma:castilo2015} is similar to Lemma 2 in \cite{castillo2015bayesian} by transferring $p$ to $L_n$.
The proof is analogous to theirs. 

We also introduce the following lemma to learn about the tail probability of the dimension prior $\pi_{L_n}(s)$. 

\begin{lemma}[Lemma 2.1 in \cite{ohn2022posterior}]
\label{lemma:ohn2022}
    For any fixed $\alpha$, for $Z_\ell$ following the prior distribution in \eqref{prior:h}, we have for any $s\ge 0$
    $$
    Pr\{|\bm{Z}|> k|\alpha\} \le \frac{14\alpha^{k+1}}{3 (\alpha+1)^k}. 
    $$
\end{lemma}
Lemma \ref{lemma:ohn2022} is a special case with $\kappa=0$ and $p=1$ of the two parameter construction of IBP weights in \cite{ohn2022posterior}. 
Based on Lemma \ref{lemma:ohn2022}, we immediately have the following corollary.  

\begin{corollary}[Tail probability of $\pi_{L_n}(s)$]
\label{coro:tailprob}
    Let $a= c_1 L_n^{-c_3}, b = c_2 L_n^{c_4}$ with $c_1, c_2 >0$, $c_3>c_4+2\ge3$ in the Gamma hyperprior in \eqref{prior:h}. 
    For any $k \ge 0$, $S \sim \pi_{L_n}$, as $L_n \to \infty$, we have 
    $$
    Pr\{S > k\} = o(L_n^{-2(k+1)}). 
    $$
\end{corollary}
\begin{proof}
  \begin{align*}
      Pr\{S > k\} &= \int  Pr\{|\bm{Z}| > k|\alpha\} \text{Gamma}(\alpha; a, b) d\alpha \\
      &\le \frac{14}{3} E\left(\frac{\alpha^{k+1}}{(\alpha+1)^k}\right). 
  \end{align*}
For any $k \ge 1$, $x^{k+1}/x^k$ is concave and thus, by Jensen's inequality we have 
$$
E\left(\frac{\alpha^{k+1}}{(\alpha+1)^k}\right) \le \frac{[E(\alpha)]^{k+1}}{[E(\alpha+1)]^k} = o(L^{-2(k+1)}). 
$$
\end{proof}

The following lemma provides the property of the adaptive precision parameter $\lambda_n(\delta)$. 
\begin{lemma}[Adaptive $\lambda_n(\delta)$]
\label{lemma:lambda}
   Given $\delta>0$, for $\lambda_n(\delta)$ in \eqref{lambdan}, as $K_n/p \to 0$, $n, p, L_n \to \infty$, we have 
   $$
  \sup_{\bm{d}_0 \in \Tilde{l}_0[K_n]}P_{\bm{d}_0}\{\lambda_n(\delta) ||\bm{d}_0||_1 \ge \delta\} < \frac{1}{p}. 
   $$
\end{lemma}
\begin{proof}
As $y_i^* \sim N(d_{0i}, 1)$, $|y_i^*|$  follows a folded normal distribution so that 
\begin{align*}
  E(|y_i^*|) &= \sqrt{\frac{2}{\pi}}\exp(-d_{0i}^2) + d_{0i}(1- 2\Phi(-d_{0i})), \\
  \text{Var}(|y_i^*|) &= d_{0i}^2 + 1 - E^2(|y_i^*|). 
\end{align*}
For $d_{0i} = 0$, $E(|y_i^*|) = \sqrt{2/\pi} \equiv \mu_0, \text{Var}(|y_i^*|) = 1- \mu_0^2$.; 
for $d_{0i} \not = 0$, as $L_n \to \infty$, 
$E(|y_i^*|) \to d_{0i}, \text{Var}(|y_i^*|) \to 1$. 
Therefore, for sufficiently large $p$, we have
\begin{align*}
E(\Bar{|\bm{y}|}) \to \mu_0 + \frac{1}{p}||\bm{d}_0||_1, \text{Var}(\Bar{|\bm{y}|}) \to \frac{1}{p}. 
\end{align*}

Then, by Chebyshev's inequality, we have 
\begin{align*}
     P_{\bm{d}_0}&\{\lambda_n(\delta) ||\bm{d}_0||_1 \ge \delta\} \\
     &= P_{\bm{d}_0} \{\Bar{|\bm{y}|} \ge \frac{1}{p} ||\bm{d}_0||_1\}\\
     & = P_{\bm{d}_0}\{|\Bar{|\bm{y}|} -  E(\Bar{|\bm{y}|})| \ge \mu_0 \}\\
     & \le \frac{1}{p\mu_0^2} < \frac{1}{p}. 
\end{align*}
\end{proof}

Now we start the proof of Theorem \ref{theo:NoSupset}. 
\begin{proof}
Let $\sigma(\bm{y}*)$ be the sigma field generated by the data $\bm{y}*$. 
Lemma \ref{lemma:lambda} indicates that there exists a Borel set $\mathbb{B}_n \in \sigma(\bm{y}*)$ so that $P_{\bm{d}_0}(\mathbb{B}_n^c) < 1/p$ and $\lambda_n(\delta)||\bm{d}_0||_1 < \delta$ holds on $\mathbb{B}_n$. 

Note that 
\begin{align*}
   E_{\bm{d}_0} \Pi_{n, L_n}{(\mathcal{B}|\bm{y}^*)} &= \int \frac{N_n(\mathcal{B})}{R_n} f_{p, \bm{d}_0}(\bm{y}^*) d\bm{y}^* \\
   &= Rn^{-1} \int\int_{\mathcal{B}} f_{p, \bm{d}}(\bm{y}^*)d\pi(\bm{d})d\bm{y}^* \\
   & = Rn^{-1} \int_{\mathcal{B}} \int f_{p, \bm{d}}(\bm{y}^*)d\bm{y}^*d\pi(\bm{d}) \\
   & = Rn^{-1} \pi(\mathcal{B}).  
\end{align*}

Hence, by Lemma \ref{Lemma:castilo2015} and Corollary \ref{coro:tailprob}, we have 
\begin{align*}
      E_{\bm{d}_0} \Pi_{n, L_n} &\{\bm{d}: |\bm{d}| > K_n|\bm{y}^*\}  \\
      & \le P_{\bm{d}_0}(\mathbb{B}_n^c) + E_{\bm{d}_0} [\pi(|d| > K_n ) \bm{1} \mathbb{B}_n] \\
      & <\frac{1}{p} + R_n^{-1} \pi(|d| > K_n ) \\
      & \le \frac{1}{p} + Q_1{K_n\log(L_n)}L_n^{-2}\exp(\lambda ||\bm{d}_0||_1), \\
      & <  \frac{1}{p} + Q_1{K_n\log(L_n)}L_n^{-2}\exp(\delta), 
\end{align*}
where $Q_1 = (1+ \log L_n)(e Q_0 \log L_n)^{-1}$ with $Q_0$ given by \eqref{pi0s}. 
Obviously, the RHS of the last inequality on the above tends to zero as $n, L_n \to \infty$. 
\end{proof}

\subsection{Proof of Corollary \ref{Falsenegative}}

\begin{proof}
    Corollary \ref{Coro:MAPasym} implies that $d_i^{\text{MAP}}$ is a consistent estimator of ${d}_{0i}$. 
    Therefore, with the cut-off of $\Tilde{l}_0[K_n]$, it suffices to showing that, for $M$ in Theorem \ref{theo:RecDim}, 
   $$
   \inf_{\bm{d}_0 \in \Tilde{l}_0[K_n]}E_{\bm{d}_0} \Pi_{n, L_n}\left\{\psi < \frac{M}{3} \sqrt{K_n \log(L_n/K_n)} | \bm{y}^*\right\} \to 1, 
   $$   
   for as $n, L_n \to \infty$. 
    Since 
    $$
    \psi_0 = p^{-1/2}||\bm{d}_0 - \Bar{d}_0 \bm{1}_p||_2 \le p^{-1/2} ||\bm{d}_0||_2, 
    $$
    therefore, $3\psi_0 < M\sqrt{K_n \log(L_n/K_n)}$ by Assumption (A2). 

    Corollary \ref{Coro:MAPasym} indicates that $\Bar{d} \to \Bar{d}_0$. 

    Then by triangle inequality, 
    we have 
    \begin{align*}
  &E_{\bm{d}_0} \Pi_{n, L_n}\left\{\psi < \frac{M}{3} \sqrt{K_n \log(L_n/K_n)}|\bm{y}^*\right\} \ge \\
  &E_{\bm{d}_0} \Pi_{n, L_n}\left\{\psi_0 + p^{-1/2}||\bm{d} - \bm{d}_0||_2 < \frac{M}{3} \sqrt{K_n \log(L_n/K_n)}|\bm{y}^*\right\}. 
    \end{align*}
Theorem \ref{minimax} indicates that the RHS of the above inequality tends to $1$. 
\end{proof}

\subsection{Proof of Theorem \ref{theorem:converge}}
\label{subsec:proofTheo4}
\begin{proof}

It is trivial that
$$|\sum_{\ell=1}^\infty h_\ell I(\xi_\ell\le t)| \leq \sum_{\ell=1}^\infty|h_\ell|.$$
Then, for any integers $m_1 < m_2$, 
we have 
\begin{align*}
P(\sum_{\ell=m_1+1}^{m_2}|h_\ell| > \epsilon ) & \leq P\left(\bigcup_{\ell=m_1+1}^{m_2} |h_\ell| > \frac{\epsilon}{m_2-m_1}\right) \\
    & \leq \sum_{\ell=m_1+1}^{m_2} P\bigg(|h_\ell| > \frac{\epsilon}{m_2-m_1} \bigg) \\
    & \leq \sum_{\ell=m_1+1}^{m_2} \bigg[ 1 - F_{\text{0}}\bigg( \frac{\epsilon}{m_2-m_1}  \bigg)  \bigg] \eta_\ell \\
    &\quad + F_{\text{0}}\bigg( \frac{-\epsilon}{m_2-m_1}  \bigg) \eta_\ell \\
    & \leq 2 \sum_{\ell=m_1+1}^{m_2} \eta_\ell. 
\end{align*}
This inequality indicates that if $\sum_{\ell=1}^\infty \eta_\ell$ is converged, then we have  $\sum_{\ell=1}^\infty |h_\ell|$ converged according to probability. 
To prove the convergence of $\sum_{\ell=1}^\infty \eta_\ell$, it is equivalent to prove $\sum_{\ell=1}^\infty E(\eta_\ell) < \infty$. Firstly, we have 
\begin{align*}
    E(\eta_\ell) =\prod_{j=1}^\ell E\{E(p_j|\alpha)\} = \left\{E\left(\frac{\alpha}{1+\alpha} \right)\right\}^\ell. \\
\end{align*}
Then by Jensen's inequality,  for any fixed $a, b$ in the Gamma prior, 
\begin{align*}
    \sum_{\ell=1}^\infty E(\eta_\ell) \le    \sum_{\ell=1}^\infty \left\{\frac{ab}{1+ab}\right\}^{\ell} =ab <\infty.
\end{align*}
\end{proof}

\section{Additional simulations}
\label{sec:addSim}
\subsection{Model misspecification}
\label{subsec:modelmis}
We conduct additional simulations under the case where our method meets with model misspecification, including heavy-tailed noises in mean-shifted models, auto-correlated noises in mean-shifted models, and an AR(2) model with structural changes. 
We generate simulated data under the following settings and conduct 300 Monte Carlo replicates under each setting. 

\begin{enumerate}
    \item[] (\textbf{MS.1}) Changes of means with heavy tailed noises. 
    We generate $n=400$ $y_i = \mu_i + \epsilon_i$, where $\epsilon_i \sim \sqrt{2}^{-1}t(4)$ are i.i.d. heavy-tailed noises. 
    We set $K=7$ change-points at $(50, 100, 150, 200, 250, 300, 350)$, leading to $8$ segments with segment mean $\mu = (0, 1.5, 3, 1.5, 3, 0.5, 2, 0)$. 
    This setting is similar to setting \textbf{S.1} except for the heavy-tailed noise type. 
    \item[] (\textbf{MS.2}) Changes of means with auto-correlated noises. 
    We generate $n=400$ $y_t = \mu_t + \epsilon_t$, where $\epsilon_1 \sim N(0, 1), \epsilon_{t} = 0.5 \epsilon_{t-1} + \alpha_t$, and $\alpha_t \sim N(0, 1)$ are i.i.d. Gaussian noises. 
    We take the same setting on the means $\mu$ as in setting \textbf{S.1}. 
\item[] (\textbf{MS.3}) Changes of auto-correlation coefficients in mixture of AR(1) and AR(2) model. 
We generate $n=450$ observations and $y_1 \sim N(0, 1)$. 
For $t\ge 2$, 
\begin{equation*}
y_t=\left\{
\begin{aligned}
&0.5y_{t-1} + \epsilon_t,  t \le 50;\\
&-0.5y_{t-1} + \epsilon_t, 50<t\le 100; \\
&0.65y_{t-1} + 0.35y_{t-1} + \epsilon_t, 100<t\le 200; \\
&-0.25y_{t-1} + \epsilon_t, 300<t \le 300; \\
&-0.85y_{t-1} -0.35 y_{t-2} + \epsilon_t; 300<t \le 400; \\
&0.45y_{t-1} + \epsilon_t, 400<t \le 450. 
\end{aligned}
\right.
\end{equation*}
Here $\epsilon_t \sim N(0,1)$ are i.i.d. Gaussian noises. 
Under this setting, $K=5$ change-points are located at $(50, 100, 200, 300, 400)$. 
\end{enumerate}

Examples of the simulated data under cases \textbf{MS.1} to \textbf{MS.2} are presented in Figures \ref{HTMean7CP} to \ref{AR2mis5CP}. 
In Figure \ref{AR2mis5CP}, the red line denotes the first order auto-correlation coefficient. 
Note that on the interval $(100, 200)$, both the first and the second order auto-correlation coefficients are positive and hence the signs of the data on the interval are grouped together. 

\begin{figure*}[!htb]
    \centering
    \subfigure[]{
    \begin{minipage}[t]{0.3\linewidth}
      \centering
\includegraphics[height = .7\textwidth]{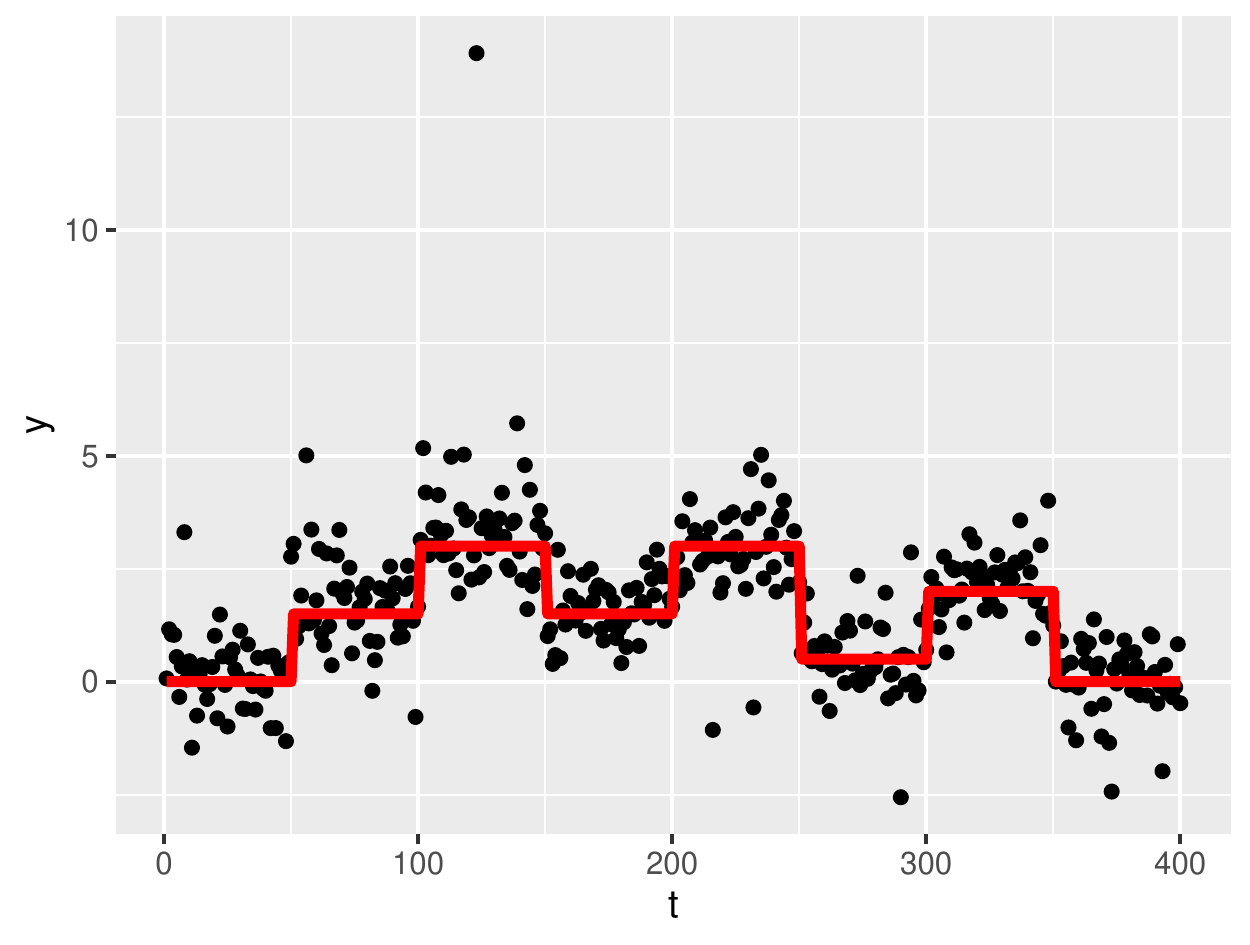}
\label{HTMean7CP}
    \end{minipage}
    }
\subfigure[]{
    \begin{minipage}[t]{0.3\linewidth}
      \centering
\includegraphics[height = .7\textwidth]{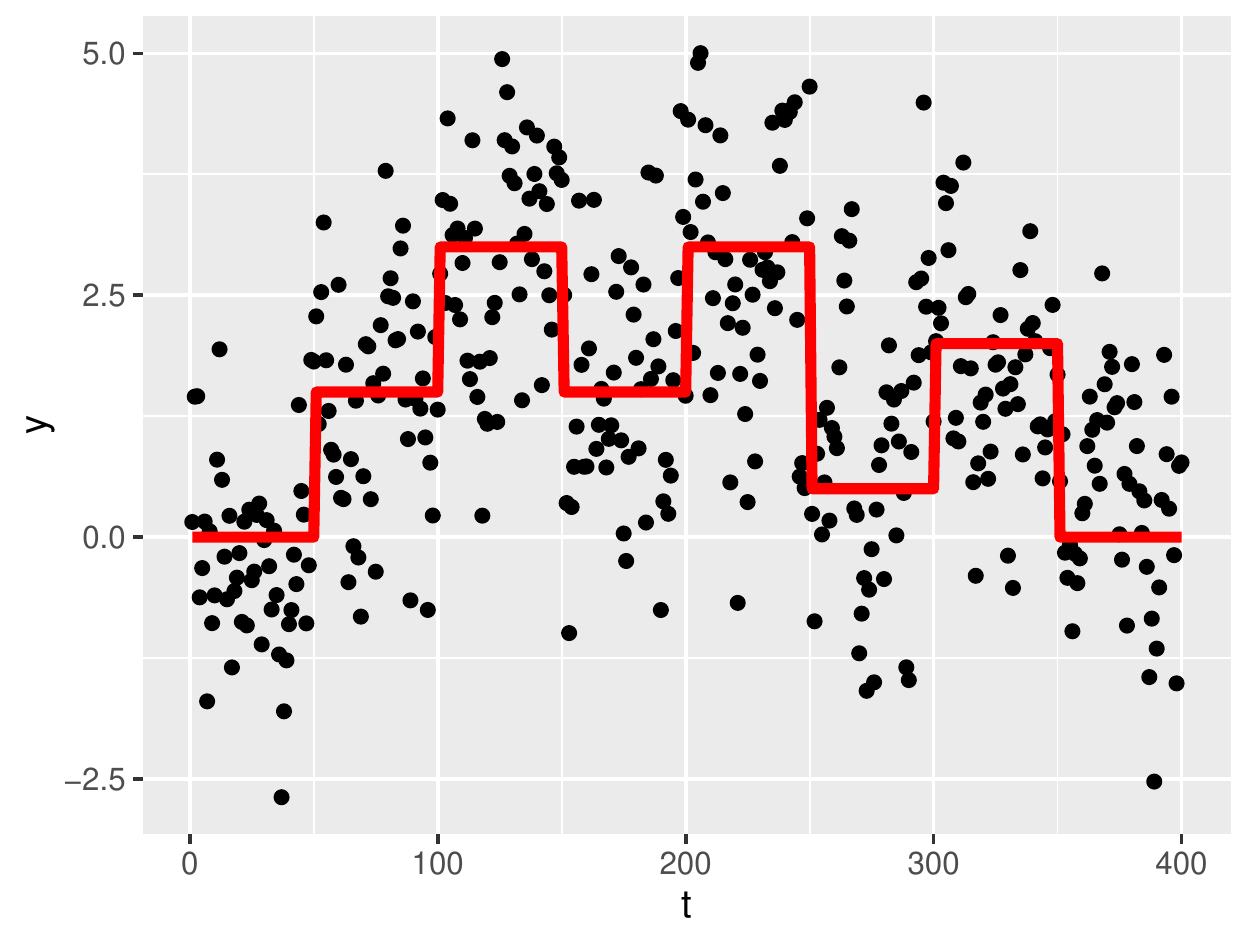}
\label{Dependent7CP}
    \end{minipage}
    }
    \subfigure[]{
    \begin{minipage}[t]{0.3\linewidth}
      \centering
\includegraphics[height = .7\textwidth]{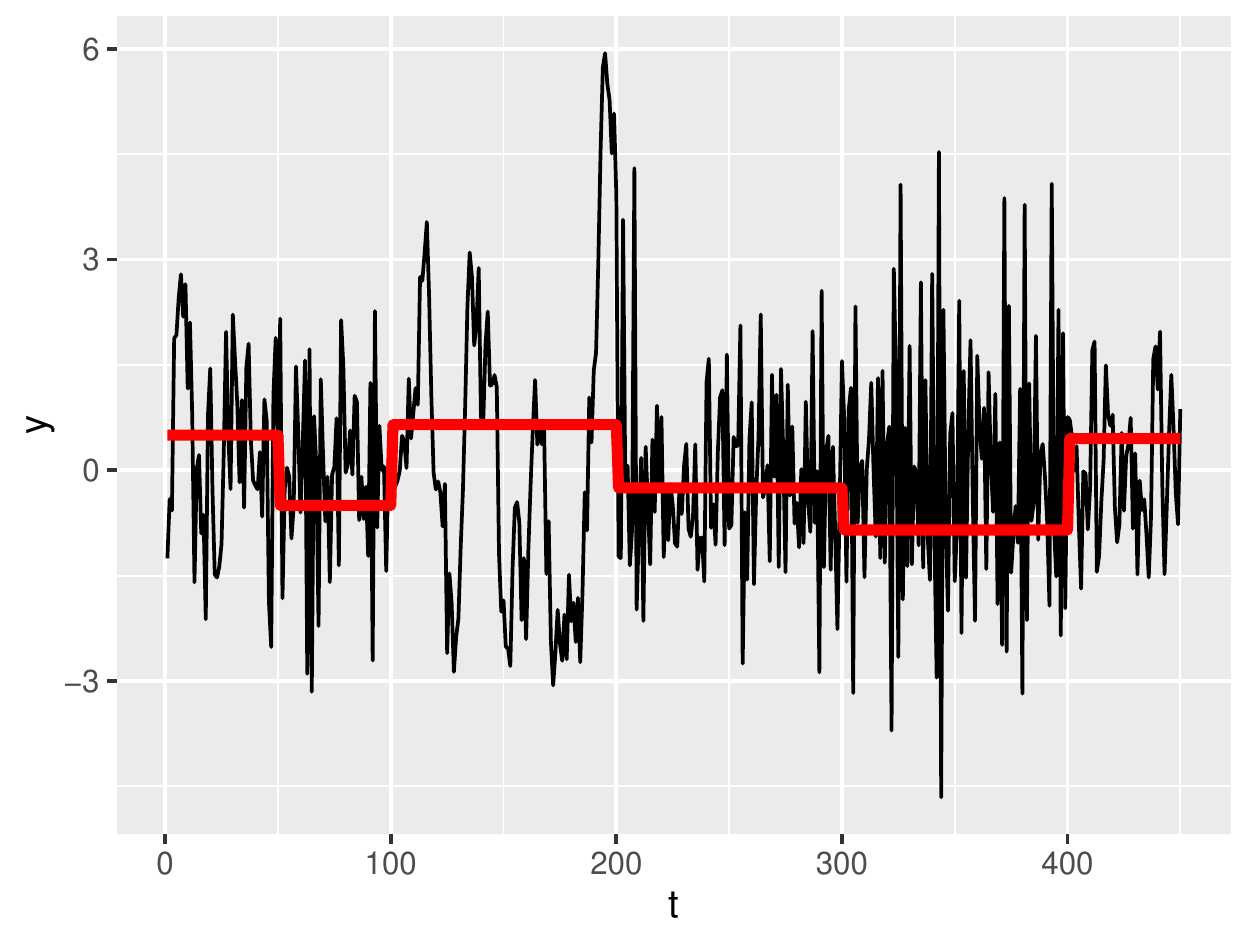}
\label{AR2mis5CP}
    \end{minipage}
    }
    \caption{\footnotesize{Examples of generated data in simulations. (a) to (c), settings \textbf{MS.1} to \textbf{MS.3}. }}
    \label{fig:missimexample}
\end{figure*}

Besides competitors under simulation settings \textbf{S.1} to \textbf{S.5}, we add the heavy-tailed version of package $\texttt{not}$ \cite{baranowski2019narrowest} under setting \textbf{MS.1}, named NOT-HT; 
we also include a nonparametric estimator of change-point \texttt{changepoint.np} by \cite{haynes2017computationally} in settings \textbf{MS.1} and \textbf{MS.2}. 

Results are given by Table \ref{tab:addsimres}. 
We find that under setting \textbf{MS.1}, NOSE is comparable with the best approach even though under model misspecifications. 
Under setting \textbf{MS.2}, MOSUM outperforms since it does not require independent assumptions on the data stream with shifts in the mean. 
Under setting \textbf{MS.3}, although \texttt{wbsts} has a higher frequency of correct detection of the number of change-points, their estimation of the locations is poor, leading to much lower precision and recall, and higher Hausdorff distance. 
\begin{table*}[!htb]
	\centering
 \footnotesize
	\caption{\footnotesize{Results of change-points detection under model mispecification settings \textbf{MS.1} to \textbf{MS.3} among 300 Monte Carlo replicates. The best results are bold. }}
 \label{tab:addsimres}
	\begin{tabular*}{\textwidth}{@{\extracolsep{\fill}}cccccccccccc @{\extracolsep{\fill}}}
		\hline
	Setting &	Method & \multicolumn{7}{c}{Frequency of $\hat{K}-K$} & Precision & Recall & $d_H \times 10^2$\\
		\hline
  & & $\le-3$ &-2 & -1 & 0 & +1 & +2 & $\ge +3$  &  &  & \\
  \hline
	\textbf{MS.1}&	NOSE & 1 & 3 & 4 & $260$ & 31 & 1 & 0 & $0.97$ & $0.98$ & $1.6$\\
        &    NOT-HT & 0 & 0 & 0 & $\bm{295}$ & 4 & 1 & 0 & $\bm{0.99}$ & $0.98$ & $\bm{0.9}$ \\
         &   SMUCE & 0 & 0 & 1 & 107 & 63 & 59 & 70 & 0.84 & 0.99 & 3.8\\
          &  WBS & 0 & 0 & 0 & 34 & 18 & 59 & 189 & 0.67 & 0.99 & 5.6\\
          &  FDRSeg & 0 & 0 & 0 & 15 & 8 & 22 & 255 & 0.55 & 0.99 & 6.7\\
        &    PELT & 0 &0 & 0 & 73 & 45 & 87 & 95 & 0.80 & 0.99 & 3.8\\
        & PELT-np & 0 & 0 & 0 & 227 & 43 & 26 & 4 & 0.95 & 0.99 & 1.8\\
        & TUGH & 0 & 0 & 0 & 242 & 48 & 9 & 1 & 0.97 & 0.99 & 1.8 \\
        & MOSUM & 0 & 0 & 3 & 255 & 41 & 1 & 0 & 0.98 & 0.99 & 1.9\\
        \hline
\textbf{MS.2} & NOSE & 0 & 2 & 19 & $87$ & 89 & 65 & 38 & $0.70$ & 0.80 & 5.2\\
& NOT & 1 & 0 & 9 & 57 & 32 & 49 & 153 & 0.64 & 0.87 & 6.1\\
 & SMUCE & 0& 0 & 1 & 2 & 7 & 27 & 264 & 0.55 & 0.91 & 7.5 \\
 & WBS & 0 & 0 & 0 & 0 & 4 & 1 & 295 & $0.43$ & 0.94 & 8.4\\
 & FDRSeg & 0 & 0 & 0 & 0 & 0 & 1 & 299 & 0.28 & 0.95 & 9.9 \\
 & PELT & 4 & 11 & 28 & $126$ & 83 & 30 & 18 & $0.79$ & 0.83 & $4.6$ \\
 &PELT-NP & 0 & 1 & 2& 46 & 76 & 68 & 107 & 0.66 & 0.84 & 5.8 \\
 & TUGH & 0 & 0 & 0 & 1 & 13 & 14 & 272 & 0.53 & 0.91 & 7.1 \\
& MOSUM & 0 & 3 & 39 & \textbf{176} & 70 & 12 & 0 & \textbf{0.96} & \textbf{0.91} & \textbf{4.3}\\
\hline
\textbf{MS.3} & NOSE & 0 & 55 & 144 & 78 & 22 & 7 & 0 & $\bm{0.83}$ & $\bm{0.69}$ & $3.8$\\
& WBSTS & 14 & 57 & 84 & \textbf{90} & 40 & 15 & 0 & 0.54 & 0.46 & 7.0\\ 
& B-P & 191 & 74 & 35 & 0 & 0 & 0 & 0 & $0.79$ & 0.38 & $\bm{2.6}$\\
		\hline
	\end{tabular*}
\end{table*}

\subsection{Simulations for DRAIP data}
\label{subsec:SimDRAIP}
We generate a series of independent Gaussian data to simulate the DRAIP data. 
We generate synthetic data based on the detection result given by NOSE in the real DRAIP data. 
That is, 7 change-points are set at $(37, 137, 206, 336, 426, 510, 630)$. 
On each segment divided by these change-points, data are i.i.d. Gaussian variables with means $\mu = (0.141, 0.124, 0.399, 0.214, -0.112, -0.093, -0.053, \\0.116)$ (the sample mean of the DRAIP data on each segment) and $\sigma$ being the sample SDs on those segments divided by NOSE. 
We conduct 300 Monte Carlo replicates for the simulation. 
An example is presented in Figure \ref{fig:simDRAIP}. 

\begin{figure} [H]
	\centering
	\includegraphics[width = .5 \textwidth]{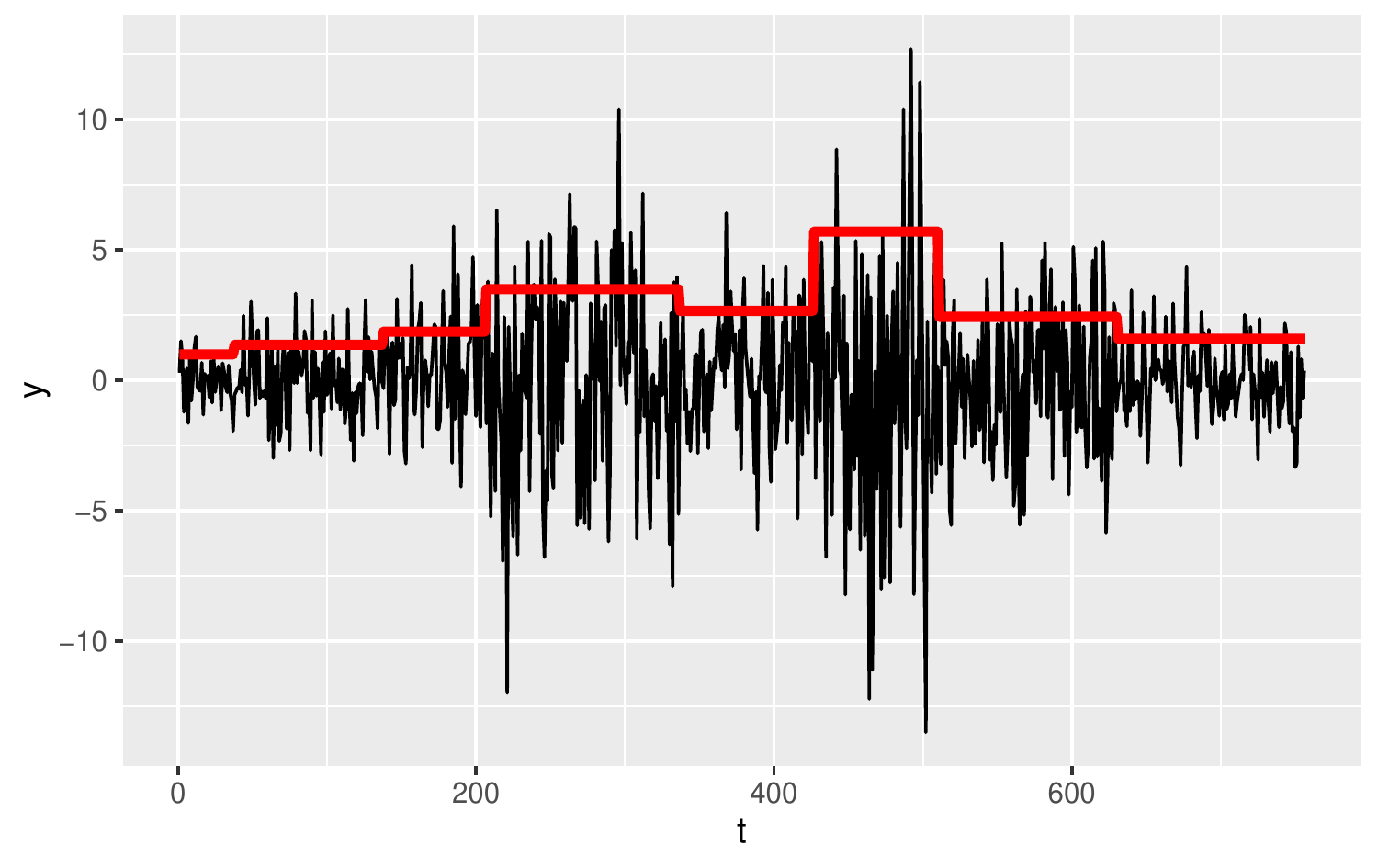}
	\caption{\footnotesize{Simulated example for the DRAIP data and the true values of scale parameters (in red polyline).}} \label{fig:simDRAIP}
\end{figure}
We present the detection results in Table \ref{tab:draipsim}. 
As expected, the small jump sizes and varying means lead to serious under-detection of change-points for all approaches. 
Even so, NOSE performs much better in correctly detecting change-points compared with other approaches. 
This simulation demonstrates the reliability of detection results given by NOSE on the DRAIP data.

\begin{table*}[!htb]
	\centering
 \footnotesize
	\caption{\footnotesize{Results of change-points detection under simulations for the DRAIP data and the ACGH data. }}\label{tab:draipsim}
\begin{tabular*}{\textwidth}{@{\extracolsep{\fill}}cccccccccccc @{\extracolsep{\fill}}}		
\hline
	Setting &	Method & \multicolumn{7}{c}{Frequency of $\hat{K}-K$} & Precision & Recall & $d_H \times 10^2$\\
		\hline
  & & $\le-3$ &-2 & -1 & 0 & +1 & +2 & $\ge +3$  &  &  & \\
  \hline
  DRAIP & NOSE & 8 & 147 & 110 & \textbf{33} & 2  & 0 & 0 & 0.90 & \textbf{0.71} & $3.4$\\
        &  NOT & 224 & 60 & 11 & 5 & 0 & 0 & 0 & 0.94 & 0.54 & \textbf{2.0} \\
        &   SMUCE & 282 & 17 & 1 & 0 & 0 & 0 & 0 & \textbf{1} & 0.48 & 19.5\\
        &  PELT & 95 & 119 & 78 & 8 & 0 & 0 & 0 & 0.92 & 0.64 & 2.6\\
		\hline
  ACGH&	NOSE & 0 & 0 & 1 & 108 & 140  & 44 & 7 & 0.93 & 0.99 & $2.5$\\
  &HSMUCE & 0 & 0 & 1 & 35 & 131 & 102 & 31 & 0.90 & 0.93 & 15.5\\
        &    NOT & 0 & 0 & 0 & 28 & 12 & 107 & 153 & 0.81 & $0.98$ & 18.2 \\
         & R-FPOP & 0 &53 & 166 & 21 & 60 & 0 & 0 & 0.99 & 0.84 & 3.25\\
         &   SMUCE & 0 & 0 & 0 & 0 & 0 & 0 & 300 & 0.51 & 0.98 & 20.9\\
          &  WBS & 0 & 0 & 0 & 0 & 0 & 0 & 300 & 0.52 & 0.98 & 20.9\\
          &  FDRSeg & 0 & 0 & 0 & 0 & 0 & 0 & 300 & 0.30 & 0.97 & 21.3\\
          & TUGH & 1 & 0 & 0 & 1 & 0 & 0 & 298 & 0.48 & 0.96 & 20.2 \\
          & MOSUM & 0 & 0 & 0 & 3 & 5 & 34 & 258 & 0.74 & 0.94 & 13.1\\
          \hline
	\end{tabular*}
\end{table*} 

\subsection{Simulations for ACGH data}
\label{subsec:SimACGH}
We generate a series of independent Gaussian data to simulate the ACGH data. 
We use the smooth signal estimated by \texttt{DeCAFS} \citep{romano2022detecting} as the means of Gaussian variables. 
The scale parameter is set as the sum of the estimated standard deviations of the drift and the AR(1) noise process. 
An example is presented in Figure \ref{fig:simACGH}. 
As can be found in the figure, such a data-generating process simulates the true data quite well with an average mean square error of 0.0265 (0.001) among the simulated datasets (standard deviation in bracket). 
The Gaussian scheme naturally avoids most possible outliers. 
For comparison, we use the detection result on the real ACGH dataset given by NOSE as the golden standard. 
That is, 13 change-points are set at $(73, 123, 263, 342, 524, 583, 657, 745, 1724, 1906, 1965, 2041, 2143)$. 
Since the data stream is long, we set the window size for true positive detection as $25$ in the simulation. 
We conduct 300 Monte Carlo replicates for the simulation. 
The simulation results combined in Table \ref{tab:draipsim} shows that both NOSE and R-FPOP provide consistent estimation results with that of the real-data experiment in the simulations. 
By removing most outliers, the results of HSMUCE tend to more similar to that of NOSE. 
Compared with the real-data experiment, \texttt{NOT} seems to be slightly over-detect change-points in simulations. 
In terms of the remaining methods, they significantly over-detect change-points in both real-data experiments and simulations. 
We do not incorporate the PELT method here since it fails to detect any change-points in most cases. 

\begin{figure} [H]
	\centering
	\includegraphics[width = .5\textwidth]{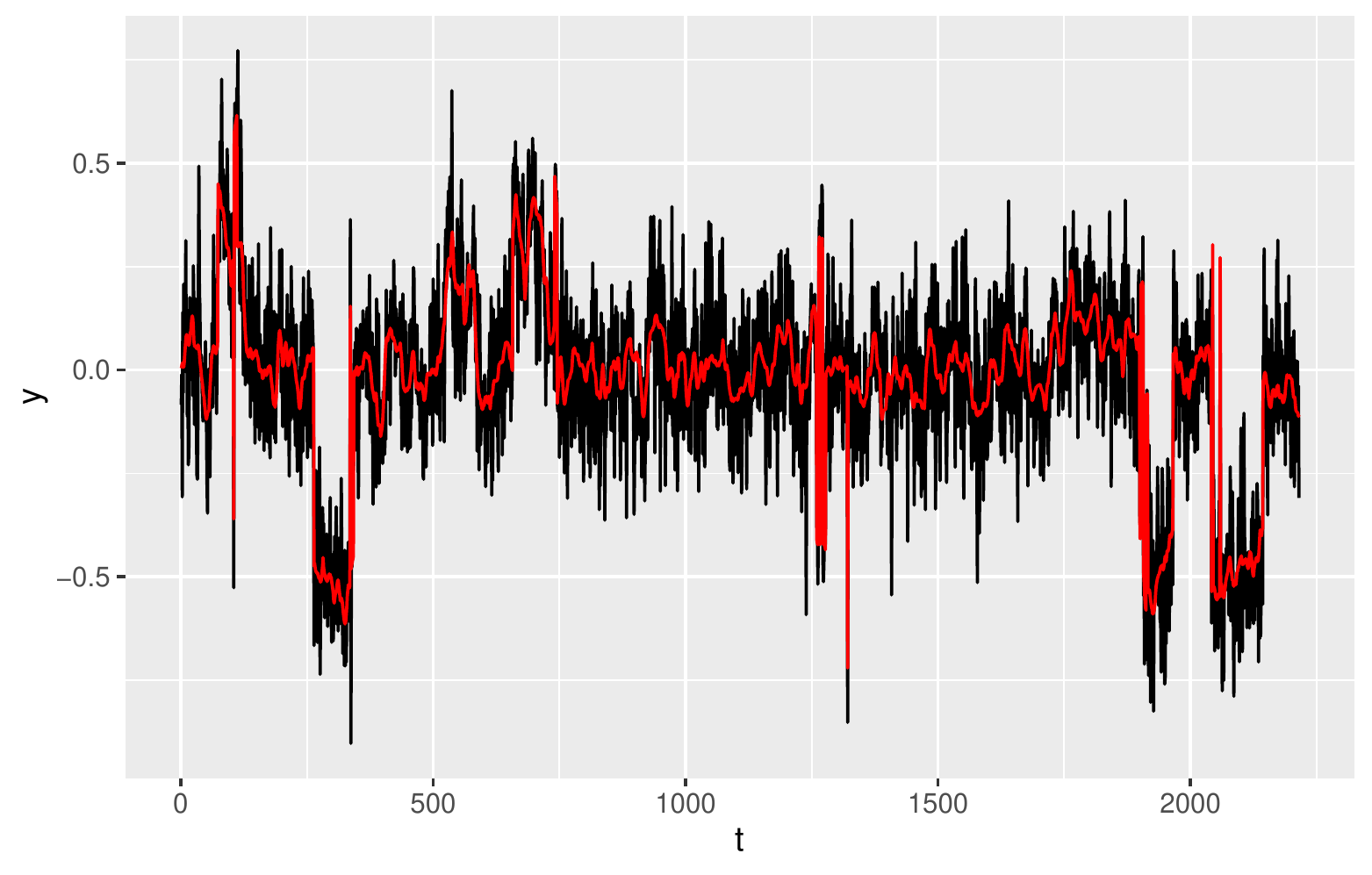}
	\caption{\footnotesize{Simulated example for the ACGH data and the smooth signal estimated by DeCAFS (in red curves).} }
 \label{fig:simACGH}
\end{figure}

\newpage
\bibliographystyle{apalike}
\linespread{0.5}
\selectfont
\bibliography{main.bbl}

\begin{thebibliography}{}

\bibitem[Anastasiou et~al., 2022]{breakfast2022}
Anastasiou, A., Chen, Y., Cho, H., and Fryzlewicz, P. (2022).
\newblock {\em breakfast: Methods for Fast Multiple Change-Point Detection and
  Estimation}.
\newblock R package version 2.3.

\bibitem[Bai and Perron, 2003]{bai2003computation}
Bai, J. and Perron, P. (2003).
\newblock Computation and analysis of multiple structural change models.
\newblock {\em Journal of Applied Econometrics}, 18(1):1--22.

\bibitem[Bai et~al., 2020]{bai2020spike}
Bai, R., Rockova, V., and George, E.~I. (2020).
\newblock {Spike-and-slab meets LASSO: A review of the spike-and-slab LASSO}.
\newblock {\em arXiv preprint arXiv:2010.06451}.

\bibitem[Baranowski et~al., 2019]{baranowski2019narrowest}
Baranowski, R., Chen, Y., and Fryzlewicz, P. (2019).
\newblock Narrowest-over-threshold detection of multiple change points and
  change-point-like features.
\newblock {\em Journal of the Royal Statistical Society: Series B (Statistical
  Methodology)}, 81(3):649--672.

\bibitem[Barbieri and Berger, 2004]{barbieri2004optimal}
Barbieri, M.~M. and Berger, J.~O. (2004).
\newblock Optimal predictive model selection.
\newblock {\em The Annals of Statistics}, pages 870--897.

\bibitem[Birte and Claudia, 2018]{eichinger2018mosum}
Birte, E. and Claudia, K. (2018).
\newblock A mosum procedure for the estimation of multiple random change
  points.
\newblock {\em Bernoulli}, 24(1):526--564.

\bibitem[Cappello et~al., 2023]{cappello2023bayesian}
Cappello, L., Madrid~Padilla, O.~H., and Palacios, J.~A. (2023).
\newblock Bayesian change point detection with spike and slab priors.
\newblock {\em Journal of Computational and Graphical Statistics},
  (just-accepted):1--24.

\bibitem[Carlstein et~al., 1988]{carlstein1988nonparametric}
Carlstein, E. et~al. (1988).
\newblock Nonparametric change-point estimation.
\newblock {\em The Annals of Statistics}, 16(1):188--197.

\bibitem[Castillo et~al., 2015]{castillo2015bayesian}
Castillo, I., Schmidt-Hieber, J., and van~der Vaart, A. (2015).
\newblock Bayesian linear regression with sparse priors.
\newblock {\em The Annals of Statistics}, 43(5):1986--2018.

\bibitem[Castillo and van~der Vaart, 2012]{castillo2012needles}
Castillo, I. and van~der Vaart, A. (2012).
\newblock Needles and straw in a haystack: Posterior concentration for possibly
  sparse sequences.
\newblock {\em The Annals of Statistics}, 40(4):2069--2101.

\bibitem[Chernozhukov et~al., 2017]{chernozhukov2017lava}
Chernozhukov, V., Hansen, C., and Liao, Y. (2017).
\newblock {A LAVA attack on the recovery of sums of dense and sparse signals}.
\newblock {\em The Annals of Statistics}, 45(1):39--76.

\bibitem[Cho and Fryzlewicz, 2015]{cho2015multiple}
Cho, H. and Fryzlewicz, P. (2015).
\newblock Multiple-change-point detection for high dimensional time series via
  sparsified binary segmentation.
\newblock {\em Journal of the Royal Statistical Society: Series B (Statistical
  Methodology)}, 77(2):475--507.

\bibitem[de~Valpine et~al., 2017]{de2017programming}
de~Valpine, P., Turek, D., Paciorek, C.~J., Anderson-Bergman, C., Lang, D.~T.,
  and Bodik, R. (2017).
\newblock {Programming with models: writing statistical algorithms for general
  model structures with NIMBLE}.
\newblock {\em Journal of Computational and Graphical Statistics},
  26(2):403--413.

\bibitem[Donoho et~al., 1992]{donoho1992maximum}
Donoho, D.~L., Johnstone, I.~M., Hoch, J.~C., and Stern, A.~S. (1992).
\newblock Maximum entropy and the nearly black object.
\newblock {\em Journal of the Royal Statistical Society: Series B
  (Methodological)}, 54(1):41--67.

\bibitem[Du et~al., 2016]{du2016stepwise}
Du, C., Kao, C.-L.~M., and Kou, S. (2016).
\newblock Stepwise signal extraction via marginal likelihood.
\newblock {\em Journal of the American Statistical Association},
  111(513):314--330.

\bibitem[Fang et~al., 2020]{fang2020segmentation}
Fang, X., Li, J., and Siegmund, D. (2020).
\newblock Segmentation and estimation of change-point models: false positive
  control and confidence regions.
\newblock {\em The Annals of Statistics}, 48(3):1615--1647.

\bibitem[Fearnhead, 2006]{fearnhead2006exact}
Fearnhead, P. (2006).
\newblock Exact and efficient bayesian inference for multiple changepoint
  problems.
\newblock {\em Statistics and Computing}, 16(2):203--213.

\bibitem[Fearnhead and Rigaill, 2019]{fearnhead2019changepoint}
Fearnhead, P. and Rigaill, G. (2019).
\newblock Changepoint detection in the presence of outliers.
\newblock {\em Journal of the American Statistical Association},
  114(525):169--183.

\bibitem[Frick et~al., 2014]{frick2014multiscale}
Frick, K., Munk, A., and Sieling, H. (2014).
\newblock Multiscale change point inference.
\newblock {\em Journal of the Royal Statistical Society: Series B (Statistical
  Methodology)}, 76(3):495--580.

\bibitem[Fryzlewicz, 2014]{fryzlewicz2014wild}
Fryzlewicz, P. (2014).
\newblock Wild binary segmentation for multiple change-point detection.
\newblock {\em The Annals of Statistics}, 42(6):2243--2281.

\bibitem[Fryzlewicz, 2018]{fryzlewicz2018tail}
Fryzlewicz, P. (2018).
\newblock Tail-greedy bottom-up data decompositions and fast multiple
  change-point detection.
\newblock {\em The Annals of Statistics}, 46(6B):3390--3421.

\bibitem[Fryzlewicz, 2021]{fryzlewicz2021narrowest}
Fryzlewicz, P. (2021).
\newblock Narrowest significance pursuit: inference for multiple change-points
  in linear models.
\newblock {\em arXiv preprint arXiv:2009.05431}.

\bibitem[Hahn and Carvalho, 2015]{hahn2015decoupling}
Hahn, P.~R. and Carvalho, C.~M. (2015).
\newblock {Decoupling shrinkage and selection in Bayesian linear models: a
  posterior summary perspective}.
\newblock {\em Journal of the American Statistical Association},
  110(509):435--448.

\bibitem[Haynes et~al., 2017]{haynes2017computationally}
Haynes, K., Fearnhead, P., and Eckley, I.~A. (2017).
\newblock A computationally efficient nonparametric approach for changepoint
  detection.
\newblock {\em Statistics and Computing}, 27(5):1293--1305.

\bibitem[James, 2017]{james2017bayesian}
James, L.~F. (2017).
\newblock {Bayesian Poisson calculus for latent feature modeling via
  generalized Indian buffet process priors}.
\newblock {\em The Annals of Statistics}, 45(5):2016--2045.

\bibitem[James et~al., 2015]{james2015ecp}
James, N.~A., Matteson, D.~S., et~al. (2015).
\newblock {ecp: An R Package for nonparametric multiple change point analysis
  of multivariate data}.
\newblock {\em Journal of Statistical Software}, 62(i07).

\bibitem[Jeong and Ghosal, 2021]{jeong2021unified}
Jeong, S. and Ghosal, S. (2021).
\newblock Unified bayesian theory of sparse linear regression with nuisance
  parameters.
\newblock {\em Electronic Journal of Statistics}, 15(1):3040--3111.

\bibitem[Jula~Vanegas et~al., 2021]{jula2021multiscale}
Jula~Vanegas, L., Behr, M., and Munk, A. (2021).
\newblock Multiscale quantile segmentation.
\newblock {\em Journal of the American Statistical Association}, pages 1--14.

\bibitem[Killick and Eckley, 2014]{killick2014changepoint}
Killick, R. and Eckley, I. (2014).
\newblock {changepoint: An R package for changepoint analysis}.
\newblock {\em Journal of statistical software}, 58(3):1--19.

\bibitem[Killick et~al., 2012]{killick2012optimal}
Killick, R., Fearnhead, P., and Eckley, I.~A. (2012).
\newblock Optimal detection of changepoints with a linear computational cost.
\newblock {\em Journal of the American Statistical Association},
  107(500):1590--1598.

\bibitem[Kingman, 1967]{kingman1967completely}
Kingman, J. (1967).
\newblock Completely random measures.
\newblock {\em Pacific Journal of Mathematics}, 21(1):59--78.

\bibitem[Knowles and Ghahramani, 2011]{knowles2011nonparametric}
Knowles, D. and Ghahramani, Z. (2011).
\newblock {Nonparametric Bayesian sparse factor models with application to gene
  expression modeling}.
\newblock {\em The Annals of Applied Statistics}, 5(2B):1534--1552.

\bibitem[Ko et~al., 2015]{ko2015dirichlet}
Ko, S.~I., Chong, T.~T., Ghosh, P., et~al. (2015).
\newblock Dirichlet process hidden markov multiple change-point model.
\newblock {\em Bayesian Analysis}, 10(2):275--296.

\bibitem[Korkas and Pryzlewiczv, 2017]{korkas2017multiple}
Korkas, K.~K. and Pryzlewiczv, P. (2017).
\newblock Multiple change-point detection for non-stationary time series using
  wild binary segmentation.
\newblock {\em Statistica Sinica}, pages 287--311.

\bibitem[Li et~al., 2016]{li2016fdr}
Li, H., Munk, A., and Sieling, H. (2016).
\newblock {FDR-control in multiscale change-point segmentation}.
\newblock {\em Electronic Journal of Statistics}, 10(1):918--959.

\bibitem[Lunn et~al., 2000]{lunn2000winbugs}
Lunn, D.~J., Thomas, A., Best, N., and Spiegelhalter, D. (2000).
\newblock {WinBUGS-a Bayesian modelling framework: concepts, structure, and
  extensibility}.
\newblock {\em Statistics and computing}, 10(4):325--337.

\bibitem[Ma and Liu, 2022]{ma2022posterior}
Ma, Y. and Liu, J.~S. (2022).
\newblock {On posterior consistency of Bayesian factor models in high
  dimensions}.
\newblock {\em Bayesian Analysis}, 17(3):901--929.

\bibitem[Martin et~al., 2017]{martin2017empirical}
Martin, R., Mess, R., and Walker, S.~G. (2017).
\newblock {Empirical Bayes posterior concentration in sparse high-dimensional
  linear models}.
\newblock {\em Bernoulli}, 23(3):1822--1847.

\bibitem[Matteson and James, 2014]{matteson2014nonparametric}
Matteson, D.~S. and James, N.~A. (2014).
\newblock A nonparametric approach for multiple change point analysis of
  multivariate data.
\newblock {\em Journal of the American Statistical Association},
  109(505):334--345.

\bibitem[Meier et~al., 2021]{meier2021mosum}
Meier, A., Kirch, C., and Cho, H. (2021).
\newblock mosum: A package for moving sums in change-point analysis.
\newblock {\em Journal of Statistical Software}, 97:1--42.

\bibitem[Narisetty and He, 2014]{narisetty2014bayesian}
Narisetty, N.~N. and He, X. (2014).
\newblock Bayesian variable selection with shrinking and diffusing priors1.
\newblock {\em The Annals of Statistics}, 42(2):789--817.

\bibitem[Natarajan et~al., 2023]{natarajan2023cohesion}
Natarajan, A., De~Iorio, M., Heinecke, A., Mayer, E., and Glenn, S. (2023).
\newblock Cohesion and repulsion in bayesian distance clustering.
\newblock {\em Journal of the American Statistical Association},
  (just-accepted):1--18.

\bibitem[Ohn and Kim, 2022]{ohn2022posterior}
Ohn, I. and Kim, Y. (2022).
\newblock Posterior consistency of factor dimensionality in high-dimensional
  sparse factor models.
\newblock {\em Bayesian Analysis}, 17(2):491--514.

\bibitem[Pati et~al., 2014]{pati2014posterior}
Pati, D., Bhattacharya, A., Pillai, N.~S., and Dunson, D. (2014).
\newblock {Posterior contraction in sparse Bayesian factor models for massive
  covariance matrices}.
\newblock {\em The Annals of Statistics}, pages 1102--1130.

\bibitem[Pein et~al., 2017]{pein2017heterogeneous}
Pein, F., Sieling, H., and Munk, A. (2017).
\newblock Heterogeneous change point inference.
\newblock {\em Journal of the Royal Statistical Society. Series B (Statistical
  Methodology)}, 79(4):1207--1227.

\bibitem[Pronzato and P{\'a}zman, 2013]{pronzato2013design}
Pronzato, L. and P{\'a}zman, A. (2013).
\newblock {\em Design of experiments in nonlinear models}.
\newblock Springer.

\bibitem[Pukelsheim, 1994]{pukelsheim1994three}
Pukelsheim, F. (1994).
\newblock The three sigma rule.
\newblock {\em The American Statistician}, 48(2):88--91.

\bibitem[Ray and Szab{\'o}, 2022]{ray2022variational}
Ray, K. and Szab{\'o}, B. (2022).
\newblock {Variational Bayes for high-dimensional linear regression with sparse
  priors}.
\newblock {\em Journal of the American Statistical Association},
  117(539):1270--1281.

\bibitem[Ro{\v{c}}kov{\'a}, 2018]{rovckova2018bayesian}
Ro{\v{c}}kov{\'a}, V. (2018).
\newblock Bayesian estimation of sparse signals with a continuous
  spike-and-slab prior.
\newblock {\em The Annals of Statistics}, 46(1):401--437.

\bibitem[Ro{\v{c}}kov{\'a} and George, 2016]{rovckova2016fast}
Ro{\v{c}}kov{\'a}, V. and George, E.~I. (2016).
\newblock {Fast Bayesian factor analysis via automatic rotations to sparsity}.
\newblock {\em Journal of the American Statistical Association},
  111(516):1608--1622.

\bibitem[Romano et~al., 2022]{romano2022detecting}
Romano, G., Rigaill, G., Runge, V., and Fearnhead, P. (2022).
\newblock Detecting abrupt changes in the presence of local fluctuations and
  autocorrelated noise.
\newblock {\em Journal of the American Statistical Association},
  117(540):2147--2162.

\bibitem[Scheipl et~al., 2012]{scheipl2012spike}
Scheipl, F., Fahrmeir, L., and Kneib, T. (2012).
\newblock Spike-and-slab priors for function selection in structured additive
  regression models.
\newblock {\em Journal of the American Statistical Association},
  107(500):1518--1532.

\bibitem[Shin and Liu, 2021]{shin2021neuronized}
Shin, M. and Liu, J.~S. (2021).
\newblock {Neuronized priors for Bayesian sparse linear regression}.
\newblock {\em Journal of the American Statistical Association}, pages 1--16.

\bibitem[Stransky et~al., 2006]{stransky2006regional}
Stransky, N., Vallot, C., Reyal, F., Bernard-Pierrot, I., De~Medina, S. G.~D.,
  Segraves, R., De~Rycke, Y., Elvin, P., Cassidy, A., Spraggon, C., et~al.
  (2006).
\newblock Regional copy number--independent deregulation of transcription in
  cancer.
\newblock {\em Nature genetics}, 38(12):1386--1396.

\bibitem[Teh et~al., 2007]{teh2007stick}
Teh, Y.~W., Gr{\"u}r, D., and Ghahramani, Z. (2007).
\newblock {Stick-breaking construction for the Indian buffet process}.
\newblock In {\em Artificial intelligence and statistics}, pages 556--563.
  PMLR.

\bibitem[Tibshirani et~al., 2005]{tibshirani2005sparsity}
Tibshirani, R., Saunders, M., Rosset, S., Zhu, J., and Knight, K. (2005).
\newblock {Sparsity and smoothness via the fused lasso}.
\newblock {\em Journal of the Royal Statistical Society: Series B (Statistical
  Methodology)}, 67(1):91--108.

\bibitem[Vostrikova, 1981]{vostrikova1981detecting}
Vostrikova, L.~Y. (1981).
\newblock Detecting “disorder” in multidimensional random processes.
\newblock 259(2):270--274.

\bibitem[Williamson et~al., 2010]{williamson2010ibp}
Williamson, S., Wang, C., Heller, K.~A., and Blei, D.~M. (2010).
\newblock {The IBP compound Dirichlet process and its application to focused
  topic modeling}.
\newblock In {\em ICML}.

\bibitem[Wyse et~al., 2011]{wyse2011approximate}
Wyse, J., Friel, N., Rue, H., et~al. (2011).
\newblock Approximate simulation-free bayesian inference for multiple
  changepoint models with dependence within segments.
\newblock {\em Bayesian Analysis}, 6(4):501--528.

\bibitem[Yang et~al., 2016]{yang2016computational}
Yang, Y., Wainwright, M.~J., and Jordan, M.~I. (2016).
\newblock {On the computational complexity of high-dimensional Bayesian
  variable selection}.
\newblock {\em The Annals of Statistics}, pages 2497--2532.

\bibitem[Yao, 1984]{yao1984estimation}
Yao, Y.-C. (1984).
\newblock {Estimation of a noisy discrete-time step function: Bayes and
  empirical Bayes approaches}.
\newblock {\em The Annals of Statistics}, pages 1434--1447.

\bibitem[Yen, 2011]{yen2011majorization}
Yen, T.-J. (2011).
\newblock A majorization--minimization approach to variable selection using
  spike and slab priors.
\newblock {\em The Annals of Statistics}, 39(3):1748--1775.

\bibitem[Zeileis et~al., 2002]{zeileis2002strucchange}
Zeileis, A., Leisch, F., Hornik, K., and Kleiber, C. (2002).
\newblock {strucchange: An R package for testing for structural change in
  linear regression models}.
\newblock {\em Journal of statistical software}, 7:1--38.

\end{thebibliography}

\end{document}